\documentclass[11pt]{article}

\usepackage{mathtools}
\usepackage[margin=1in]{geometry}
\usepackage{amsmath,amssymb,amsfonts,mathrsfs}
\usepackage{bm}
\usepackage{algpseudocode}
\usepackage{enumitem}
\usepackage{amsthm}
\usepackage{subfigure}
\usepackage{float}
\usepackage{stmaryrd}
\usepackage{fancybox}
\usepackage{bigstrut,array,multirow}
\usepackage{here}
\usepackage{xcolor}
\usepackage{graphicx}
\usepackage{tabularx}
\usepackage{tikz}
\usetikzlibrary{arrows.meta}
\usepackage[normalem]{ulem}
\usepackage[backend=biber,style=alphabetic,sorting=anyt,maxcitenames=3,mincitenames=1,maxbibnames=99]{biblatex}
\usepackage{hyperref}
\hypersetup{
    colorlinks=true,
    linkcolor=blue,
    citecolor=blue,
    urlcolor=blue
}
\usepackage{cleveref}
\usepackage{thmtools}
\usepackage{thm-restate}

\graphicspath{{figures/}}

\newtheorem{theorem}{Theorem}
\newtheorem{corollary}[theorem]{Corollary}

\newtheorem{claim}[theorem]{Claim}
\newtheorem{observation}[theorem]{Observation}
\newtheorem{lemma}[theorem]{Lemma}
\newtheorem{definition}[theorem]{Definition}

\newcommand{\dist}{\textnormal{dist}}

\newcommand{\pset}{{\mathcal{P}}}

\newcommand{\mset}{{\mathcal M}}

\newcommand{\cro}{{\mathsf{cr}}}
\newcommand{\pot}{{\mathsf{pot}}}
\newcommand{\eps}{{\varepsilon}}

\newcommand{\set}[1]{\left\{ #1 \right\}}

\newcommand{\myparskip}{3pt}
\parskip \myparskip
\begin{document}

\begin{titlepage}
\title{Paths and Intersections: Repelling Pairs}
\author{Yu Chen\thanks{National University of Singapore, Singapore. Email: {\tt yu.chen@nus.edu.sg}.}
\and Pavlo Pylyavskyy\thanks{University of Minnesota Twin Cities, MN, USA. Email: {\tt ppylyavs@umn.edu}. This work was supported by a grant from Simons Foundation International SFI-MPS-SFM-00011393.}
\and Zihan Tan\thanks{University of Minnesota Twin Cities, MN, USA. Email: {\tt ztan@umn.edu}.}}

\maketitle
\thispagestyle{empty}
\begin{abstract}
We study two inverse problems for shortest-path metrics of Okamura-Seymour instances: recognizing metrics realizable by outerplanar graphs, and computing
minimum-edge Okamura-Seymour realizations. We introduce the notion of \emph{repelling
pairs}, a metric certificate that the shortest paths corresponding to two
terminal pairs must be vertex-disjoint in every realization. Our central
structural result is that, for an Okamura-Seymour metric with a prescribed cyclic order, a terminal path structure in an Okamura-Seymour instance can be
realized by nonnegative edge lengths if and only if the paths assigned to
every repelling pair are vertex-disjoint. 

Building on the notion of repelling pairs, we give algorithmic answers to the inverse problems. First, we design an algorithm that, given a metric, decides in polynomial time whether or it admits an outerplanar
realization and constructs one when one exists. 
Second, given an Okamura-Seymour metric,
we efficiently compute a canonical medial template whose crossing number equals the minimum number of edges in any Okamura-Seymour realization. The minimum-edge graph structures are exactly the primal graphs of arrangements of this template, and each can be assigned realizing edge lengths in polynomial time.

\end{abstract}

\end{titlepage}

\tableofcontents
\newpage

\section{Introduction}\label{sec:introduction}

A fundamental problem in metric graph theory is to reconstruct a graph
from the observed shortest-path distances on a designated set of vertices called \emph{terminals}.  Formally, given a
metric $D$ on a terminal set $T$ and a family $\mathcal G$ of graphs, the
\emph{distance realization problem} asks whether there is a nonnegatively
edge-weighted graph $G\in\mathcal G$, possibly containing additional Steiner
vertices, whose shortest-path metric on $T$ is exactly $D$.  Distance
realization goes back to Hakimi and Yau~\cite{HakimiYau65}, and it has applications in 
phylogenetics, chemistry, hierarchical classification, network
tomography~\cite{dress2012basic,janezic2015graph,gordon1987review,chung2001distance}, and so on.
Classical graph families studied in the distance realization problem include trees, ultrametrics, and cactus graphs~\cite{pereira1969note,buneman1974note,dress1984trees,gordon1987review,hayamizu2020recognizing}.

In this paper, we study shortest-path metrics of \emph{Okamura-Seymour instances}.
An Okamura-Seymour instance (or \emph{an OS instance}) is a planar graph embedded
in a disk with all terminals on the boundary.  This graph class originates in
the study of multicommodity flows in planar graphs~\cite{okamura1981multicommodity}.
If the cyclic order of the terminals is prescribed, distance realizability by an OS
instance is characterizable by an elegant four-point condition: for all quadruples
$a,b,c,d$ appearing on the boundary in this cyclic order,
\begin{equation}\label{eq:kalmanson}
  D(a,c)+D(b,d)
  \ge
  \max\bigl\{
    D(a,b)+D(c,d),
    D(a,d)+D(b,c)
  \bigr\}.
\end{equation}
Conversely, every metric satisfying~\eqref{eq:kalmanson} with respect to the prescribed
cyclic order has an OS realization
\cite{hurkens1988tidy,ChangO20,schrijver2003combinatorial}.  
We also call it a \emph{Kalmanson metric}; the terminology comes from
Kalmanson's work on edge-convex tours~\cite{kalmanson1975}, and these metrics
also arise naturally in split-decomposition and circular-network theory
\cite{bandelt1992,forcey2023}.

The characterization above makes OS metrics a natural starting point for
further questions in distance realization. A natural next graph
family to study is outerplanar graphs: they include trees
and cactus graphs but are a subclass of OS instances.
A second natural direction is to remain within the class of OS instances
and ask for the simplest realization of a given metric. Specifically, we ask for a realization
with the minimum possible number of edges, and more generally which features
of an edge-minimum realization are determined by the terminal distances. 

Although outerplanar recognition and minimum OS realization ask different
questions, they are connected by a common structural viewpoint. Both our
proofs and our algorithms use a recent approach to analyzing graph structure
by viewing a graph through a family of paths and the pattern of their
intersections~\cite{chen2025path,li2025paths}. In both settings, the given metric constrains which shortest paths may intersect, and the same elementary obstruction
involving two terminal pairs becomes the key organizing principle.

\subsection{Repelling pairs}\label{subsec:intro-repelling}

Suppose that an $s$-$t$ path and a $u$-$v$ path intersect.  Switching the two
paths at an intersection produces paths connecting their four endpoints
according to the two alternative pairings.  Consequently, if the original
paths are shortest, then the following inequalities must hold.
\[
  \dist(s,t)+\dist(u,v)\ge \dist(s,u)+\dist(t,v),
\and\text{ }
  \dist(s,t)+\dist(u,v)\ge \dist(s,v)+\dist(t,u).
\]
This motivates our central definition. Given a metric $D$, two terminal pairs $(s,t)$ and
$(u,v)$ \emph{repel} if
\[
  D(s,t)+D(u,v)
  <
  \max\bigl\{
    D(s,u)+D(t,v),
    D(s,v)+D(t,u)
  \bigr\}.
\]
Thus, repelling pairs provide metric certificates of forced path disjointness:
shortest paths connecting repelling pairs cannot intersect in any graph realizing
$D$\footnote{Here we allow repeated terminal pairs such as $(t,t)$; the corresponding
path is the trivial path at $t$.}.

We define a \emph{shortest path structure} in an OS instance to be a set $\pset$ that contains, for every pair $s,t$ of terminals, a path $P_{s,t}$, such that the
intersection of any two paths in $\pset$ is either empty or a common subpath of
both.  We call the shortest path structure \emph{good with respect to $D$}, or simply \emph{$D$-good}, if
\[
  P_{s,t}\cap P_{u,v}=\varnothing
  \qquad
  \text{whenever $(s,t)$ and $(u,v)$ repel.}
\]
Our main structural theorem says that these obvious necessary conditions are
also sufficient.

\begin{theorem}
\label{lem: 2-path condition}
For a Kalmanson metric $D$ on $T$, an Okamura-Seymour instance $(G,T)$ and a shortest path structure $\pset=\set{P_{t,t'}}_{t,t'\in T}$, the following are equivalent:
\begin{itemize}
    \item There exist nonnegative edge lengths $\set{\ell_e}_{e\in E(G)}$, such that
    for every pair $t,t'$ of terminals,
    \label{prop_2}
    \begin{itemize}
        \item $P_{t,t'}$ is indeed the shortest $t$-$t'$ path in $G$; and
        \item the length of $P_{t,t'}$ equals $D(t,t')$.
    \end{itemize}
    \item For every repelling pair $(t_1,t_2)$ and $(t'_1,t'_2)$, the paths $P_{t_1,t_2}$ and $P_{t'_1,t'_2}$ are disjoint.
\end{itemize}
\end{theorem}

That is, an OS graph realizes $D$ if and only if it contains a $D$-good
shortest path structure.  Theorem~\ref{lem: 2-path condition} cleanly separates the
realization problem into a combinatorial layer and a numerical layer.  To
construct a realization, it is now enough to construct a graph together with a
shortest path structure satisfying the required intersection and disjointness
relations.

\subsection{Recognizing outerplanar metrics}\label{subsec:intro-outerplanar}

Tree metrics and Kalmanson metrics admit elegant four-point characterizations. It is therefore natural to ask whether outerplanar realizability also admits an $O(1)$-point characterization. Our first result rules this out in the strongest possible form.

\begin{theorem}\label{thm: no O(1) point}
For every integer $k\ge 20$, there exists a metric $D$ on $k$ terminals such that $D$ is not outerplanar, while the restriction of $D$ to every proper subset of terminals is outerplanar.
\end{theorem}

Thus, no bounded-size collection of induced submetrics can characterize outerplanar realizability. This contrasts with the classical four-point characterizations for trees and fixed-order OS metrics. Despite this negative result, the recognition problem remains tractable.

\begin{theorem}\label{thm: main}
There is an algorithm that, given a metric $D$ on $k$ terminals, decides whether there is an outerplanar graph realizing $D$, and constructs one if the answer is yes. Its running time is $O(k^5)$.
\end{theorem}

Efficient recognition algorithms were previously known for trees, cactus graphs, and Kalmanson metrics. Theorem~\ref{thm: main} extends this to outerplanar graphs, a class that contains trees and cactus graphs and is contained in the OS class. Section~\ref{sec:outerplanar-proof} gives the proof.

\subsection{Minimum OS realizations}\label{subsec:intro-minimum}

The natural next question after distance realization is minimum realization~\cite{feder2003representing}: among all graphs in a prescribed family realizing a given metric, find those with the fewest edges. Whereas the distance-realization problem concerns the structural property of a metric, the minimum-realization problem concerns its structural complexity and seeks a succinct graph representation. This objective differs from the classical optimal-realization problem, which minimizes total edge length~\cite{Althoefer88,Winkler88}.

For general graph metrics, Feder et al.~\cite{feder2003representing} gave approximation algorithms for several restricted settings and strong hardness results for the general unit-weight problem. We give an exact polynomial-time solution for OS metrics with a prescribed cyclic order.

\setcounter{theorem}{3}
\begin{theorem}\label{thm: minimum main}
Given any OS metric $D$, the underlying graphs of all its minimum OS realizations can be efficiently recovered. Moreover, for each such graph structure, one can efficiently compute nonnegative edge lengths realizing $D$.
\end{theorem}

More specifically, we compute in polynomial time a template for the input OS metric $D$, together with the minimum edge count. This template compactly represents all minimum graph structures; any one such graph and realizing edge lengths can be computed in polynomial time. The graph structures are determined by the metric at the level of this template, although realizing edge lengths need not be unique. Section~\ref{sec:minimum-proof} gives the proof and the detailed construction.

\subsection{Related work, and relations to electrical networks}\label{subsec:related}

Our minimum-realization result has a close formal parallel with the inverse problem for circular planar electrical networks. In that setting, one starts with a finite graph embedded in a disk, assigns a positive conductance $c(e)$ to each edge, and records the linear map $\Lambda v=i$ from boundary voltages to boundary currents, where $\Lambda$ is the response matrix. Classical work shows that the response matrix determines the reduced medial structure of a critical network, determines its conductances, and relates critical representatives with the same response matrix by local $Y$-$\Delta$ transformations~\cite{de1996,curtis1998,kenyon2008,kenyon2011}. Table~\ref{table} compares this theory with the distance-realization results developed here.

Distance realization has been studied extensively since the work of Hakimi and Yau~\cite{HakimiYau65}. Early and classical work includes tree realization and recognition~\cite{patrinos1972distance,pereira1969note,simoes1982submatrices,simoes1987graph,bandelt1990recognition,varone1998trees,buneman1974note,dress1984trees}, and more recent work considers, among other variants, cactus metrics and distance-set realization~\cite{hayamizu2020recognizing,bar2022graph,bar2023composed}. Goldman~\cite{Goldman66} described the irreducible realization with as few positive-length arcs as possible when the prescribed points must constitute the full vertex set. This fixed-vertex model is related in spirit to minimum realization, but differs from our setting because we allow Steiner vertices.

\begin{table}[h]
\centering
\renewcommand{\arraystretch}{1.25}
\begin{tabular}{p{0.16\textwidth}p{0.4\textwidth}p{0.38\textwidth}}
\hline
\textbf{OS instances} & \textbf{Electrical networks} & \textbf{Distance realization (this paper)} \\
\hline
Boundary data & Response matrix $\Lambda$, mapping boundary voltages to boundary currents. & Terminal metric $D$, recording shortest-path distances between boundary terminals. \\
Realizability & Characterized by nonnegativity of circular minors of $\Lambda$~\cite{curtis1998}. & Characterized by the Monge property/Kalmanson condition~\cite{hurkens1988tidy,ChangO20}. \\
Reduced objects & Critical networks. & Minimum realizations, i.e., realizations with the fewest edges. \\
Medial object & The medial template of a critical network is determined by $\Lambda$~\cite{de1996,curtis1998}. & The medial template $\Phi(D)$ of a minimum realization is determined by $D$. \\
Local moves & The critical networks of a response matrix are related by $Y$--$\Delta$ transformations~\cite{de1996,curtis1998}. & Underlying embedded graphs of minimum realizations of an OS metric are related by $Y$--$\Delta$ transformations. \\
Numerical & For critical networks, the conductances are recoverable from $\Lambda$; in standard representatives they are given by explicit rational formulas~\cite{kenyon2008,kenyon2011}. & The metric determines the minimum medial template and the graph structures; realizing edge lengths can be computed efficiently, but in general are not unique; see Appendix~\ref{apd: example}. \\
\hline
\end{tabular}
\caption{Comparison between conductance and distance realizations for Okamura--Seymour instances.}\label{table}
\end{table}

Another major line of work concerns optimal realizations, which minimize total edge length rather than the number of edges. Computing an optimal realization of an integral metric is NP-hard~\cite{Althoefer88,Winkler88}. The minimum-edge objective for general graph metrics was studied by Feder et al.~\cite{feder2003representing}. For fixed-order undirected OS metrics, realizability is characterized by the Kalmanson inequalities~\cite{hurkens1988tidy,ChangO20}. Chen and Tan studied the directed OS realization problem and proved the corresponding directed Monge characterization~\cite{chen2025path}. More broadly, our path-based viewpoint belongs to a recent line of research that analyzes graph structure through prescribed shortest paths and their intersections~\cite{bodwin2019structure,akmal2022local,cizma2022geodesic,cizma2023irreducible,chen2025path,li2025paths}.

\subsection{Organization}

Section~\ref{sec:prelim} gives the preliminaries and introduces repelling pairs, shortest path structures, boundary cuts, and chains. Section~\ref{sec: repelling paths} proves Theorem~\ref{lem: 2-path condition}. Section~\ref{sec:outerplanar-proof} begins with a high-level overview and then proves the outerplanar recognition result. Section~\ref{sec:minimum-proof} likewise begins with a high-level overview and then proves the minimum-realization result. The counterexample to bounded-point characterization and several technical proofs are deferred to the appendices.

\section{Preliminaries, and Repelling Pairs}\label{sec:prelim}

\paragraph{Graphs and realizations.}
All graphs are finite and may have nonnegative edge lengths. For a weighted graph $G$, write $\dist_G(u,v)$ for the shortest-path distance between $u$ and $v$. Let $T$ be a designated set of vertices called \emph{terminals}. We say that $G$ \emph{realizes} a metric $D$ on $T$ if $\dist_G(s,t)=D(s,t)$ for all $s,t\in T$. Steiner vertices in $V(G)\setminus T$ are allowed.

\paragraph{Okamura-Seymour and outerplanar instances.}
An Okamura-Seymour instance~\cite{okamura1981multicommodity}, or OS instance, is a pair $(G,T)$ in which $G$ is embedded in a closed disk and all terminals lie on its boundary. Other graph vertices may also lie on the boundary. An outerplanar realization is an OS realization in which every graph vertex lies on the boundary. In the minimum-realization part, we may perturb nonterminal boundary vertices slightly into the interior, and there we assume without loss of generality that the only boundary vertices are terminals.

For a prescribed cyclic order of $T$, a metric $D$ is an OS metric, or Kalmanson metric, if whenever $t_1,t_2,t_3,t_4$ occur in this order,
\[
\max\{D(t_1,t_2)+D(t_3,t_4),D(t_1,t_4)+D(t_2,t_3)\}
\le D(t_1,t_3)+D(t_2,t_4).
\]
These inequalities characterize metrics realizable by OS instances with the given order~\cite{hurkens1988tidy,ChangO20,schrijver2003combinatorial}.

\paragraph{Repelling pairs.}
For terminal pairs $(t_1,t_2)$ and $(t'_1,t'_2)$, not necessarily with four distinct endpoints, we say that the two pairs \emph{repel} if
\[
D(t_1,t_2)+D(t'_1,t'_2)<
\max\{D(t_1,t'_1)+D(t_2,t'_2),D(t_1,t'_2)+D(t'_1,t_2)\}.
\]
A family of terminal pairs is \emph{repelling} if its members repel pairwise. We allow the repeated pair $(t,t)$, represented by the trivial path at $t$.

\begin{lemma}\label{lem: intersecting shortest paths}
Let $P$ be an $a$-$b$ shortest path and $Q$ a $c$-$d$ shortest path in a weighted graph. If $P$ and $Q$ intersect, then
\[
\dist_G(a,b)+\dist_G(c,d)\ge \dist_G(a,c)+\dist_G(b,d)
\]
and
\[
\dist_G(a,b)+\dist_G(c,d)\ge \dist_G(a,d)+\dist_G(b,c).
\]
\end{lemma}

\begin{proof}
Let $x$ be any common vertex of $P$ and $Q$. Pairing the $a$-$x$ and $x$-$c$ subpaths with the complementary $b$-$x$ and $x$-$d$ subpaths gives
\[
\dist_G(a,c)+\dist_G(b,d)\le \ell(P)+\ell(Q).
\]
Pairing the four subpaths in the other way similarly gives $\dist_G(a,d)+\dist_G(b,c)\le \ell(P)+\ell(Q)$. Since $\ell(P)=\dist_G(a,b)$ and $\ell(Q)=\dist_G(c,d)$, both inequalities follow.
\end{proof}

It follows immediately that paths realizing repelling terminal pairs must be vertex-disjoint. The following observation on repelling pairs will be used repeatedly.

\begin{observation}\label{lem: W-repel}
Let terminals $t_1,t_2,t_3,t_4,t_5$ appear on the boundary in this order. If $(t_1,t_5)$ does not repel $(t_2,t_4)$, and $(t_2,t_5)$ does not repel $(t_3,t_4)$, then $(t_1,t_5)$ does not repel $(t_3,t_4)$.
\end{observation}

\begin{proof}
Note that
\begin{align*}
&D(t_1,t_4)+D(t_3,t_5)-D(t_1,t_5)-D(t_3,t_4)\\
&=\bigl(D(t_1,t_4)+D(t_2,t_5)-D(t_1,t_5)-D(t_2,t_4)\bigr)\\
&\quad+\bigl(D(t_2,t_4)+D(t_3,t_5)-D(t_2,t_5)-D(t_3,t_4)\bigr)\\
&=0+0=0.
\end{align*}
Therefore, $(t_1,t_5)$ does not repel $(t_3,t_4)$.
\end{proof}

\paragraph{Shortest path structures.}
For an OS instance $(G,T)$, a \emph{shortest path structure} is a collection $\pset$ that contains, for every pair $t,t'$ of terminals in $T$, a path $P_{t,t'}$ from $t$ to $t'$ in $G$, such that the intersection between any two paths in $\pset$ is either empty or a subpath of both. We say that a shortest path structure $\pset$ is \emph{$D$-good}, or simply \emph{good} when $D$ is clear from the context, if for all repelling pairs $(t_1,t_2),(t'_1,t'_2)$, paths $P_{t_1,t_2}$ and $P_{t'_1,t'_2}$ are vertex-disjoint.

\paragraph{Boundary cuts and chains.}
For the minimum-realization part, let $X$ contain every terminal and one additional reference point in each boundary interval between consecutive terminals. We consider the disk boundary as part of the drawing of an OS instance. A face is \emph{peripheral} if its boundary contains a reference point in $X\setminus T$.

Two points $x,y\in X$ divide the disk boundary into two arcs and define the $x$-$y$ \emph{boundary cut}. A terminal pair crosses the cut if its endpoints lie on opposite arcs, or if an endpoint coincides with $x$ or $y$.

\begin{definition}\label{def:chains}
A \emph{chain} in an OS instance is a sequence $A=(A_1,A_2,\ldots,A_r)$ such that
\begin{itemize}
\item for all $i<r$, one of $A_i$ and $A_{i+1}$ is a vertex of $G$, the other is a region, and they are incident; and
\item $A_i$ is peripheral if and only if $i=1$ or $i=r$.
\end{itemize}
The length of $A$, denoted by $|A|$, is the number of terms $A_i$ that are vertices. Its boundary cut is the cut whose endpoints are the terminal endpoint when an end of the chain is a terminal vertex, and otherwise the reference point contained in the corresponding peripheral region. A terminal pair crosses $A$ if it crosses this boundary cut.
\end{definition}

\section{Proof of Theorem 1}
\label{sec: repelling paths}

The construction of an OS instance realizing the input metric $D$ consists of two steps:
\begin{enumerate}
    \item constructing the graph and the shortest path structure, which requires specifying
    \label{prop_1}
    \begin{itemize}
        \item the graph $G$;
        \item for each pair $t,t'$ of terminals, a $t$--$t'$ path $P_{t,t'}$ in $G$, designated as their shortest path, such that the intersection between every pair of shortest paths is empty or a subpath of both;
    \end{itemize}
    \item setting the edge weights of $G$, such that for every pair $t,t'$ of terminals,
    \label{prop_2}
    \begin{itemize}
        \item $P_{t,t'}$ is indeed a shortest $t$--$t'$ path in $G$; and
        \item the length of $P_{t,t'}$ equals $D(t,t')$.
    \end{itemize}
\end{enumerate}

Now imagine that we are given the graph $G$ and the shortest path structure $\pset=\set{P_{t,t'}}_{t,t'\in T}$ produced by Step 1. Theorem~\ref{lem: 2-path condition} states that Step 2 can be accomplished if and only if the repelling paths condition holds for every pair of terminal pairs.

From \Cref{lem: intersecting shortest paths} and the definition of repelling pairs, the forward direction follows immediately. We now focus on the backward direction. Assume that we are given $G,\pset$ such that for every repelling pair $(t_1,t_2)$ and $(t'_1,t'_2)$, the paths $P_{t_1,t_2}$ and $P_{t'_1,t'_2}$ are disjoint. We will show that there exist edge lengths realizing $\pset$ and $D$.

The existence of such edge lengths is actually characterized by the feasibility of the following LP. The variables are $\set{\ell_e}_{e\in E(G)}$. There is no objective function, as only its feasibility matters.
\begin{eqnarray*}
&\sum_{e\in E(P_{t,t'})}\ell_e\leq D(t,t') &\forall t,t', \\
&\sum_{e\in E(Q)}\ell_e\geq D(t,t') &\forall t,t', \forall \text{ simple path }Q\text{  from }t \text{ to  }t'\\
	&\ell_e\geq 0&\forall e\in E(G)
\end{eqnarray*}
By Farkas' lemma, either the above LP is feasible (and therefore there exist edge lengths for Theorem~\ref{lem: 2-path condition}), or
there exist flows $F,F'$ in $G$, such that:
    \begin{enumerate}
        \item $F$ only uses paths in $\pset$, while $F'$ may use any paths; 
        \label{flowprop_1}
        \item for each edge $e$, the amount of flow in $F$ sent through $e$ is greater than the amount of flow in $F'$ sent through $e$; and \label{flowprop_2}
        \item if we denote, for each pair $t,t'$, by $F_{t,t'}$ the amount of $F$ sent between $t,t'$, then
        \[\sum_{t,t'\in T}F_{t,t'}\cdot D(t,t')<\sum_{t,t'\in T}F'_{t,t'}\cdot D(t,t').\] \label{flowprop_3}
    \end{enumerate}

Using standard techniques \cite{chen2025path}, we can convert any pair $F,F'$ of flows satisfying the above properties to a pair of flows with the following additional property:
    \begin{enumerate}[resume]
    \item for each terminal $t$, the total amount of flow in $F$ sent from (received by, resp.) $t$ is identical to the total amount of flow in $F'$ sent from (received by, resp.) $t$. \label{prop: aligned}
    \end{enumerate}

Therefore, to complete the proof of Theorem~\ref{lem: 2-path condition}, it suffices to show that, in an Okamura-Seymour instance $(G,T)$ with designated paths $\pset=\set{P_{t,t'}}_{t,t'\in T}$ such that $P_{t_1,t_2}$ and $P_{t'_1,t'_2}$ are disjoint for every repelling pair $(t_1,t_2),(t'_1,t'_2)$, there do not exist flows $F,F'$ satisfying Properties \ref{flowprop_1}, \ref{flowprop_2}, \ref{flowprop_3} and \ref{prop: aligned}.
Before we give the proof, we make an additional assumption: as $\mathbb{Q}$ is dense in $\mathbb{R}$, we assume that $F,F'$ are integral (they send integer units of flow through every path they use), so both $F,F'$ are  collections of paths.

\subsection{Preparation: chords, uncrossing, and path switching}

\paragraph{Chord configurations.}
$n$ distinct points labeled by $1,2,\ldots,n$ ($n$ is even) lie on the boundary of a disc clockwise in this order. A \emph{chord} $(a,b)$ where $a,b\in [n]$ is a straight line segment connecting $a,b$. Such a chord is also viewed as a matched pair $(a,b)$.
A \emph{configuration} $C$ is defined to be a set of chords.

For each $1\le i\le n$, there is a point lying on the boundary segment between points $i$ and $i+1$ that we call $\bar i$ (the point between $n$ and $1$ is called $\bar n$). For each pair $\bar i, \bar j$, we denote by $\dist_C(\bar i, \bar j)$ the number of chords in $C$ that the straight line segment $(\bar i, \bar j)$ intersects. It is immediate to verify that $\dist_C(\cdot , \cdot)$ is a metric on $\set{\bar 1,\bar 2,\ldots,\bar n}$. We call it the \emph{metric induced by $C$}. The line segment $(\bar i, \bar j)$ is also viewed as a cut, which partition the set of points into subsets $\set{i+1,\ldots,j}$ and $\set{j+1,\ldots,i}$.

\paragraph{Uncrossing and cut-metric dominance.}
For a pair of chords $(x,y),(z,w)$ that cross each other (that is, their distinct endpoints $x,z,y,w$ appear on the boundary in this order), we can \emph{uncross} them to obtain chords $(x,z),(y,w)$ or chords $(x,w),(y,z)$.
If a configuration $C$ contains $(x,y)$ and $(z,w)$, and we define $C^*=\big(C\setminus \set{(x,y),(z,w)}\big)\cup \set{(x,z),(y,w)}$, then we say that $C^*$ can be obtained from $C$ by uncrossing $(x,y),(z,w)$.
For configurations $C,C'$, we say that we can \emph{uncross $C$ to $C'$}, denoted by $C\to C'$, if there exists a sequence of configurations $C_0,\ldots,C_t$, such that $C_0=C$, $C_t=C'$; and for each $0\le i<t$, $C_{i+1}$ can be obtained from $C_i$ by uncrossing one pair of its chords.

We say that a configuration $C$ \emph{dominates} configuration $C'$, denoted by $C\succeq C'$, iff for every pair $\bar i,\bar j$, $\dist_C(\bar i,\bar j)\ge \dist_{C'}(\bar i,\bar j)$. We write $C\succ C'$ if additionally there exists some pair $\bar i,\bar j$ with $\dist_C(\bar i,\bar j) > \dist_{C'}(\bar i,\bar j)$. 
It is easy to see that, if $C\to C'$, then $C\succeq C'$. 

\paragraph{Path switching.}
Let $a,b,c,d$ be distinct terminals.
Let $P$ be an $a$-$b$ path, viewed as being directed from $a$ to $b$. Let $Q$ be a $d$-$c$ path in $G$, viewed as being directed from $d$ to $c$. Let $R=P\cap Q$.
So $P$ is the concatenation of: the subpath from $a$ to the $a$-side endpoint of $R$, which we denote by $P_a$; $R$; and the subpath from $b$ to the $b$-side endpoint of $R$, which we denote by $P_b$. We define subpaths $Q_c, Q_d$ of $Q$ similarly.
Assume that $P_a$ and $Q_d$ share an endpoint, and $P_b$ and $Q_c$ share an endpoint.
Define $X$ as the concatenation of paths $P_a$, $R$, and $Q_c$, so $X$ is an $a$-$c$ path.
Define $Y$ as the concatenation of paths $Q_d$, $R$, and $P_b$, so $Y$ is a $d$-$b$ path.
We say that $X, Y$ are obtained by \emph{switching paths $P,Q$}. See \Cref{fig: switch} for an illustration.

\begin{figure}[h]
\centering
\subfigure
{\scalebox{0.09}{\includegraphics{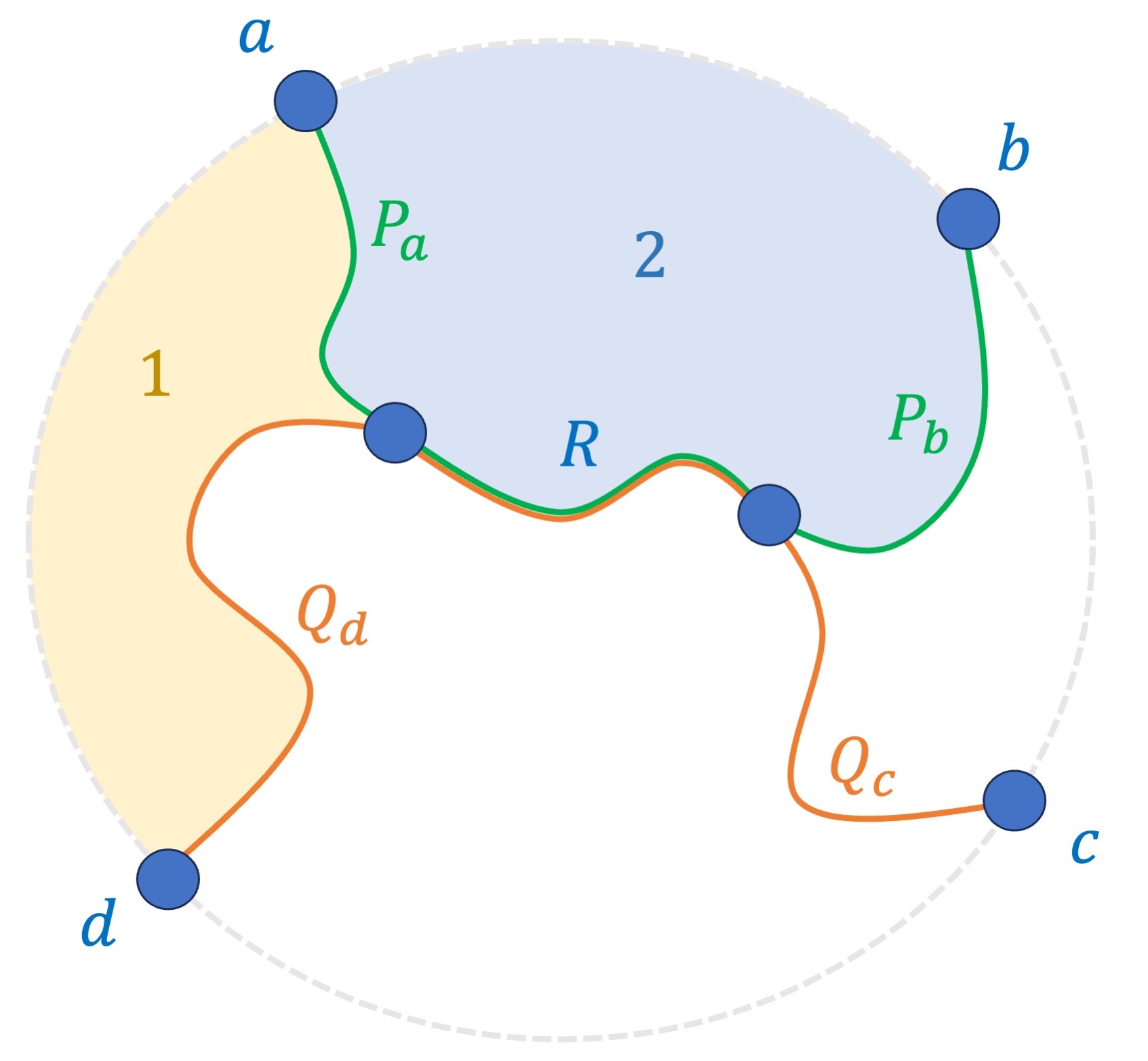}}}
\hspace{2.0cm}
\subfigure
{\scalebox{0.09}{\includegraphics{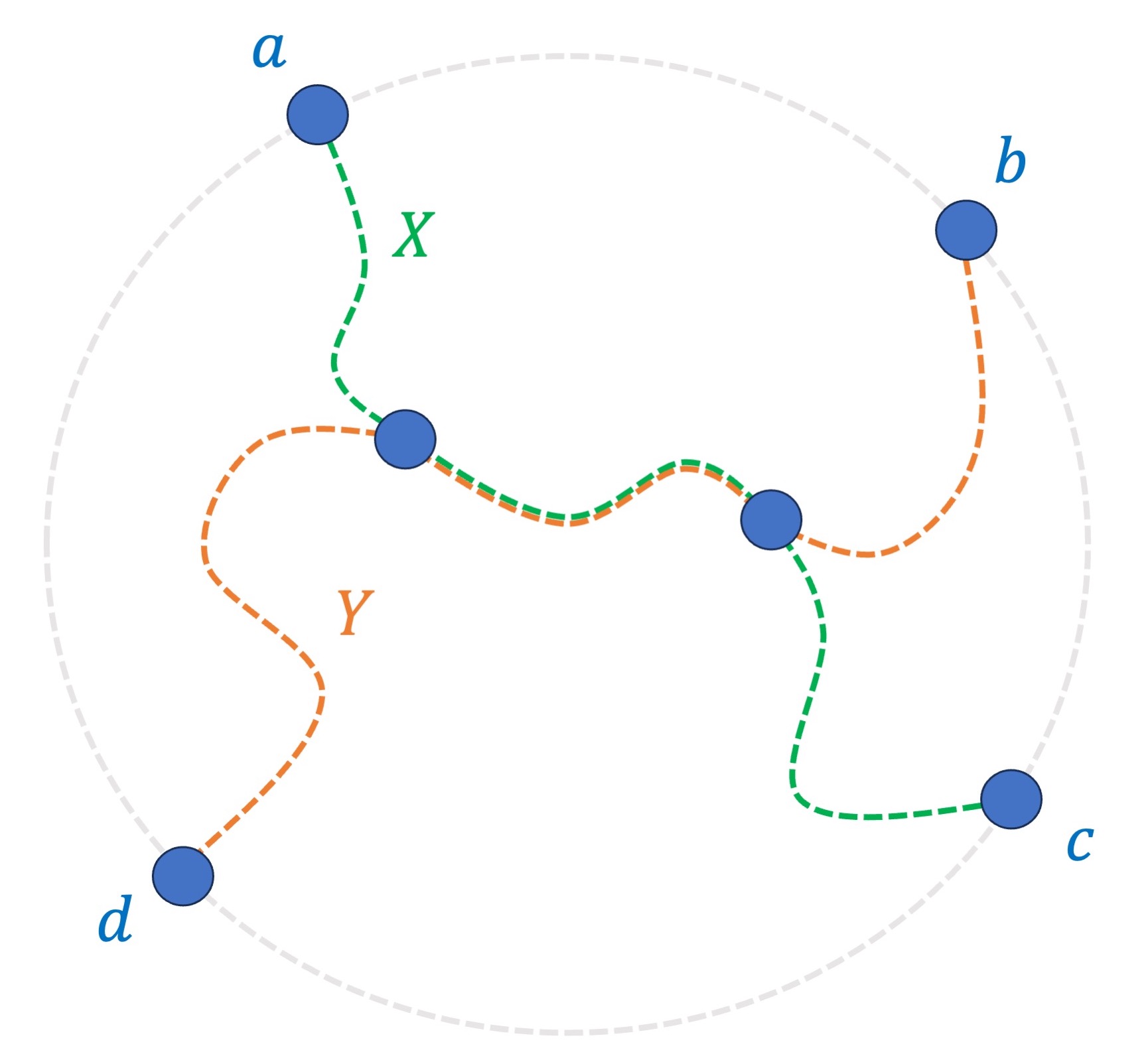}}}
\caption{An illustration of path switching.}\label{fig: switch}
\end{figure}

Given the Okamura-Seymour instance $(G,T)$ and a target metric $D$, 
consider a shortest path structure $\pset$ in it that satisfies all repelling paths conditions.
We first prove the following claim on $\pset$.

\begin{claim} \label{clm:not-repel-switch}
Let $a,b,c,d$ be distinct terminals appearing on the boundary in this order. Assume that the $a$-$b$ shortest path $P_{a,b}$ intersects the $d$-$c$ shortest path $P_{d,c}$, and let $X,Y$ be paths obtained by switching $P_{a,b}$ and $P_{d,c}$ (where $X$ connects $a$ to $c$ and $Y$ connects $d$ to $b$), then $X$ is the designated $a$-$c$ shortest path in $\pset$, namely $X=P_{a,c}$, and similarly $Y=P_{d,b}$.
\end{claim}
\begin{proof}
We will show that $P_{a,c}$ has no choice but to follow the trajectory of $X$. The proof of $Y=P_{d,b}$ is symmetric.
Let $R=P\cap Q$.
In \Cref{fig: switch}, if the first segment of $P_{a,c}$ does not follow $P_a$, then it must enter the interior of either region $1$ or region $2$. If it enters region $1$, it can only exit the region either from $P_a$, in which case  $P_{a,c}$ intersects twice with  $P_{a,b}$, a contradiction, or from $Q_d$, in which case  $P_{a,c}$ intersects twice with  $P_{d,c}$, a contradiction. So the first segment of $P_{a,c}$ must follow  $P_a$, and similarly the last segment of $P_{a,c}$ must follow path $Q_c$.
To avoid intersecting with $P_{a,b}$ or $P_{d,c}$ twice, between its first and last segments it must follow $R$ exactly. Altogether, $P_{a,c}=X$, and similarly $Y=P_{d,b}$.
\end{proof}

\subsection{Step 1. Maximally switching $F$}

In the first step, our goal is to ensure the following additional property of $F$:
\begin{enumerate}[resume]
\item for every tuple $s,s',t',t$ of terminals appearing on the boundary in this order with $F_{s,s'},F_{t,t'}>0$, their shortest paths $P_{s,s'}$ and $P_{t,t'}$ in $\pset$ are disjoint.  \label{prop: disjoint}
\end{enumerate}

Specifically, we show that if there exist flows $F,F'$ satisfying Properties \ref{flowprop_1}, \ref{flowprop_2}, \ref{flowprop_3} and \ref{prop: aligned}, then there exist flows $F,F'$ satisfying Properties \ref{flowprop_1}, \ref{flowprop_2}, \ref{flowprop_3}, \ref{prop: aligned}, and \ref{prop: disjoint}.

Start with the flow $F$ that we are given. While there exists a tuple $s,s',t',t$ where the above property does not hold, 
we take one unit of $s$-$s'$ flow and one unit of $t$-$t'$ flow and consider the corresponding $s$-$s'$ path $P_{s,s'}$ and $t$-$t'$ path $P_{t,t'}$ (recall that $F$ is viewed as a set of paths). We switch them and obtain an $s$-$t'$ path and a $t$-$s'$ path, which,
from \Cref{clm:not-repel-switch}, must be the designated paths $P_{s,t'},P_{t,s'}$, respectively. We replace $P_{s,s'}$ and $P_{t,t'}$ in $F$ with $P_{s,t'}$ and $P_{t,s'}$. We repeat until there are no such tuples.

We first show that after each iteration of path switching, Properties \ref{flowprop_1}, \ref{flowprop_2}, \ref{flowprop_3}, and \ref{prop: aligned} still hold for $F,F'$: \ref{flowprop_1} holds as $F$ still only uses the designated paths in $\pset$; \ref{flowprop_2} holds by definition of path switching; since paths $P_{s,s'}$ and $P_{t,t'}$ are not disjoint, the pairs $(s,s')$ and $(t,t')$ do not repel each other, which implies that $D(s,t')+D(t,s')=D(s,s')+D(t,t')$, which means that \ref{flowprop_3} also holds; and \ref{prop: aligned} holds trivially. 

We now proceed to show that the algorithm will terminate. Define a potential function $\pot(F)$ as
\[
\pot(F) =   \sum_{(s,s'), (t,t'): \text{ do not cross}} F(s,s') \cdot F(t,t').
\]
Before each iteration, $(s,s')$ and $(t,t')$ do not cross each other, and after the iteration, $(s,t')$ and $(s',t)$ cross each other. 
Moreover, it is easy to verify that for every other pair of terminals, the number of pairs in $\set{(s,s'),(t,t')}$ that it crosses is at most that of $\set{(s,t'),(t,s')}$. Therefore, after each iteration, the potential decreases by at least $1$, and so the algorithm will terminate.


\subsection{Step 2. $F$ vs $F'$: flow and cut-metric dominance}

We now show that there do not exist flows with Properties \ref{flowprop_1}, \ref{flowprop_2}, \ref{flowprop_3}, \ref{prop: aligned}, and \ref{prop: disjoint}.
Consider a pair $F,F'$ of flows satisfying Properties
\ref{flowprop_1}, \ref{flowprop_2}, \ref{prop: aligned} and \ref{prop: disjoint}. We now show that they may not satisfy \ref{flowprop_3}:      
\[\sum_{t,t'\in T}F_{t,t'}\cdot D(t,t')\ge \sum_{t,t'\in T}F'_{t,t'}\cdot D(t,t').\]

Recall that $F$ is viewed as a collection of paths.
We define a configuration $C$ as follows: for each unit of flow sent from $t$ to $t'$ in $F$ (corresponding to a $t$-$t'$ path in $F$), we add a chord $(t,t')$ to $C$. We define a configuration $C'$ in a similar way for $F'$. 

\begin{figure}[h]
\centering
\includegraphics[width=0.4\linewidth]{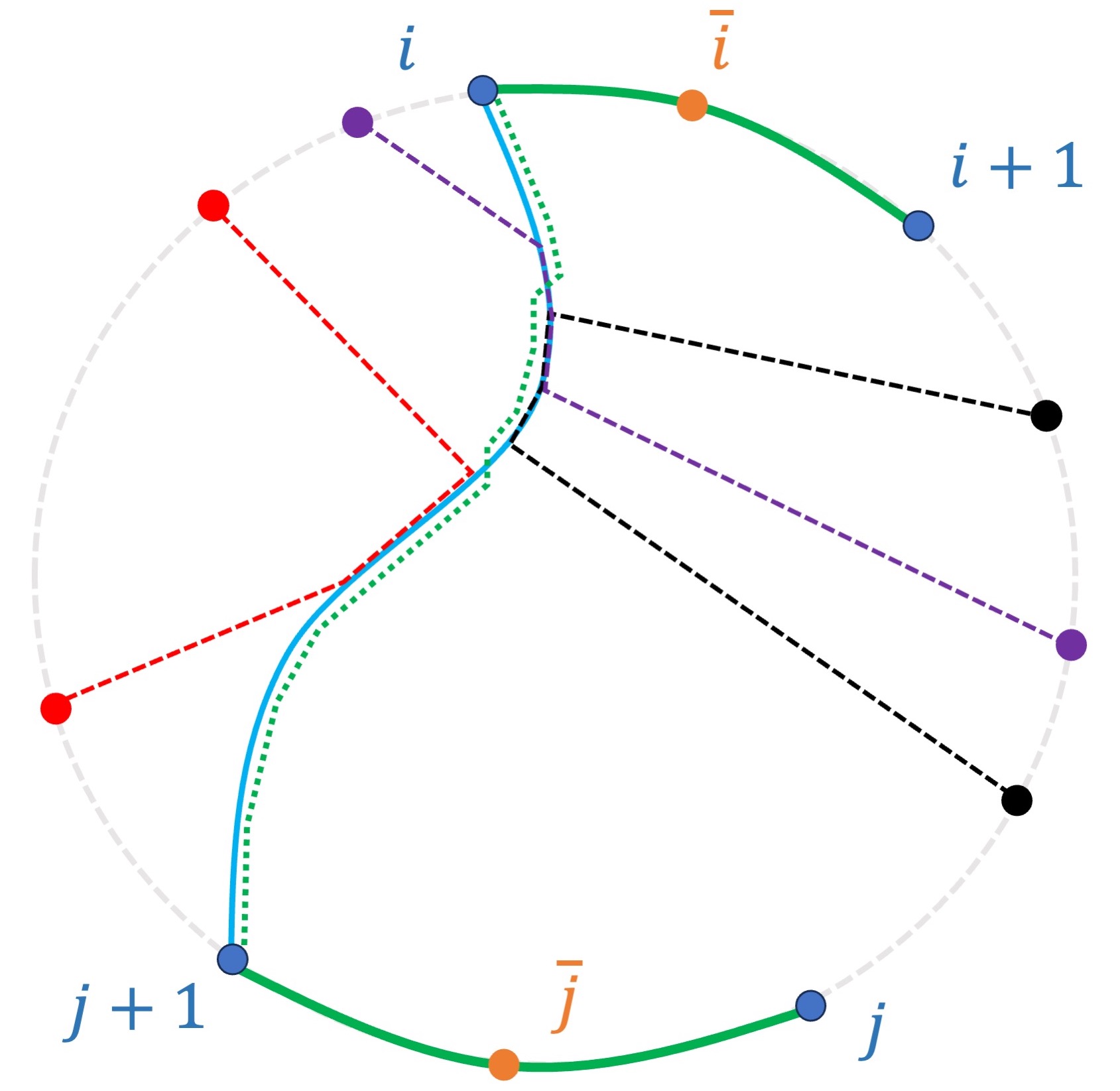}
\caption{The $i$-$(j+1)$ path $P$ (blue), and the curve $\gamma$ along it (dashed green line).\label{fig: dominate}}
\end{figure} 

\begin{claim}
\label{clm: dominance}
    $C \succeq C'$.
\end{claim}

\begin{proof}
Consider any cut $(\bar{i},\bar{j})$. On the one hand, by definition,
$\dist_C(\bar{i},\bar{j})=\sum_{(t,t'): \text{ crosses } (\bar{i},\bar{j})} F_{t,t'}$. On the other hand, we show that
$\dist_{C'}(\bar{i},\bar{j})\le \sum_{(t,t'): \text{ crosses } (\bar{i},\bar{j})} F_{t,t'}$.

Assume that $(\bar{i},\bar {j})$ cuts the boundary into two segments (named left and right).
Consider the $i$-$(j+1)$ path, which we denote by $P$.
We define $L$ as the set of vertices in $P$ that belongs to some designated path $P_{\ell,\ell'}$ connecting a pair of terminals $\ell,\ell'$ lying on the left boundary with $F_{\ell,\ell'}>0$, and we define $R$ similarly.
From Property \ref{prop: disjoint}, $L\cap R=\emptyset$. We now draw a curve $\gamma$ along $P$ as follows: start from the right side of $i$, follow the trajectory of $P$ on the same side until we reach a vertex in $R$, switch to the left side of $P$, follow the trajectory of $P$ until we reach a vertex in $L$, switch to the right side of $P$, $\ldots$ , until finally we reach $(j+1)$ from the right side (see \Cref{fig: dominate}).

We now show that the curve $\gamma$ satisfies the following properties:
\begin{itemize}
    \item for every pair $\ell,\ell'$ of terminals on the left segment with $F_{\ell,\ell'}>0$, $P_{\ell,\ell'}$ does not intersect $\gamma$; 
    \item for every pair $r,r'$ of terminals on the left segment with $F_{r,r'}>0$, $P_{r,r'}$ does not intersect $\gamma$; 
    \item for every pair $\ell,r$ of terminals with $\ell$ on the left segment and $r$ on the right segment such that $F_{\ell,r}>0$, $P_{\ell,r}$ intersects $\gamma$ exactly once; 
\end{itemize}

The first and second properties are immediate according to our construction of $\gamma$.
We now verify the third property.
Assume that some left-right terminal path $P'$ intersects the curve $\gamma$ more than once. Denote $X=P\cap P'$, so on $X$, $\gamma$ switches from the left side to the right side of $P$. If $\gamma$ switches back to the right side of $P$ before reaching the end of $X$, this means that it meets a vertex from $L$, which is a contradiction since such a vertex belongs to some $P_{\ell,\ell'}$ with $\ell,\ell'$ on the left segment and $F_{\ell,\ell'}>0$, and from Property \ref{prop: disjoint} it should be disjoint from $P'$. Therefore, $\gamma$ cannot switch back to the right side of $P$ before reaching the end of $X$, which means that $\gamma$ intersects $P'$ exactly once.

Finally, by Property \ref{flowprop_3}, for each edge $e$, the amount of $F'$ through $e$ is smaller than the amount of $F$ through $e$. Therefore, with the above three properties of $\gamma$, the total amount of $F'$ that goes through the cut $(\bar{i},\bar {j})$ is at most $\sum_{(t,t'): \text{ crosses } (\bar{i},\bar{j})} F_{t,t'}$, which means that $\dist_{C'}(\bar{i},\bar{j})\le \sum_{(t,t'): \text{ crosses } (\bar{i},\bar{j})} F_{t,t'}$.
\end{proof}

\subsection{Step 3. Uncrossing vs dominance: chords and intersections}

In the last step, we prove the following pure graph-theoretic lemma on chords and intersections, which implies Theorem~\ref{lem: 2-path condition} with the help of \Cref{clm: dominance} and the four-point condition of Okamura-Seymour metrics.

\begin{lemma} \label{thm:main-uncross}
For any configurations $C,C'$, if $C\succ C'$, then $C\to C'$.
\end{lemma}

We first consider the special case where the chords do not share endpoints.
That is, we only consider configurations where the chords give a perfect matching on all points on the boundary.
We will prove the following lemma, which immediately implies \Cref{thm:main-uncross} in this special case.

\begin{lemma} \label{lem:find-one}
If $C\succ C'$, then $C$ contains a pair of chords that cross each other, such that the configuration obtained from $C$ by uncrossing these chords still dominates $C'$.
\end{lemma}

\begin{proof}
For a point $x$, denote by $C(x)$ ($C'(x)$, resp.) the point to which $x$ is matched in $C$ ($C'$, resp.), and denote by $S(x)$ the boundary segment from $C(x)$ to $C'(x)$ that does not contain $x$. 
We say a pair $x,y$ of points is \emph{good}, iff 
\begin{itemize}
\item the chords $(x,C(x))$ and $(y,C(y))$ cross each other;
\item $y$ and $C'(x)$ lie on different sides of the chord $(x,C(x))$; and
\item either $C(y)\in S(x)$ or $C'(y) \in S(x)$.
\end{itemize}
The proof is centered around good pairs, and we first prove the following two claims.
\begin{claim} \label{clm:good}
        There is at least one good pair.
    \end{claim}
    \begin{proof}
Assume for contradiction that there is no good pair. Choose a point $x$ such that $C(x)\neq C'(x)$. Assume w.l.o.g. that $x, C'(x), C(x)$ appear on the boundary clockwise in this order. Consider the cut $\left(\overline{x},\overline{C(x)-1}\right)$. The chord $(x,C(x))$ does not cross this cut, while $(x,C'(x))$ does.

For a point $y$ that lies clockwise between $C(x)$ and $x$, if the chord $(y,C(y))$ crosses $\left(\overline{x},\overline{C(x)-1}\right)$, then it must also cross $(x,C(x))$. Therefore, $C(y)$ must lie clockwise between $x$ and $C'(x)$, otherwise $(x,y)$ is a good pair. 
Now consider the location of $C'(y)$. If it lies clockwise between $y$ and $x$, then $(y, C(x))$ is a good pair. If it lies clockwise between $C(x)$ and $y$, then $(y,x)$ is a good pair. If it lies clockwise between $C'(x)$ and $C(x)$, then $(x,y)$ is a good pair. Hence $C'(y)$ can only lie clockwise between $x$ and $C'(x)$. Therefore, $(y, C'(y))$ also crosses the cut $\left(\overline{x},\overline{C(x)-1}\right)$.
However, this implies that
\[\dist_C\left(\overline{x},\overline{C(x)-1}\right)<\dist_{C'}\left(\overline{x},\overline{C(x)-1}\right),\]  a contradiction towards the condition that $C\succ C'$.
\end{proof}

We now choose the good pair $x,y$ such that the $x$-$y$ segment that does not contains $C(x)$ contains the minimum number of points. For such a good pair, we prove the following lemma.
    \begin{claim} \label{clm:ge2}
        For any $\bar i$ between $x$ and $y$, and $\bar j$ between $C(x)$ and $C'(x)$, $\dist_C(\bar i,\bar j)-\dist_{C'}(\bar i,\bar j)\ge 2$.
    \end{claim}
    \begin{proof}
Assume w.l.o.g. that points $x$, $C(x)$, and $C'(x)$ appear on the boundary clockwise in this order. For any $\bar j \in S(x)$, consider the cuts $(\overline{x-1},\bar j)$ and $(\bar x,\bar j)$. Observe that 
$\dist_C(\overline{x-1},\bar j)-\dist_C(\bar x,\bar j)$ and
$\dist_C'(\overline{x-1},\bar j)-\dist_C'(\bar x,\bar j)$ only depends on the locations of chords $(x,C(x))$ and $(x,C'(x))$. Under our positional assumption, we have $\dist_C(\bar x,\bar j)=\dist_C(\overline{x-1},\bar j)+1$ and $\dist_{C'}(\bar x,\bar j)=\dist_{C'}(\overline{x-1},\bar j)-1$. Since $\dist_C(\overline{x-1},\bar j)\ge\dist_{C'}(\overline{x-1},\bar j)$, it follows that $\dist_C(\bar x,\bar j)-\dist_{C'}(\bar x,\bar j)\ge 2$.

We now show that the inequality still holds if we replace $\bar x$ in the above discussion by any $\bar i$ between $x$ and $y$. We prove by induction, where for base case $\bar i = \bar x$ the proof is provided above.
Consider now a general $\bar i$ between $x$ and $y$, the induction hypothesis gives $\dist_C(\overline{i-1},\bar j)-\dist_{C'}(\overline{i-1},\bar j)\ge 2$. We now examine the locations of $C(i)$ and $C'(i)$.

\paragraph{Case 1: $C(i)$ lies clockwise between $i$ and $C(x)$.} In this case, $\dist_C(\bar i,\bar j)=\dist_C(\overline{i-1},\bar j)+1$. Consequently,
$\dist_C(\bar i,\bar j)-\dist_{C'}(\bar i,\bar j)\ge \dist_C(\overline{i-1},\bar j)-\dist_{C'}(\overline{i-1},\bar j)\ge 2.$

\paragraph{Case 2: $C(i)$ lies clockwise between $C(x)$ and $C'(x)$.} In this case, $x,i$ is a good pair and the $x$-$i$ segment contains fewer points than the $x$-$y$, a contradiction to our assumption on the pair $x,y$.

\paragraph{Case 3: $C(i)$ lies clockwise between $C'(x)$ and $x$.} In this case, $\dist_C(\bar i,\bar j)=\dist_C(\overline{i-1},\bar j)-1$. We then examine the position of $C'(i)$. If $C'(i)$ lies clockwise between $i$ and $C'(x)$, then $i,x$ forms a good pair such that the $x$-$i$ segment contains fewer points than the $x$-$y$, a contradiction.
If it lies clockwise between $C'(x)$ and $i$, then $\dist_{C'}(\bar i,\bar j)=\dist_{C'}(\overline{i-1},\bar j)-1$, which implies that
        \[
            \dist_C(\bar i,\bar j)-\dist_{C'}(\bar i,\bar j) = \dist_C(\overline{i-1},\bar j)-\dist_{C'}(\overline{i-1},\bar j)\ge 2.
        \]
        
\paragraph{Case 4: $C(i)$ lies clockwise between $x$ and $i$.} In this case, $\dist_C(\bar i,\bar j)=\dist_C(\overline{i-1},\bar j)-1$, and we examine the position of $C'(i)$. If $C'(i)$ lies clockwise between $C'(x)$ and $i$, then the proof is the same as Case 3. It remains to consider the case where $C'(i)$ lies clockwise between $i$ and $C'(x)$. Consider the cut $(\bar i,\overline{C(i)-1})$. The chord $(i,C(i))$ does not cross it, while the chord $(i,C'(i))$ does. This implies that there exists a point $k$ between $i$ and $C(i)$ such that $(k,C(k))$ crosses $(\bar i,\overline{C(i)-1})$ while $(k,C'(k))$ does not. But then either $(k,i)$ or $(k,C(i))$ would form a good pair whose boundary segment contains fewer points than the $x$-$y$ segment, a contradiction.
\end{proof}

We now show that if we uncross the chords $(x, C(x))$ and $(y, C(y))$ in $C$, then the resulting configuration still dominates $C'$. Observe that this uncrossing operation only changes the distances between a point in the  $x$-$y$ segment and a point in the $C(x)$-$ C(y)$ segment, and all such distances are decreased by 2.
    
Since $x,y$ is a good pair, either $C(y)$ or $C'(y)$ lies in $S(x)$. If $C(y)$ lies in $S(x)$, then from \Cref{clm:ge2}, the entire $C(x)$-$C(y)$  boundary segment is contained in $S(x)$. Therefore, after uncrossing, all affected distances remain greater than those in $C'$. 
If $C(y)$ does not lie in $S(x)$ but $C'(y)$ lies in $S(x)$, then $y,x$ is also a good pair. From \Cref{clm:ge2}, for any $\bar i$ lying in the $x$-$y$ segment and any $\bar j \in S(x)\cup S(y)$, $\dist_C(\bar i,\bar j)-\dist_{C'}(\bar i,\bar j) \ge 2$. Observe that in this case $S(x) \cup S(y)$ is exactly the $C(x)$-$C(y)$ segment. Hence, after uncrossing, all affected distances remain greater than those in $C'$. 
    
This completes the proof of \Cref{lem:find-one}.
\end{proof}

We now complete the proof of \Cref{thm:main-uncross} by a reduction to the special case where each point belongs to exactly one chord. We perform iterations to split vertices.
Initially, set $\hat C=C$ and $\hat C'=C'$, so $\hat C, \hat C'$ are configurations on the same set of vertices.
Consider a vertex $v$ that is the endpoint of $a>1$ chords in $\hat C$, since flows $F,F'$ are aligned (Property \ref{prop: aligned}), $v$ is also the endpoint of $a$ chords in $\hat C'$.
We split $v$ into $a$ new vertices $v_1,v_2,\ldots, v_a$ clockwise in this order. 
Denote by $u_1,u_2,\ldots,u_a$ ($w_1,w_2,\ldots,w_a$, resp.) 
the other endpoints of $v$-incident chords in $\hat C$ ($\hat C'$, resp.) clockwise in this order. 
We replace the $v$-incident chords in $\hat C$ by $(v_1,u_1),\dots,(v_a,u_a)$, and replace the $v$-incident chords in $\hat C'$ by $(v_1,w_a),\dots,(v_a,u_1)$. 

Initially, $\hat C=C$ dominates $\hat C'=C'$.
We show that after each iteration, $\hat C$ still dominates $\hat C'$. 
There are two types of cuts to verify: those that separate the newly created vertices and those that do not.
For a cut that does not separate new vertices, clearly the values $\dist_{\hat C}(\bar i,\bar j)$ and $\dist_{\hat C'}(\bar i,\bar j)$ do not change, and so $\dist_{\hat C}(\bar i,\bar j)\ge \dist_{\hat C'}(\bar i,\bar j)$ still holds.
Consider now a cut $(\bar i,\bar j)$ that separate new vertices.

If both $\bar i$ and $\bar j$ lie within the segment $(v_1,\ldots,v_a)$, then $\dist_{\hat C}(\bar{i},\bar{j})=\dist_{\hat C'}(\bar{i},\bar{j})$ as they both equal the number of new vertices between $\bar{i}$ and $\bar{j}$. Consider now the case where $\bar i$ is in $(v_1,\ldots,v_a)$ while $\bar j$ is not. Specifically, let $\bar{i_0}$ be the point to the left of $v_1$, and for each $1 \le \ell \le a-1$, let $\bar{i_{\ell}}$ be the be the point between $v_{\ell}$ and $v_{\ell+1}$ on the boundary, and let $\bar{i_a}$ be the point to the right of $v_a$. 

Consider the number of new chords in $\hat C$ crossed by the cut $(\bar i,\bar j)$.
Assume that $(\bar{i_{0}},\bar{j})$ crosses $a_1$ new chords, then $(\bar{i_{a}},\bar{j})$ crosses $a-a_1$ new chords, for each $0 \le \ell \le a-a_1$, $(\bar{i_{\ell}},\bar{j})$ crosses $a_1+\ell$ new chords, and for each $a-a_1 \le \ell \le a$, $(\bar{i_{\ell}},\bar{j})$ crosses $a-(\ell-a_1)=a_1+a-\ell$ new chords. 
Therefore, when $\ell$ goes from $0$ to $a$, the value of $\dist_{\hat C}(\bar{i_{\ell}},\bar{j})$ first increases and then decreases. 
Similarly, we can show that value of $\dist_{\hat C'}(\bar{i_{\ell}},\bar{j})$ first decreases and then increases. Since we have shown that $\dist_{\hat C}(\bar{i_{0}},\bar{j}) \ge \dist_{\hat C'}(\bar{i_{0}},\bar{j})$ and $\dist_{\hat C}(\bar{i_{a}},\bar{j}) \ge \dist_{\hat C'}(\bar{i_{a}},\bar{j})$, it follows that for each $0 \le \ell \le a$, $\dist_{\hat C}(\bar{i_{\ell}},\bar{j}) \ge \dist_{\hat C'}(\bar{i_{\ell}},\bar{j})$ holds, and so after this iteration, $\hat C$ still dominates $\hat C'$.

Therefore, we can continue with the next iteration until all vertices incident to more than chords have been separated. And then we are in the special case where the chords do not share endpoints, whose proof is already provided.

\emph{Completing the proof of Theorem~\ref{lem: 2-path condition}}.
It is easy to verify that if configurations $C$ and $C'$ induce the same metric on $\set{\bar 1,\ldots,\bar n}$, then $C=C'$. From \Cref{clm: dominance}, $C\succeq C'$, so either $C=C'$, which means that 
\[\sum_{t,t'\in T}F_{t,t'}\cdot D(t,t')= \sum_{t,t'\in T}F'_{t,t'}\cdot D(t,t'),\]
or $C\succ C'$, which by \Cref{thm:main-uncross} implies that $C\to C'$. Since, in an Okamura-Seymour instance $(G,T)$, for any terminals $x,z,y,w$ appearing on the boundary in this order, $$\dist_G(x,y)+\dist_G(z,w)\ge \max\set{\dist_G(x,z)+\dist_G(y,w),\dist_G(x,w)+\dist_G(z,y)}.$$ As $D(\cdot,\cdot)=\dist_G(\cdot,\cdot)$ the value of $\sum_{t,t'\in T}F_{t,t'}\cdot D(t,t')$ decreases after each uncrossing. Therefore,
\[\sum_{t,t'\in T}F_{t,t'}\cdot D(t,t')\ge  \sum_{t,t'\in T}F'_{t,t'}\cdot D(t,t').\]
To sum up, we showed that any pair $F,F'$ of flows satisfying Properties \ref{flowprop_1}, \ref{flowprop_2}, \ref{prop: aligned}, and \ref{prop: disjoint} cannot satisfy Property \ref{flowprop_3}, and therefore there exist no flows $F,F'$ satisfying all the properties. This completes the proof of Theorem~\ref{lem: 2-path condition}.

\section{Recognizing Outerplanar Metrics: Proof of Theorem 2}\label{sec:outerplanar-proof}

\subsection{High level overview}

The proof of Theorem~\ref{thm: no O(1) point} is given in Appendix~\ref{apd: no O(1)-point}. In the remainder of this section, we prove Theorem~\ref{thm: main}. Throughout this overview, we assume that Theorem~\ref{lem: 2-path condition} has already been proved. Consequently, it is enough to construct an outerplanar graph together with a shortest path structure in which the paths assigned to every repelling pair are disjoint; Theorem~\ref{lem: 2-path condition} then supplies the edge lengths.

\paragraph{Step 1. Compute the circular ordering of terminals.}
An outerplanar graph with terminals is, in particular, an Okamura--Seymour instance. We therefore begin with the circular-ordering algorithm for OS metrics from~\cite{chen2025path}. It either certifies that no cyclic order satisfies the Kalmanson inequalities, in which case the metric is not outerplanar, or produces a compatible cyclic order.

There is one subtlety: the compatible order need not be unique. In this case, the ordering algorithm identifies two consecutive blocks of terminals, which we denote by $A$ and $B$, whose cross distances have the additive form
\[
D(a,b)=\alpha_a+\beta_b\qquad (a\in A,b\in B).
\]
We introduce a new terminal $o$ with $D(a,o)=\alpha_a$ and $D(o,b)=\beta_b$, and recurse on the two induced metrics on $A\cup\{o\}$ and $B\cup\{o\}$. If both are outerplanar, their realizations can be glued at $o$. Conversely, from any outerplanar realization of the original metric, the relevant points on the cross-block shortest paths can be contracted to obtain outerplanar realizations of the two subinstances. Thus any compatible order produced by the procedure is sufficient. The details appear in Section~\ref{sec: ordering}.

\paragraph{Step 2. Find a canonical outerplanar graph and segment traces.}
Once the terminal order is fixed, we massage any outerplanar realization into a canonical form: the outer boundary is a simple cycle, each terminal has degree two, every other vertex has degree three, and the shortest path between each consecutive pair of terminals is their boundary segment. Denote these boundary segments by $S_1,\ldots,S_k$. The remaining edges connect distinct segments; we call them channels. After contracting each segment to a single vertex, the channels form a maximal outerplanar graph.

\begin{figure}[h]
\centering
\includegraphics[width=0.45\linewidth]{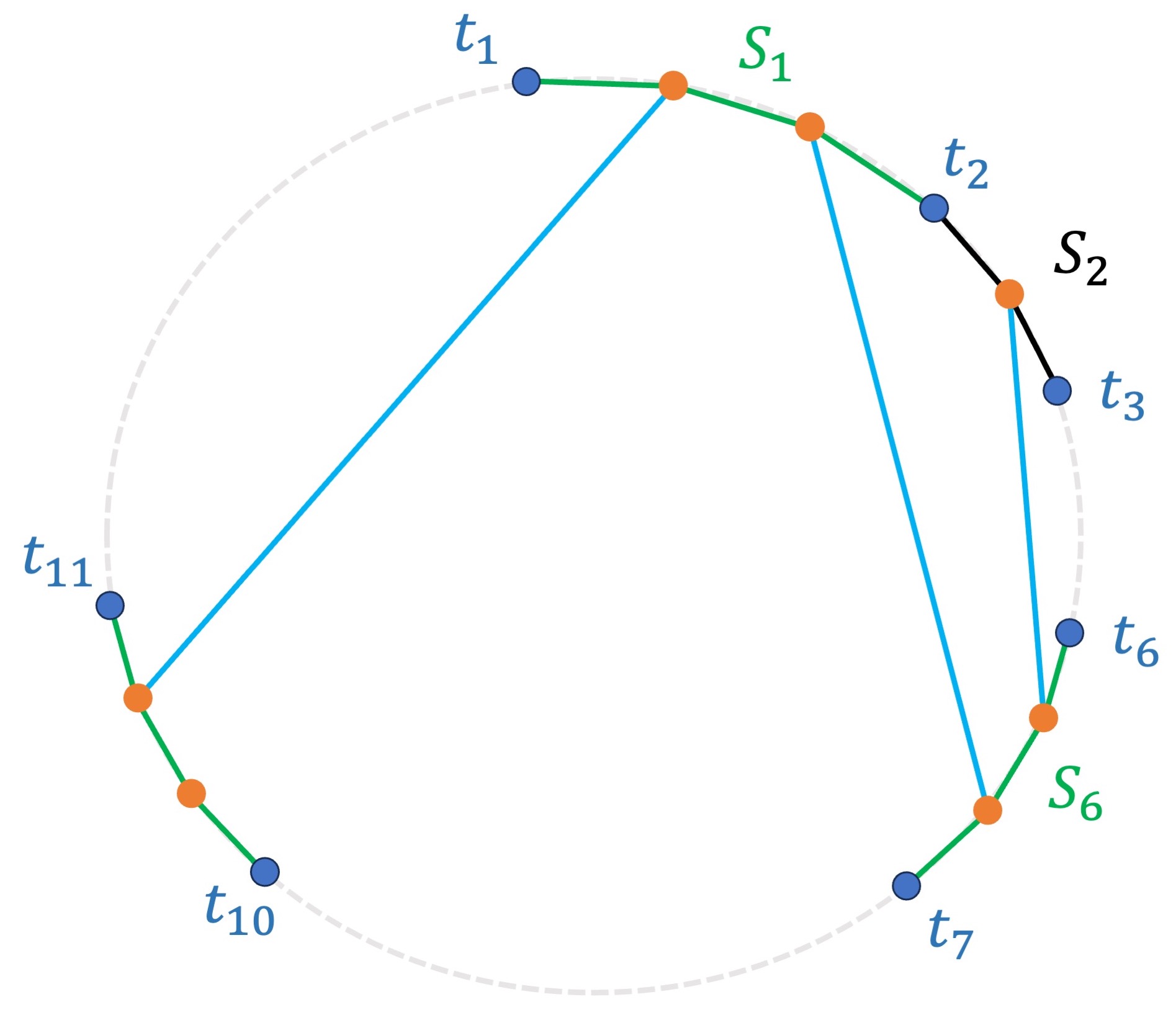}
\caption{A canonical outerplanar graph: terminals $t_1,t_2,\ldots,t_k$ and segments $S_1,\ldots,S_k$.}\label{fig: canonical_intro}
\end{figure}

The advantage of this form is that a terminal path is described, at a coarse level, by the sequence of boundary segments it visits. Since $S_r$ is the designated shortest path between $t_r$ and $t_{r+1}$, a path for the terminal pair $(t_i,t_j)$ may visit $S_r$ only if $(t_i,t_j)$ does not repel $(t_r,t_{r+1})$. We therefore first search for a maximal outerplanar channel structure together with a segment sequence for every terminal pair satisfying all of these segment-level restrictions.

This search is carried out by an interval dynamic program. A state on the segments $S_i,\ldots,S_j$ enforces the outer channel $S_iS_j$. The triangular face incident to this channel supplies an anchor $i<\ell<j$ and decomposes the state into two smaller intervals. A terminal pair separated by the anchor is routed either through the middle segment $S_\ell$ or through the outer channel $S_iS_j$; the repelling relations determine which choices are allowed. If no anchor passes these tests, the state is infeasible. This construction and its correctness proof are given in Section~\ref{sec: dp}.

\paragraph{Step 3. Convert segment traces into paths and assign weights.}
The dynamic program initially records only segment traces. To obtain an actual shortest path structure, we replace every channel by three nearby parallel edges. Paths whose two terminals lie on the same side of the channel use the copy closest to that side, while paths whose terminals lie on opposite sides use the middle copy. This local rule makes the intersection of every two selected paths either empty or a common subpath, and it keeps the paths for every repelling pair disjoint.

We may now invoke Theorem~\ref{lem: 2-path condition} to compute nonnegative edge lengths realizing the input metric. The circular-ordering stage, the dynamic program, and the final edge-length computation together run in $O(k^5)$ time.

\subsection{Computing the Circular Ordering of Terminals}
\label{sec: ordering}

As shown in the previous work \cite{ChangO20,hurkens1988tidy}, for a metric $D$ on a set $T=\set{t_1,\ldots,t_k}$ of terminals, there exists an Okamura-Seymour instance $(G,T)$ that realizes $D$ as the shortest-path distance metric on $T$, such that terminals $t_1,\ldots,t_k$ appearing on the boundary in this order, iff the following \emph{four-point condition} (called \emph{Monge property}) holds for all tuples $1\le a<b<c<d\le k$:
\[
\emph{$D(t_a,t_c)+D(t_b,t_d)\ge \max\set{D(t_a,t_b)+D(t_c,t_d), D(t_a,t_d)+D(t_b,t_c)}$.}
\]
Since outerplanar graphs with terminals are Okamura-Seymour instances, we can apply the algorithm in \cite{chen2025path} to compute a circular ordering of terminals appearing on the boundary, such that the input metric $D$ satisfies all four-point conditions or certify there are no such orderings. If the answer is \textsf{No}, then this means that $D$ is not realizable by any Okamura-Seymour instances, and we report that $D$ is not realizable by any outerplanar graphs.
However, if there is more than one circular ordering that satisfies all four-point conditions, then we are uncertain about which one we should choose to proceed with further steps.
In this section, we address this question by showing that any one will do.

One crucial feature in the algorithm of \cite{chen2025path} is that, if there is more than one ordering that satisfies all four-point conditions, then the algorithm will produce a bi-partition of the terminals into two sets $S,T$ (abusing the notation $T$ here for a subset of terminals), such that the terminals of the same set must appear together on the boundary, and for any $s_1,s_2 \in S$ and $t_1,t_2 \in T$, 
\[
D(s_1,t_1)+D(s_2,t_2)=D(s_1,t_2)+D(s_2,t_1).
\]
Denote $S=\set{s_1,\ldots,s_p}$ and  $T=\set{t_1,\ldots,t_q}$. We prove the following strengthening of this fact.


\begin{claim} \label{block-exist}
There exist non-negative real numbers $\set{\alpha_i}_{1\le i\le p}$ and $\set{\beta_j}_{1\le j\le q}$, such that
    \begin{itemize}
        \item for all $1\le i\le p$ and $1\le j\le q$, $D(s_i,t_j)=\alpha_i + \beta_j$;
        \item for all $1\le i,j\le p$, $\alpha_i+\alpha_j \ge D (s_i,s_j)$; and
        \item for all $1\le i,j \le q$, $\beta_i+\beta_j \ge D(t_i,t_j)$.
    \end{itemize}
\end{claim}

\begin{proof}
For all $s_i,s_j \in S$ and all $t_\ell,t_r \in T$, 
$D(s_i,t_\ell)+D(s_j,t_r)=D(s_i,t_r)+D(s_j,t_\ell)$.
This means that for all $s_i,s_j \in S$, the value $D(s_i,t_{\ell})-D(s_j,t_{\ell})$ is the same for all $t_{\ell} \in T$, which we denote by $\Delta_{i,j}$. Moreover, $\Delta_{i,j}+\Delta_{j,m}=\Delta_{i,m}$ holds for all $i,j,m$, which means that there are real numbers $\set{\alpha_i}_{1\le i\le p}$ and $\set{\beta_j}_{1\le j\le q}$ such that for all $s_i\in S$ and $t_j \in T$, $D(s_i,t_j)=\alpha_i+\beta_j$, and they are unique up to $+\eps$ for all $\alpha_i$ and $-\eps$ for all $\beta_j$ simultaneously.

Let $\alpha_i$'s be the smallest values such that for all $s_i,s_j\in S$, $\alpha_i+\alpha_j \ge D (s_i,s_j)$. This implies that there are $s_{i^*},s_{j^*}\in S$ with $\alpha_{i^*}+\alpha_{j^*} = D(s_{i^*},s_{j^*})$. For every $t_j\in T$, $\beta_j=D(s_1,t_j)-\alpha_1$. We now show that for any $t_i,t_j \in T$, $\beta_i+\beta_j \ge D (t_i,t_j)$, which completes the proof. 

Assume for contradiction that $\beta_i+\beta_j < D (t_i,t_j)$. Consider terminals $s_{i^*},s_{j^*},t_j,t_i$, and assume without loss of generality that they appear on the boundary in this order. By Monge property, 
    \[    
    D(s_{i^*},s_{j^*})+D (t_i,t_j) \le D(s_{i^*},t_j)+D (s_{j^*},t_j) = \alpha_{i^*}+\alpha_{j^*}+\beta_i+\beta_j.
    \]
    However, $\alpha_{i^*}+\alpha_{j^*} = D(s_{i^*},s_{j^*})$ and $\beta_i+\beta_j < D (t_i,t_j)$, a contradiction.
\end{proof}

With \Cref{block-exist}, we can design a natural divide and conquer algorithm: add a new terminal $o$, augment $D$ to a metric on $T\cup \set{o}$, where
\begin{itemize}
    \item for every $s_i\in S$, $D(s_i,o)=\alpha_i$;
    \item for every $t_j\in T$, $D(t_j,o)=\beta_j$; and
    \item for all $s_i \in S$ and $t_j \in T$, $D(s_i,t_j)=D(s_i,o)+D(o,t_j)$.
\end{itemize} 
We then proceed with the induced metrics $D_{\mid S \cup \{o\}}$ and $D_{\mid T \cup \{o\}}$. If both induced metrics are realizable by outerplanar graphs, then we simply stick these graphs together at $o$, disregarding $o$ as a terminal. From the above discussion, the resulting graph is clearly an outerplanar graph realizing $D$ on $T$.
The correctness of this divide and conquer algorithm is confirmed by the following claim.

\begin{claim}
If $D$ is outerplanar, then both $D_{\mid S \cup \{o\}}$ and $D_{\mid T \cup \{o\}}$ are outerplanar.
\end{claim}
\begin{proof}
Let $G$ be the outerplanar graph realizing $D$. We now construct outerplanar graphs $G_S$ realizing $D_{\mid S \cup \{o\}}$ and $G_T$ realizing $D_{\mid T \cup \{o\}}$.
For all $s_i\in S$ and $t_j \in T$, let $o_{i,j}$ be the point on the $s_i$-$t_j$ shortest path in $G$ with $\dist(o_{i,j},s_i) = \alpha_i$ and $\dist(o_{i,j},t_j)=\beta_j$. 
We contract all points $o_{i,j}$ into a supernode $o$. 
See \Cref{fig: cut} for an illustration.

\begin{figure}[h]
\centering
\subfigure
{\scalebox{0.11}{\includegraphics{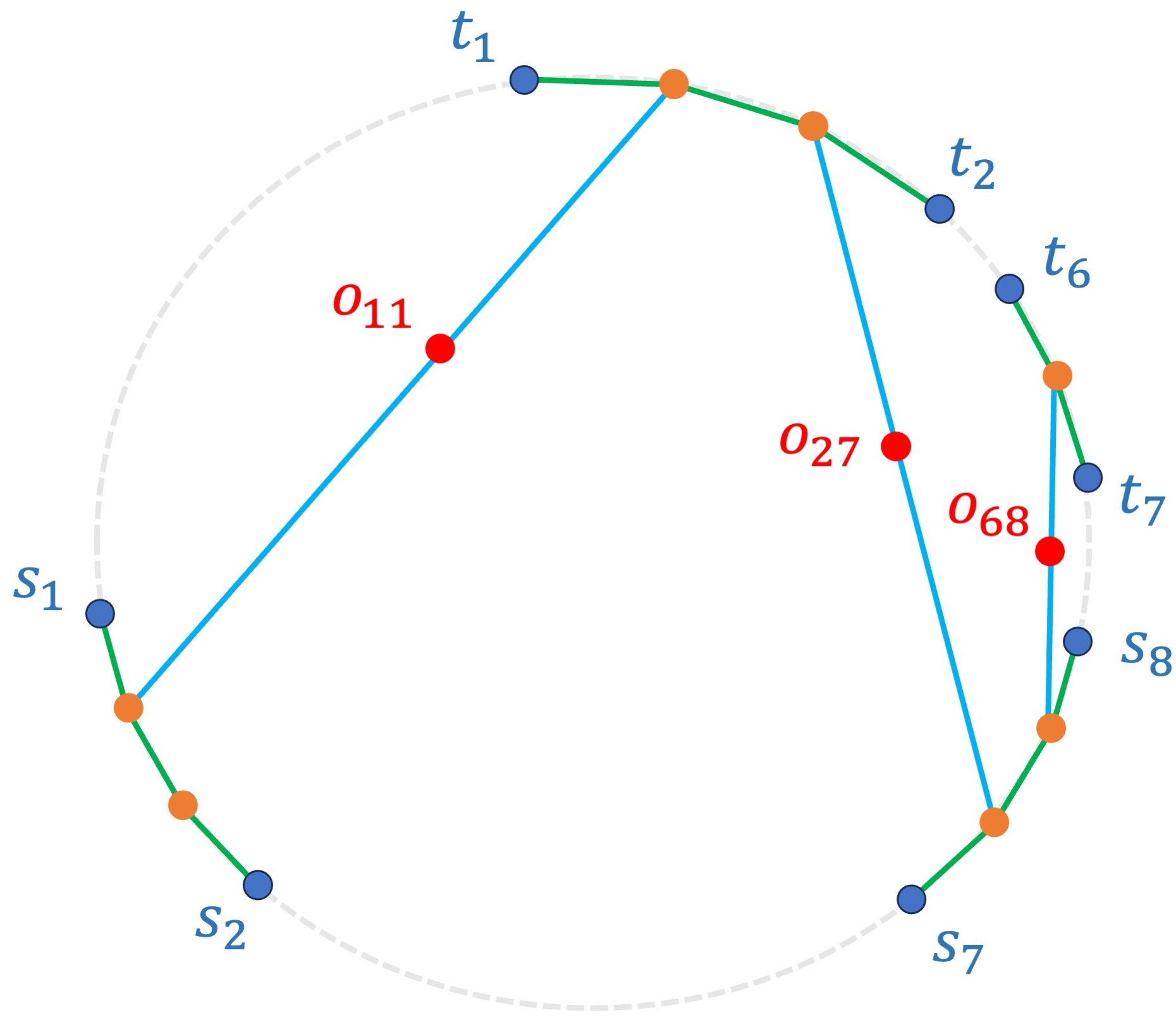}}}
\hspace{1.0cm}
\subfigure
{\scalebox{0.10}{\includegraphics{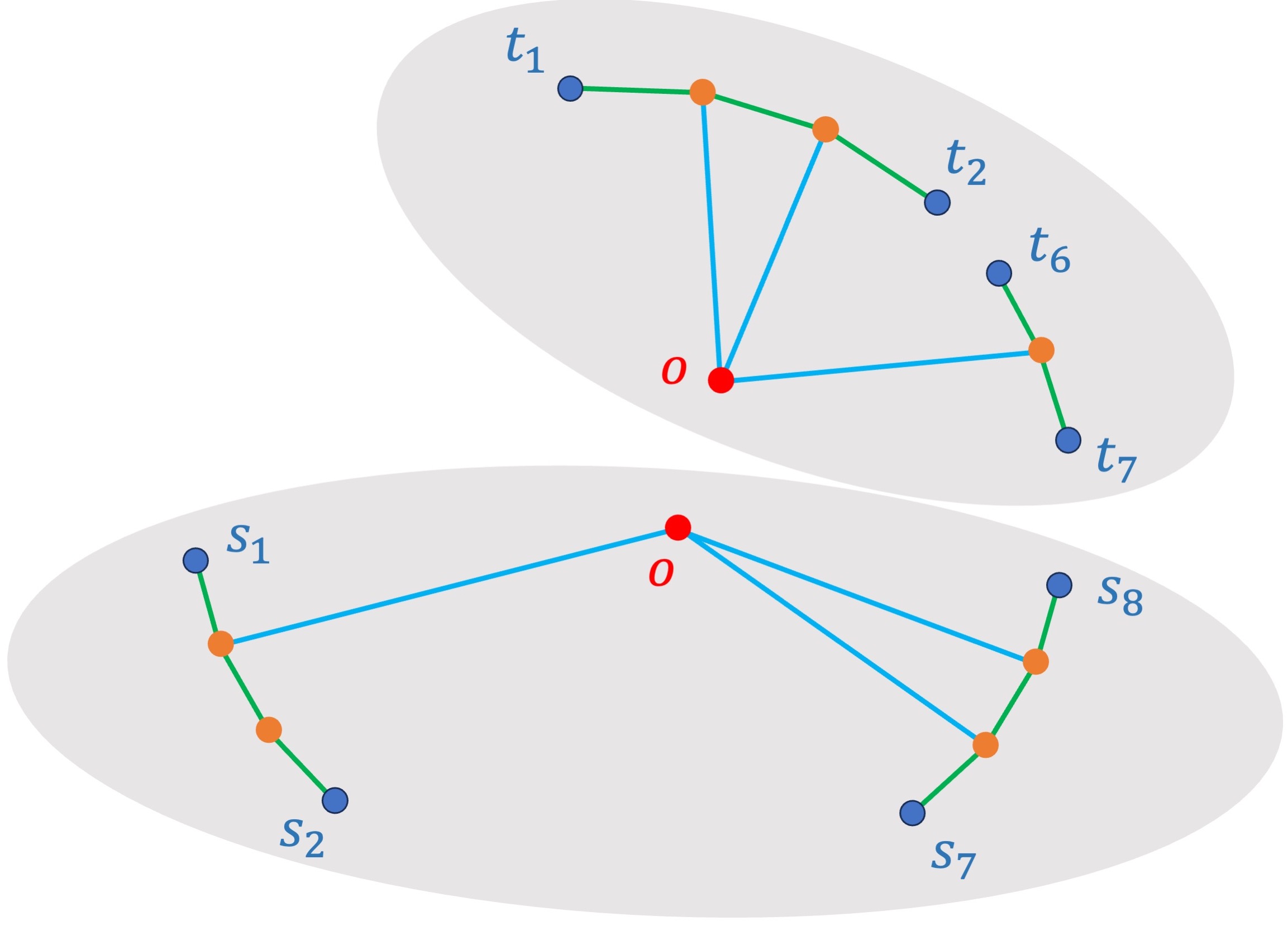}}}
\caption{Contracting all points $o_{i,j}$ into a supernode $o$, and splitting at $o$ to obtain two graphs.}\label{fig: cut}
\end{figure} 

We first prove that, after contraction, the shortest-path distance metric on terminals stays the same. Consider a terminal $s_{\ell} \in S$. For any $o_{ij}$, since $\dist(o_{ij},t_j)=\beta_j$, 
\[\dist(s_{\ell},o_{ij}) \ge \dist(s_{\ell},t_j)-\dist(o_{ij},t_j) = \alpha_{\ell}+\beta_j-\beta_j=\alpha_{\ell}.\]
Moreover, by definition $\dist(s_{\ell},o_{\ell 1})=\alpha_{\ell}$. This implies after the contraction, the new distance from $s_{\ell}$ to $o$ is exactly $\alpha_{\ell}$. 
Similarly, the new distance from $t_{\ell}$ to $o$ is exactly $\beta_{\ell}$. Therefore, for all $s_i\in S, t_j \in T$, the new distance between $s_i$ and $t_j$ is still $\alpha_i+\beta_j$. For all $s_i,s_j \in S$ (resp. $t_i,t_j \in T$), any $s_i$-$s_j$ path containing $o$ has length at least $\alpha_i+\alpha_j$, which by the properties of $\set{\alpha_i}$ is at least $D(s_i,s_j)$, so their distance is unaffected.
Similarly, the distance between any pair $t_i,t_j\in T$ is also unaffected.

We then show that, after the contraction, we indeed obtain two outerplanar graphs sharing a common vertex $o$.
From previous discussion, sets $S$ and $T$ partition the boundary of $G$  into two segments. 
We show that, for each pair $i,j$, $o_{ij}$ lies on an edge with one endpoint in $S$-segment and the other endpoint in $T$-segment.
Assume not, then without loss of generality there is a vertex $v$ on the $S$-segment with $\dist(s_i,v) > \alpha_i$, so $\dist(v,t_j) < \beta_j$. Suppose $v$ lies between $s_{\ell}$ and $s_{\ell+1}$, then $\dist(s_{\ell},v)+\dist(s_{\ell+1},v)\le \alpha_{\ell}+\alpha_{\ell+1}$. This implies either $\dist(s_{\ell},v) \le \alpha_{\ell}$ or $\dist(s_{\ell+1},v) \le \alpha_{\ell+1}$ holds. Without loss of generality, assume $\dist(s_{\ell},v) \le \alpha_{\ell}$, then $\dist(s_{\ell},v)+\dist(v,t_j) < \alpha_{\ell}+\beta_j=D(s_{\ell},t_j)$, a contradiction. 
Since every edge with one endpoint in the $S$-segment and the other endpoint in the $T$-segment belongs to some $s_i$-$t_j$ shortest path, we derive that
the $s_i$-$t_j$ shortest path $P$ contains consecutive vertices $a,b$ such that subpath $P[s_i,a]$ lies entirely in the $S$-segment, and subpath $P[b,t_j]$ lies entirely in the $T$-segment. In other words: every $S$-to-$T$ shortest path crosses segments exactly once. See \Cref{fig: cut}.
\end{proof}

\paragraph{Runtime.} The algorithm we employ in this section is a simplified version of the algorithm in \cite{chen2025path}. We give a better analysis of the runtime here: the time for computing the bipartition or certifying the uniqueness of the ordering (corresponding to refining the groups in \cite{chen2025path}) is $O(k^4)$. We may need to recurse on the obtained sets to compute bipartitions again, and the total number of times we compute the bipartitions is $O(k)$. Therefore, the runtime of the algorithm in this section is $O(k^5)$.

At the end of this section, we state the following observations on repelling pairs based on our computed circular orderings on terminals. The proofs are straightforward and only use the four-point conditions of Okamura-Seymour metrics, and are deferred to \Cref{sec: ordering proofs}.

\begin{observation} \label{obs:add}
If terminals $t_1,t_2,t_3,t_4,t_5,t_6$ lie on the boundary in this order, and $(t_1,t_6)$ does not repel $(t_3,t_4)$, then $(t_2,t_5)$ does not repel $(t_3,t_4)$.
\end{observation}

\begin{observation} \label{obs:trans}
If terminals $t_1,t_2,t_3,t_4,t_5,t_6$ lie on the boundary in this order, $(t_1,t_6)$ does not repel $(t_2,t_3)$, and $(t_2,t_6)$ does not repel $(t_4,t_5)$, then $(t_1,t_6)$ does not repel $(t_4,t_5)$.
\end{observation}

\begin{observation} \label{obs:sub}
If terminals $t_1,t_2,t_3,t_4,t_5,t_6$ lie on the boundary in this order, $(t_1,t_4)$ does not repel $(t_2,t_3)$, and $(t_2,t_4)$ does not repel $(t_6,t_5)$, then $(t_1,t_5)$ does not repel $(t_2,t_3)$.
\end{observation}

\begin{figure}[h]
\centering
\subfigure
{\scalebox{0.25}{\includegraphics{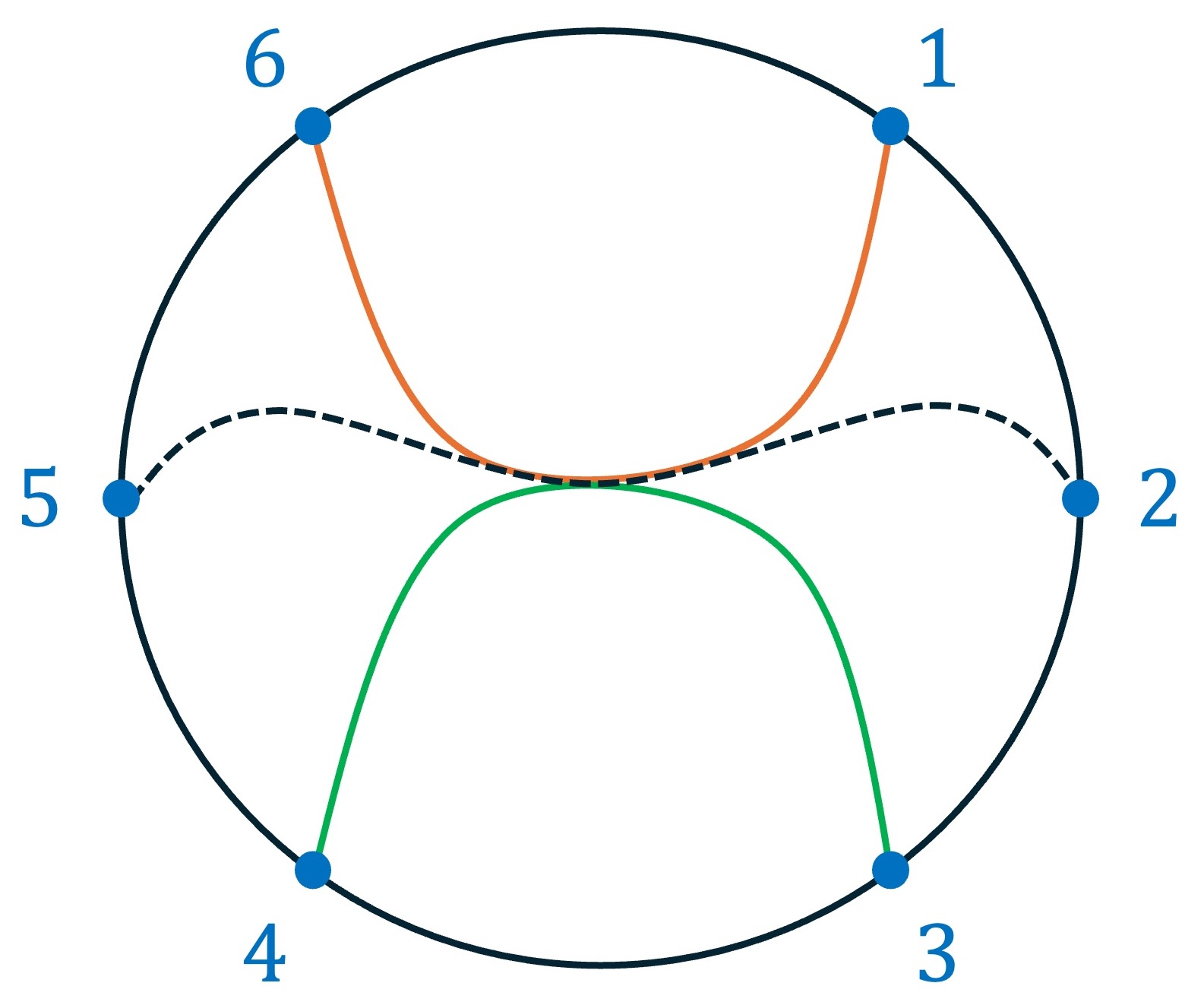}}}
\hspace{1.0cm}
\subfigure
{\scalebox{0.25}{\includegraphics{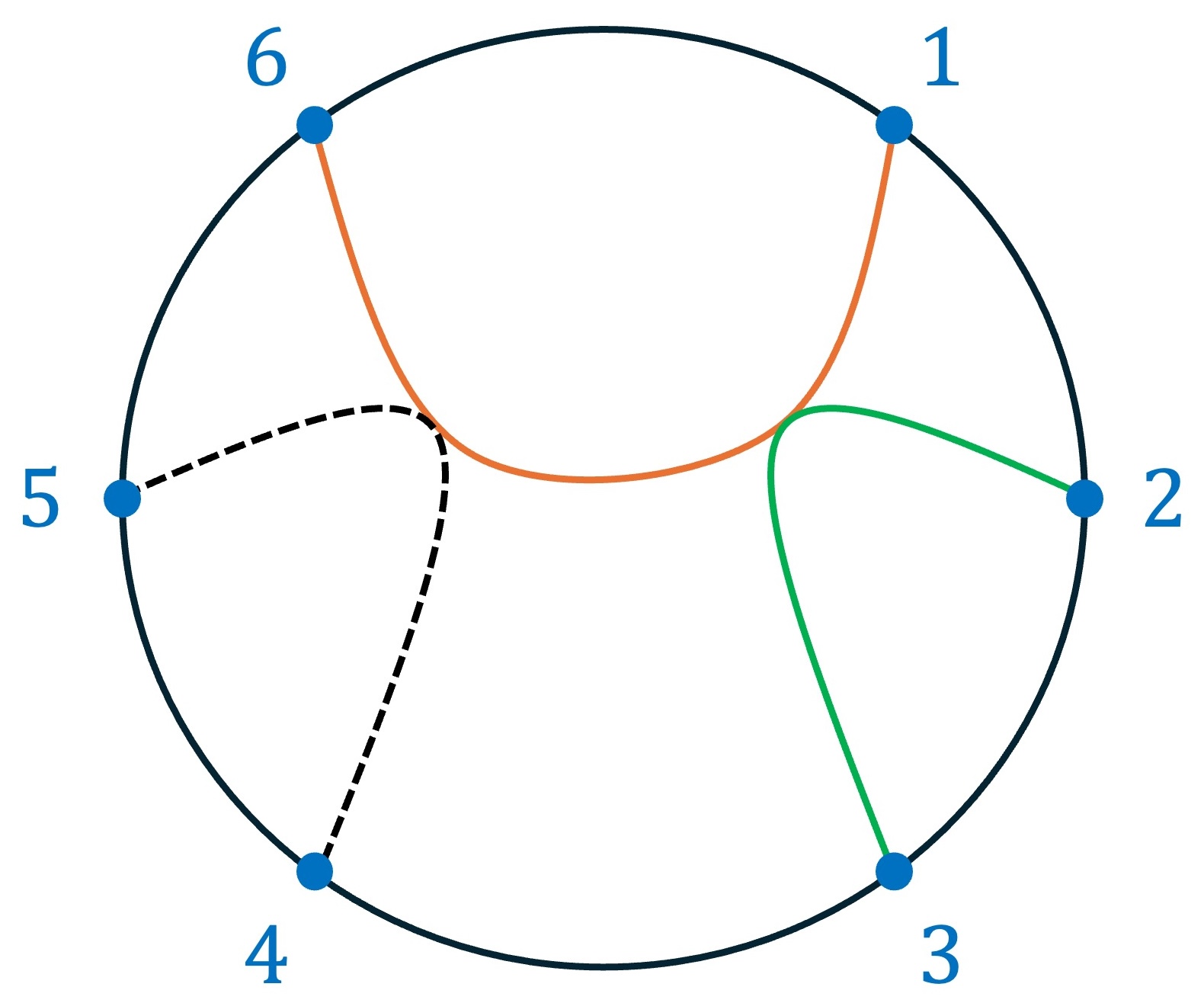}}}
\hspace{1.0cm}
\subfigure
{\scalebox{0.25}{\includegraphics{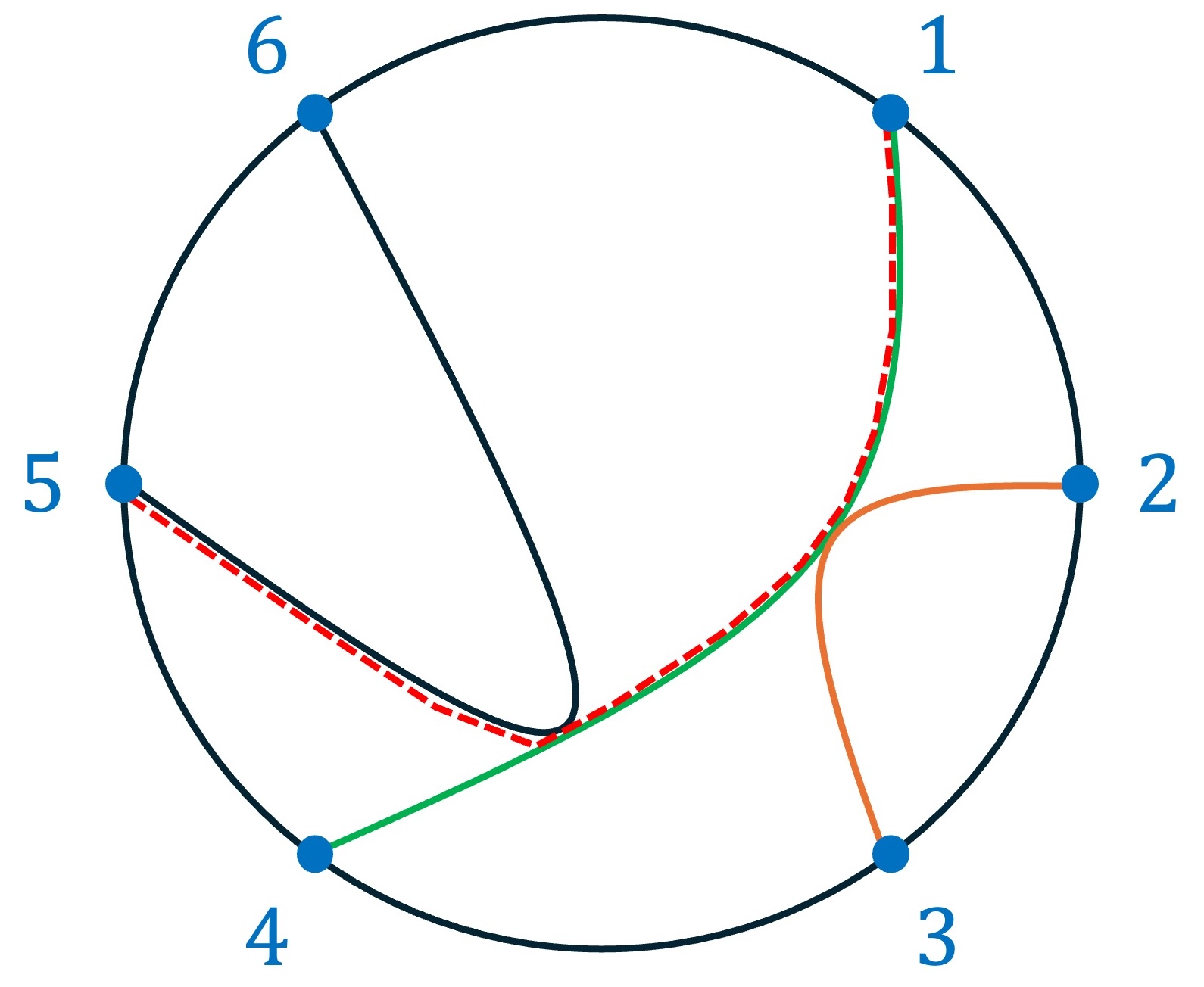}}}
\caption{Illustration of Observations \ref{obs:add} (left), \ref{obs:trans} (middle), and \ref{obs:sub} (right). Left: the $t_1$-$t_6$ path is allowed to intersect with the $t_3$-$t_4$ path, which forces the $t_2$-$t_5$ path to pass through their intersection. \label{fig: obs}}
\end{figure} 

\subsection{Dynamic Programming for Finding the Graph Structure}
\label{sec: dp}


In this section, we present the algorithm for \Cref{thm: main}. From \Cref{sec: ordering}, we have the circular ordering of terminals appearing on the boundary. From Theorem~\ref{lem: 2-path condition}, we only need a shortest path structure where the shortest paths for repelling pairs of terminals are disjoint. We will present a dynamic programming algorithm that either finds such a shortest path structure or certifies that none exists.

\subsubsection*{Step 1. Preprocessing an outerplanar graph}

Assume that the input metric $D$ is realizable by some outerplanar graph $G$. 
Our first step is to massage $G$ into a canonical form, so that later the task of finding such a graph or certifying that none exists becomes easier. Specifically, we show that $D$ can be realized by an outerplanar graph where
\begin{itemize}
    \item the boundary is a simple cycle;
    \item every terminal has degree $2$; 
    \item every other vertex has degree $3$; and
    \item the shortest path between every consecutive pair of terminals is their boundary segment.
\end{itemize}
\begin{figure}[h]
\centering
\includegraphics[width=0.45\linewidth]{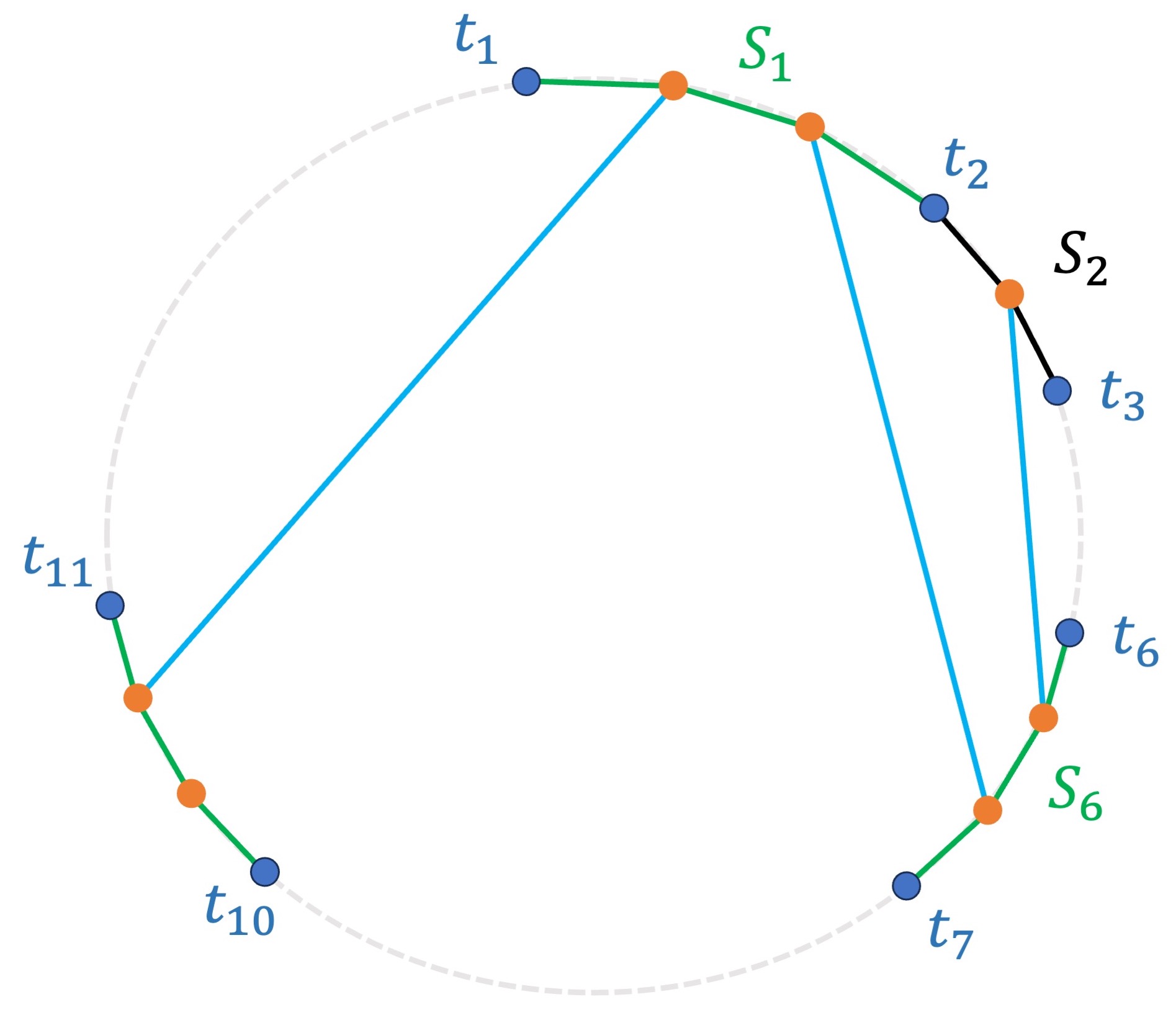}
\caption{A canonical outerplanar graph: boundary segments $S_1$ (green), $S_2$ (black), and $S_6$ (green) are shortest paths connecting their endpoints; and channels are shown in light blue.\label{fig: canonical}}
\end{figure} 
An illustration of a canonical outerplanar graph is given in \Cref{fig: canonical}.

For the first property, note that if the boundary of an outerplanar graph is not a simple cycle, then there is a cut-vertex, and we handle it in the way depicted in \Cref{fig: massage1}. For the second and third properties, we split a vertex into several depending on its degree and connect them by length-$0$ edges, as shown in \Cref{fig: massage2}. For the last property, if the shortest path $P$ between some consecutive pair $t_i,t_j$ of terminals is not their boundary segment, then we remove all vertices in the $t_i$-$t_j$ segment that does not lie in $P$ together with their incident edges (since they do not appear in any terminal-terminal shortest paths).
This either creates a cut-vertex which is then handled in \Cref{fig: massage1}, or simply make the $t_i$-$t_j$ shortest path their boundary segment in the updated graph.

\begin{figure}[h]
\centering
\subfigure[Handling a cut-vertex: we connect its lowest neighbor in the upper graph and the highest neighbor in the lower graph by an edge whose length is their distance in $G$.]
{\scalebox{0.09}{\includegraphics{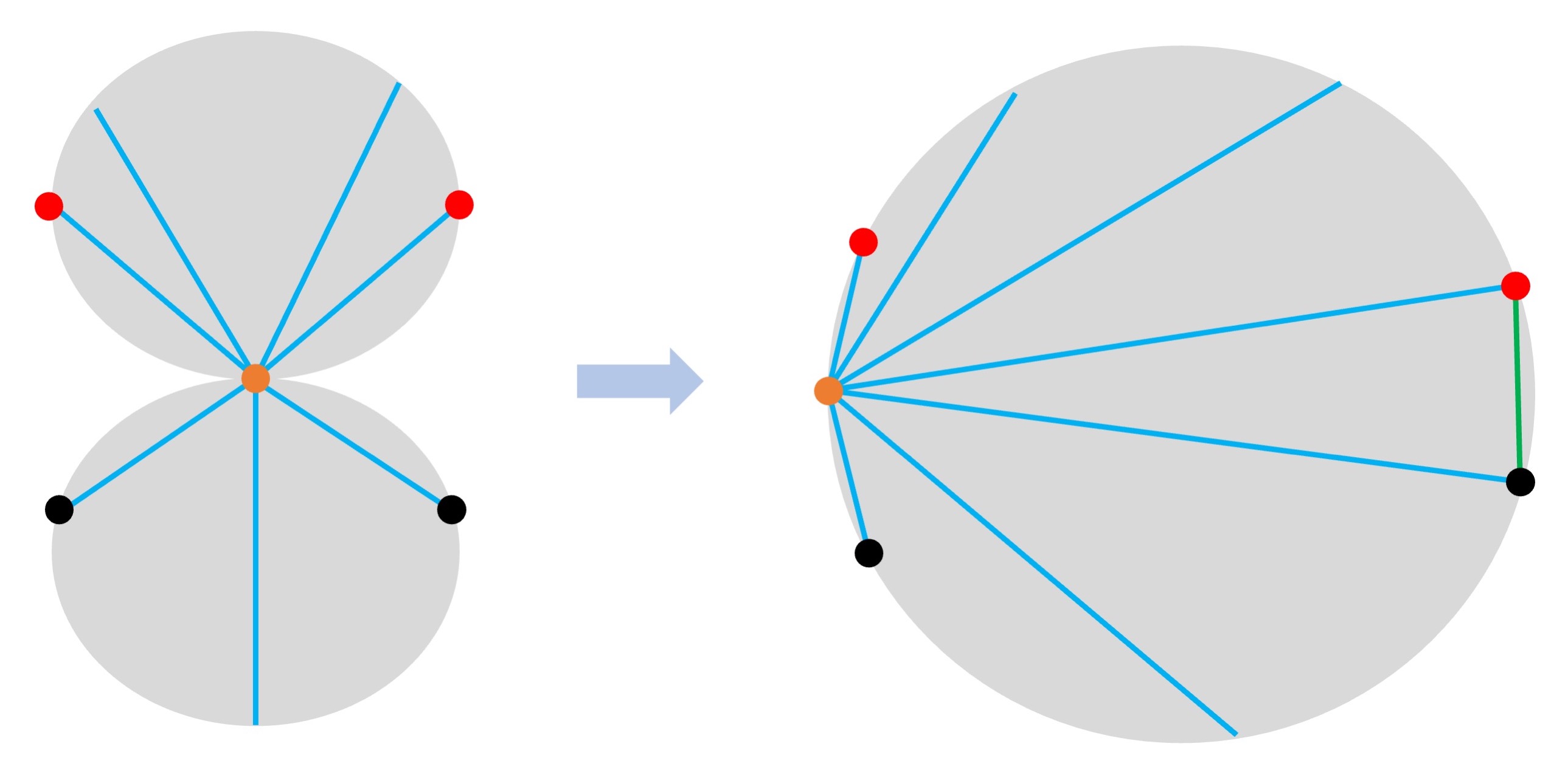}}\label{fig: massage1}}
\hspace{0.3cm}
\subfigure[Splitting a vertex: (upper) if the vertex is a terminal, and (lower) if the vertex is not a terminal. The short black and green edges have length $0$.]
{\scalebox{0.09}{\includegraphics{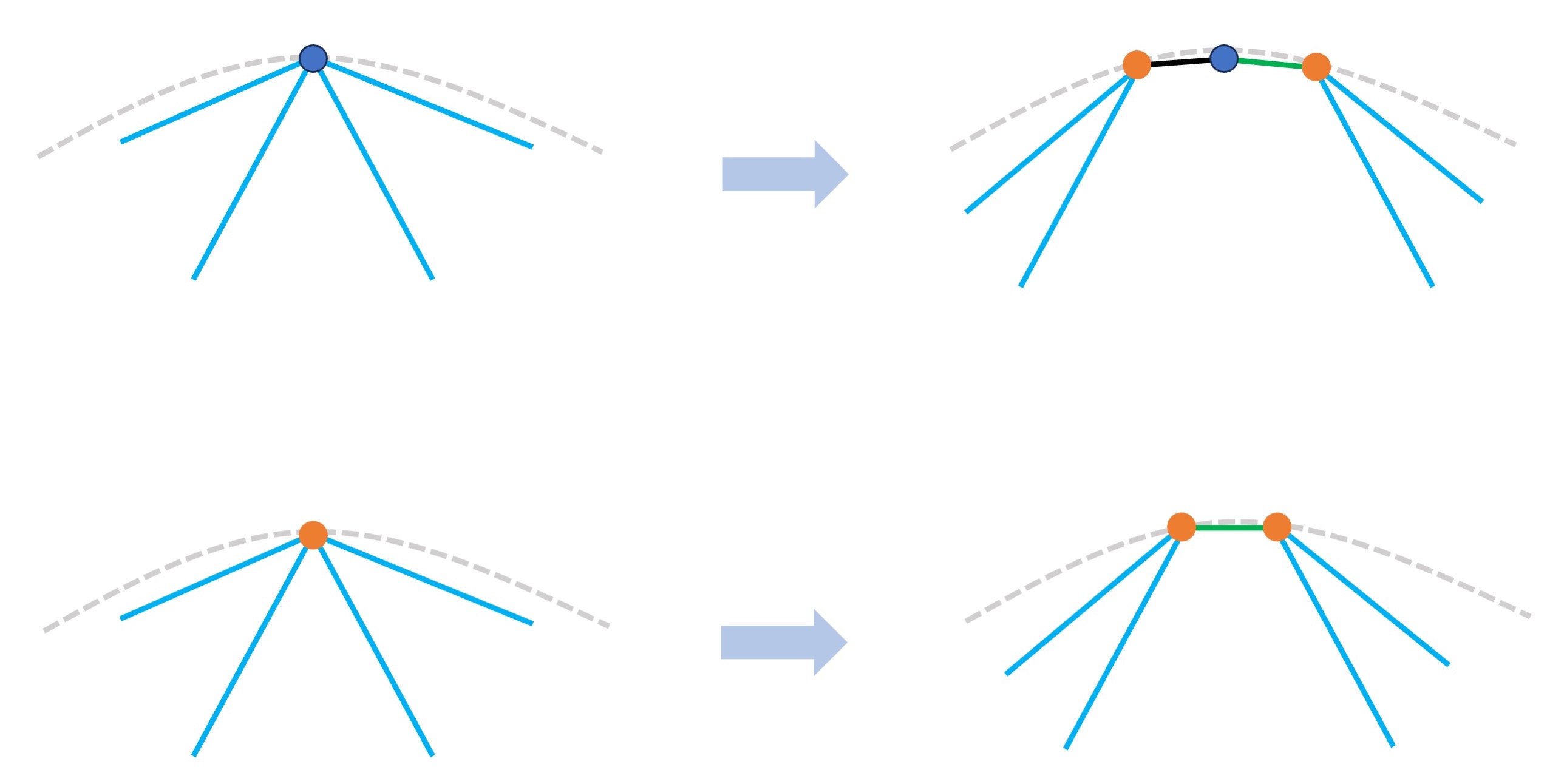}}\label{fig: massage2}}
\caption{Massaging an outerplanar graph into a canonical one.}\label{fig: massage}
\end{figure}

In a canonical outerplanar graph where the terminals $t_1,\ldots,t_k$ lie on the boundary in this order, we denote by $S_i$ the boundary segment between $t_{i}$ and $t_{i+1}$. Structurally, the segments $S_1,\ldots,S_k$ are the foundations, and the edges connecting distinct segments are the core. We call them \emph{channels}. Before we go to the next step, we state two w.l.o.g. assumptions that further regularize the channels:
\begin{itemize}
    \item between a pair $S_i,S_j$ of segments, there is at most one channel (since if there are more than one, then they must be parallel to each other and we can safely remove all but one); and
    \item the channels are maximally filled (that is, if we suppress all terminals and contract all internal vertices of each $S_i$ into a supernode $s_i$, then we get a maximal outerplanar graph on $\set{s_1,\ldots,s_k}$, which is a triangulated outerplanar graph).
\end{itemize}

\subsubsection*{Step 2. Computing the channel structure}

In a canonical outerplanar graph, the shortest path connecting a pair $t_i,t_j$ of terminals can be characterized by the sequence $(S_{r_1},S_{r_2},\ldots,S_{r_p})$ of segments it visits: such a path starts from $t_i$, goes either clockwise (if $S_{r_1}=S_{i}$) or anti-clockwise (if $S_{r_1}=S_{i+1}$) to the $S_{r_1}$-end of the $S_{r_1}$-$S_{r_2}$ channel, hops to $S_{r_2}$, finds the $S_{r_2}$-end of the $S_{r_2}$-$S_{r_3}$ channel, hops to $S_{r_3}$, $\ldots$ , and finally reaches $t_j$.

Recall that the shortest path structure we are searching for is one where all repelling pairs have disjoint shortest paths. 
So far we have only fixed the shortest paths connecting consecutive terminals: for each $r$, the $t_r$-$t_{r+1}$ shortest path is segment $S_r$. Therefore, a natural intermediate goal is to find a shortest path structure where \emph{for all $i,j,r$, if the $t_i$-$t_j$ shortest path visits $S_r$, then the pairs $(t_i,t_j)$ and $(t_r,t_{r+1})$ do not repel each other}.
In this step, we will design a dynamic programming algorithm to either find channels with a shortest path structure achieving this intermediate goal or certify that there are no such channels and shortest path structures that can achieve the goal. 

\paragraph{Subproblems.}
We set up a $2$-dimensional dynamic programming table $\Pi[i,j]$ for all $1 \le i < j \le k$. For each pair $i,j$, the entry $\Pi[i,j]$ studies the following subproblem. Imagine that we cut the boundary at $t_i,t_{j+1}$ and only retain the part consisting of segments $S_i,\ldots,S_j$. We enforce a channel between $S_i,S_j$, forming a new boundary while leaving terminals $t_i$ and $t_{j+1}$ out as tails (see \Cref{fig: DP}).
Within this new boundary, we consider the subproblem of designing channels and the shortest path structure between terminals $t_i,\ldots,t_{j+1}$.
That is, the entry $\Pi[i,j]$ stores either
\begin{itemize}
    \item a set $C[i,j]$ of maximally-filled channels between segments $S_i,\ldots,S_j$; and
    \item for each pair $t_a,t_b\in \set{t_i,\ldots,t_{j+1}}$, a path $P_{t_a,t_b}$ connecting $t_a$ to $t_b$ using only channels in $C[i,j]$, such that for each segment $S_{t_c}$ it visits, the pairs $(t_a,t_b)$, $(t_c,t_{c+1})$ do not repel each other;
\end{itemize}
or value \textsf{Null}, indicating that no such channels and shortest-path structures exist.

\begin{figure}[h]
\centering
\includegraphics[width=0.45\linewidth]{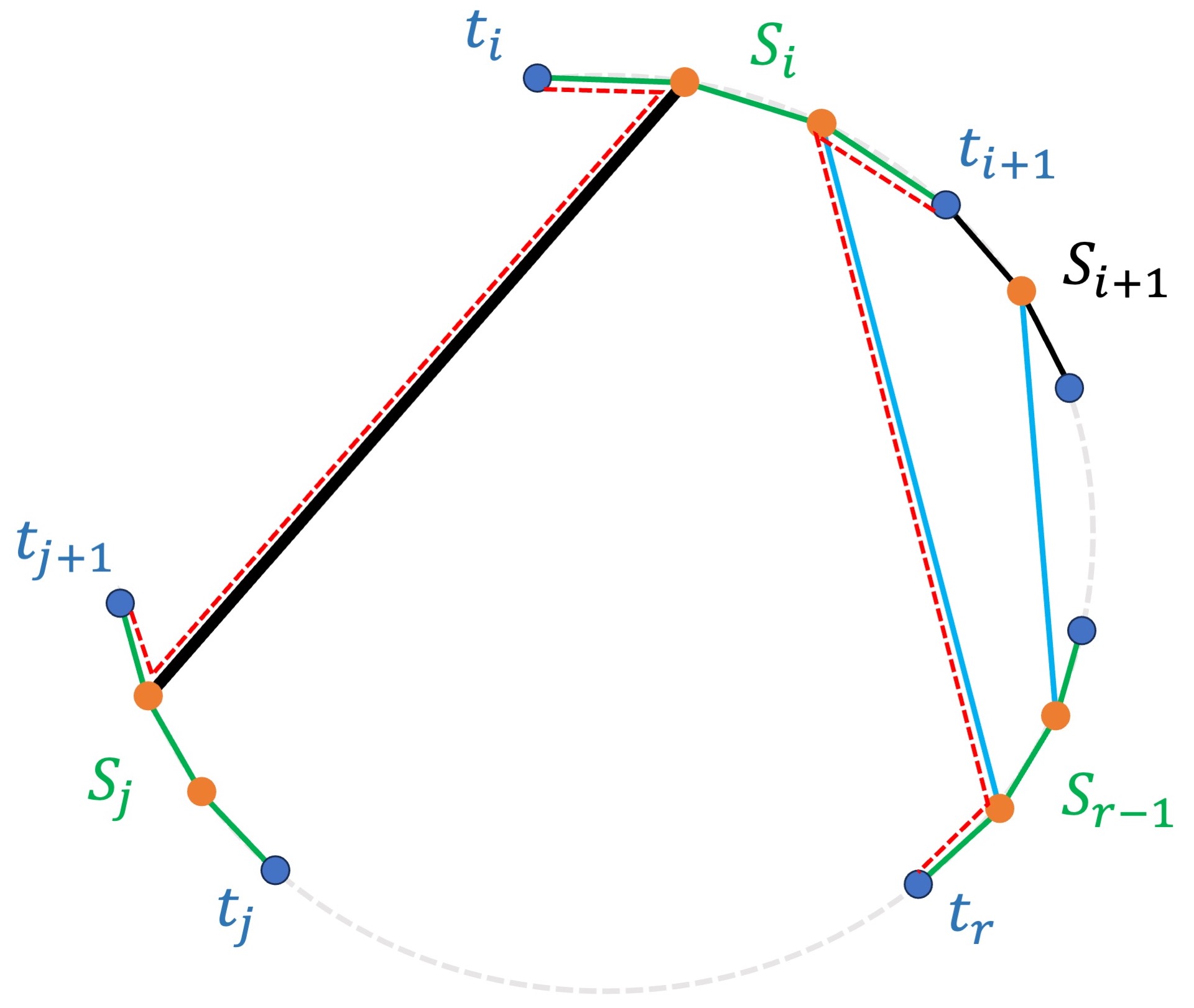}
\caption{The subproblem $\Pi[i,j]$. The $S_i$-$S_j$ channel (black) is enforced.
\label{fig: DP}}
\end{figure}

\paragraph{The dynamic programming algorithm.}
We now describe the algorithm for computing the entries in the dynamic programming table.
The base case is the entries $\Pi[i,i+1]$ for all $i \le k-1$. In these subproblems, there are no channels to add, and the shortest paths are: the $t_i$-$t_{i+1}$ shortest path is $S_i$, the $t_{i+1}$-$t_{i+2}$ shortest path is $S_{i+1}$; and the $t_i$-$t_{i+2}$ shortest path uses only the $S_i$-to-$S_{i+1}$ channel (that is, it visits both $S_i$ and $S_{i+1}$).

Consider now a general pair $i,j$ with $i\le j-2$. Since the channels need to be filled maximally, there needs to be some ``anchor index'' $i<\ell<j$ such that both $S_i$-$S_{\ell}$ and $S_{\ell}$-$S_j$ channels are included. Now for each $i<\ell<j$, we perform the following test to see if it can serve as the anchor index.

\textbf{Test 1. Both $\Pi[i,\ell]$ and $\Pi[\ell,j]$ are feasible.} That is, neither of them can be \textsf{Null}, as otherwise we cannot enforce either the $S_i$-$S_{\ell}$ channel or the $S_{\ell}$-$S_{j}$ channel.

\textbf{Test 2. For each $i\le a\le \ell$ and each $\ell\le b\le j$, either the pair $(t_a,t_b)$ does not repel $(t_\ell,t_{\ell+1})$, or it does not repel either $(t_i,t_{i+1})$ or $(t_j,t_{j+1})$.}
This is because the $t_a$-$t_b$ shortest path must visit either $S_{\ell}$ or both $S_i$ and $S_j$, due to the existence of channels $S_i$-$S_{\ell}$ $S_{\ell}$-$S_{j}$, and $S_i$-$S_{j}$.

If an index $\ell$ fails either test, then we conclude that it cannot serve as the anchor index.
If all indices fail to be the anchor index, then we set $\Pi[i,j] = \textsf{Null}$.
Otherwise, we pick any $\ell$ that passes both tests and define $\Pi[i,j]$ as follows:
\begin{itemize}
\item $C[i,j]=C[i,\ell] \cup C[\ell,j] \cup \set{(S_i,S_j)}$. Since $C[i,\ell]$ and $C[\ell,j]$ are both maximally filled, it is easy to verify that $C[i,j]$ is also maximally filled.
\item For each pair $i\le a<b\le j$,
\begin{itemize}
    \item if $a,b\le \ell$, we use the $t_a$-$t_b$ shortest path in $\Pi[i,\ell]$;
    \item if $a,b\ge \ell$, we use the $t_a$-$t_b$ shortest path in $\Pi[\ell,j]$;
    \item if $a< \ell <b$, then
\begin{itemize}
    \item if $(t_a,t_b)$ does not repel $(t_\ell,t_{\ell+1})$, then the $t_a$-$t_b$ path is the sequential concatenation of 
    \begin{enumerate}
        \item the subpath of the $t_a$-$t_{\ell+1}$ shortest path in $\Pi[i,\ell]$ between $t_a$ and segment $S_\ell$; and
        \item the subpath of the $t_\ell$-$t_b$ shortest path in $\Pi[\ell,j]$ between segment $S_\ell$ and $t_b$;
    \end{enumerate}
    \item otherwise, $(t_a,t_b)$ does not repel either $(t_i,t_{i+1})$ or $(t_j,t_{j+1})$, then the $t_a$-$t_b$ path is the sequential concatenation of \begin{enumerate}
        \item the subpath of the $t_a$-$t_i$ shortest path in $\Pi[i,\ell]$ between $t_a$ and segment $S_i$;
        \item the $S_i$-$S_{j}$ channel; and
        \item the subpath of the $t_{j+1}$-$t_b$ shortest path in $\Pi[\ell,j]$ between segment $S_j$ and $t_b$.
    \end{enumerate} 
\end{itemize}
\end{itemize}
\end{itemize}

This completes the description of the algorithm. 
In the remainder of this step, we verify its correctness by proving the following claims.

\begin{claim}
If $\Pi[1,k]=\textsf{Null}$, then no channel and shortest path structure achieves intermediate goal.
\end{claim}

\begin{proof}
We prove by induction that for every pair $1 \le i < j \le k$, if $\Pi[i,j]=\textsf{Null}$, then no channels and shortest path structure achieve the intermediate goal within the $\Pi[i,j]$ subproblem.
For the base case where $j = i+1$, $\Pi[i,j]\ne \textsf{Null}$. Consider now a pair $(i,j)$ with $j\ge i+2$. Assume that $\Pi[i,j]=\textsf{Null}$ and for contradiction that there is a set of channels and a shortest path structure achieving the intermediate goal, then there must be an anchor index $i<\ell<j$.
However, since $\Pi[i,j]=\textsf{Null}$, either
\begin{itemize}
    \item $\Pi[i,\ell]=\textsf{Null}$, which by induction means that no channels and shortest path structure achieve the intermediate goal within $\Pi[i,\ell]$, a contradiction to the assumption that $\ell$ is anchor; or
    \item $\Pi[\ell,j]=\textsf{Null}$, which by induction means that no channels and shortest path structure achieve the intermediate goal within $\Pi[\ell,j]$, a contradiction to the assumption that $\ell$ is anchor; or
    \item there is a pair $(t_a,t_b)$ with $i<a<\ell$ and $\ell<b<j$ such that $(t_a,t_b)$ repels $(t_{\ell},t_{\ell+1})$ and also repels one of $(t_i,t_{i+1})$ and $(t_j,t_{j+1})$, which means that the $t_a$-$t_b$ shortest path cannot visit either $S_{\ell}$ or both $S_i$ and $S_j$, which is impossible given the channels $S_i$-$S_j$, $S_i$-$S_{\ell}$, and $S_{\ell}$-$S_j$.
\end{itemize}
This completes the proof by induction.
\end{proof}

Assume now that $\Pi[1,k]\ne \textsf{Null}$. For each pair $t,t'$, we denote by $\Pi_{t,t'}$ the $t$-$t'$ path provided in $\Pi[1,k]$.

\begin{claim} \label{clm:path-repel}
For every pair $t,t'$, if path $\Pi_{t,t'}$ visits $S_{x}$, then $(t, t')$ and $(t_x, t_{x+1})$ do not repel each other.
\end{claim}

\begin{proof}
Assume that path $\Pi_{t,t'}$ is created during the process of computing $\Pi[i,j]$. We prove the claim by induction. For the base case where $j = i + 1$, $x$ is either $i$ or $i+1$, and $(t,t')$ is either $(t_{i},t_{i+1})$, or $(t_{i},t_{i+2})$, or $(t_{i+1},t_{i+2})$. Therefore, $(t,t')$ and $(t_x,t_{x+1})$ share a terminal, and so they do not repel each other.
Consider now a pair $(i,j)$ with $j\ge i+2$. From the algorithm, there exists an index $i<\ell<j$ such that
\begin{itemize}
    \item $\Pi[i,\ell]\ne \textsf{Null}$, $\Pi[\ell,j]\ne \textsf{Null}$;
    \item $t\in\{t_i,\dots,t_{\ell}\}$, $t' \in \{t_{\ell+1},\dots,t_{j+1}\}$, and $i\le x\le j$; and
    \item the pair $(t,t')$ either does not repel $(t_{\ell},t_{\ell+1})$ or does not repel either $(t_i,t_{i+1})$ or $(t_j,t_{j+1})$.
\end{itemize}

From the algorithm, if $(t,t')$ does not repel $(t_{\ell},t_{\ell+1})$, then path $\Pi_{t,t'}$ is a concatenation of path $\Pi_{t,t_{\ell+1}}$ provided by $\Pi[i,\ell]$ and path
$\Pi_{{t_\ell},t'}$ provided by $\Pi[\ell,j]$. If $x=\ell$, the claim holds trivially. 
Assume that $i<x<\ell$ (the proof for case $\ell<x<j$ is symmetric),
so the path $\Pi_{t,t_{\ell+1}}$ also visits $S_x$. By induction hypothesis, $(t,t_{\ell+1})$ does not repel $(t_x,t_{x+1})$. 
If $t_x$ lies between $t_i$ and $t$, then by \Cref{obs:add}, $(t,t')$ does not repel $(t_x,t_{x+1})$. 
If $t_x$ lies between $t$ and $t_{\ell}$, then \Cref{obs:trans} $(t,t')$ does not repel $(t_x,t_{x+1})$. 
If $(t,t')$ repels $(t_{\ell},t_{\ell+1})$ but does not repel $(t_i,t_{i+1})$ or $(t_j,t_{j+1})$, then the path $\Pi_{t,t'}$ is the concatenation of the path $\Pi_{t,t_i}$ provided by $\Pi[i,\ell]$, channel $S_i$-$S_j$, and path $\Pi_{t_{j+1},t'}$ provided by $\Pi[\ell,j]$. We can show that $(t, t')$ and $(t_x, t_{x+1})$ do not repel each other in a similar way.
\end{proof}

\subsubsection*{Step 3. Verifying the repelling paths condition}

In this step, we show that if the algorithm in Step 2 does not return \textsf{Null}, then the channels and the shortest path structure provided in $\Pi[1,k]$ can be slightly modified to an outerplanar graph together with a shortest path structure that satisfy the repelling paths condition.

\begin{figure}[h]
\centering
\subfigure[One channel gives birth to three edges (in parallel). The new endpoints are close to the old channel ends.]
{\scalebox{0.12}{\includegraphics{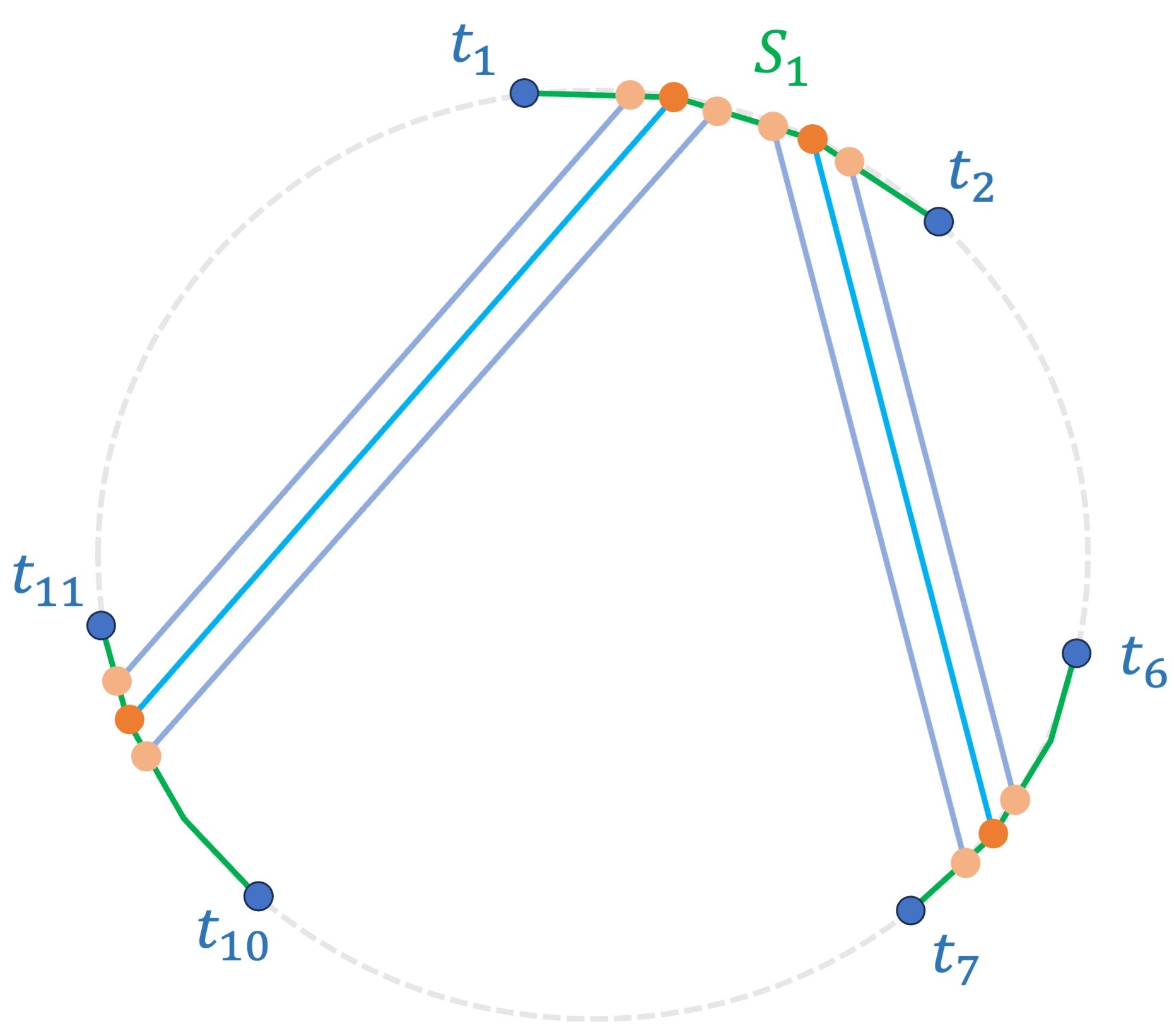}}\label{fig: split 1}}
\hspace{1.5cm}
\subfigure[Illustration of shortest paths at channel $S_1$-$S_{10}$: $t_i,t_j$ either both lie on left (red), or both lie on right (black), or lie on different segments (purple).]
{\scalebox{0.12}{\includegraphics{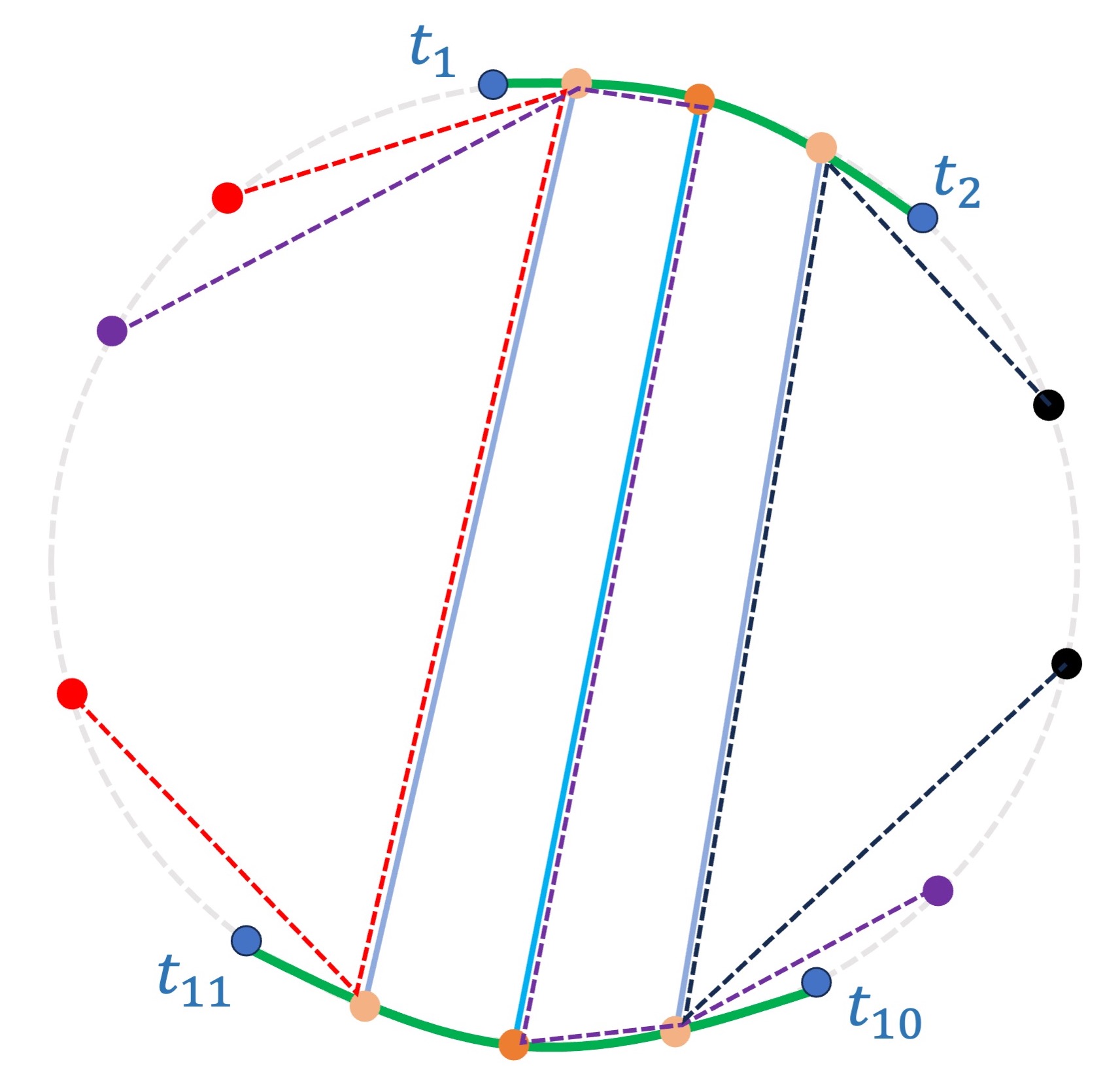}}\label{fig: split 2}}
\caption{The construction of the final outerplanar graph and the shortest path structure.}\label{fig: split}
\end{figure} 

\paragraph{The final graph.}
From the channel structure computed in Step 2, we obtain the final outerplanar graph as follows.
Consider a channel $S_i$-$S_j$. We denote by $x$ its endpoint on $S_i$ and by $y$ its endpoint on $S_j$. Around $x$ we create two new vertices: $x_L$ and $x_R$, and similarly around $y$ we create new vertices $y_L$ and $y_R$, such that vertices $x_L,x,x_R,y_R,y,y_L$ appear on the boundary in this order.
Corresponding to the channel $S_i$-$S_j$, we build three edges $(x_L,y_L), (x,y), (x_R,y_R)$ parallel to each other. We do the same for every other channel, and the resulting graph is our final outerplanar graph. So if $S_i$ has two channels connecting to $S_{j}$ and $S_{j'}$, respectively, then in the final graph, the sequence of vertices on $S_i$ is $(t_i,x_L,x,x_R,x'_L,x',x'_R,t_{i+1})$.
See \Cref{fig: split 1} for an illustration.

\paragraph{The shortest path structure.}
Recall that a terminal shortest path is characterized by the sequence of segments it visits. Consider now a $t_i$-$t_j$ shortest path and assume that it sequentially visits segments $S_{r_1},S_{r_2},\ldots,S_{r_p}$. For each channel $S_{r_x}$-$S_{r_{x+1}}$, we pick one of the three edges according to the following rules. Consider the two boundary segments complementing $S_{r_x}$ and $S_{r_{x+1}}$ (one from $t_{r_x+1}$ to $t_{r_{x+1}}$ and the other from $t_{r_{x+1}+1}$ to $t_{r_{x}}$):
\begin{itemize}
    \item if $t_i, t_j$ lies on different segments, we pick the middle edge $(x,y)$; and
    \item if $t_i, t_j$ lies on the same segments, then we pick the edge closest to that segment among the three.
\end{itemize}
The $t_i$-$t_j$ shortest path is then constructed naturally with all the picked edges: it starts at $t_i$, goes to the $S_{r_1}$-endpoint of the picked edge for the $S_{r_1}$-$S_{r_2}$ channel, hops to $S_{r_2}$, finds the $S_{r_2}$-endpoint of the picked edge for the $S_{r_2}$-$S_{r_3}$ channel, hops to $S_{r_3}$, $\ldots$ , and finally reaches $t_j$.
See \Cref{fig: split 2}.

Denote by $\Pi^*$ the final shortest path structure constructed above, where $\Pi^*_{t,t'}$ denotes the final shortest path between terminals $t,t'$.
We prove the following two claims for the required properties of the final graph and the shortest path structure.

\begin{claim} \label{clm:intersect-once}
The intersection between every pair of paths in $\Pi^*$ is either empty or a subpath of both.
\end{claim}

\begin{proof}
We prove the claim by induction on the dynamic programming table that for every $\Pi[i,j]$, the paths between terminals in $t_i,\dots,t_{j+1}$ intersect at most once.
For the base case where $j = i + 1$, the entry $\Pi[i,i+1]$ contains only three paths, and it is easy to verify that, after modification (into the final set $\Pi^*$), the intersection between every pair of them is either empty or a subpath of both.
Consider now a pair $i,j$ with $i\le j-2$.
From the algorithm, there is an anchor index $i<\ell<j$. 
Consider now terminals $t_{x_1},t_{y_1},t_{x_2},t_{y_2}$ handled in $\Pi[i,j]$ and specifically the $t_{x_1}$-$t_{y_1}$ shortest path (with $x_1<y_1$) and the $t_{x_2}$-$t_{y_2}$ shortest path (with $x_2<y_2$).  We distinguish between the following cases.

\textbf{Case 1. $y_1,y_2\le \ell$ or $x_1,x_2\ge \ell$.} By the induction hypothesis on subproblems $\Pi[i,\ell]$ and $\Pi[\ell,j]$, the final two paths intersect at most once. 

\textbf{Case 2. $y_1 \le \ell$ and $x_2 \ge \ell+1$.} The $t_{x_1}$-$t_{y_1}$ path belongs to $\Pi[i,\ell]$ and the $t_{x_2}$-$t_{y_2}$ path belongs to $\Pi[\ell,j]$. The only segment where they could potentially intersect is $S_{\ell}$. But even if both paths visit $S_{\ell}$, their final paths are disjoint as they reach $S_\ell$ from different sides.

\textbf{Case 3. $x_1,x_2 \le \ell$ and $y_1 \le \ell< y_2$.} From our algorithm, the $t_{x_2}$-$t_{y_2}$ path visits either $S_{\ell}$ or both $S_i$ and $S_j$. If it visits $S_{\ell}$, then this path is the concatenation of a subpath of the $t_{x_2}$-$t_{\ell+1}$ path and a subpath of the $t_{\ell}$-$t_{y_2}$ path. By induction hypothesis, the first subpath may intersect the $t_{x_1}$-$t_{y_1}$ path at most once, and by Case 2, the second subpath may not intersect the $t_{x_1}$-$t_{y_1}$ path. Therefore, their final paths intersect at most once. The analysis where the $t_{x_2}$-$t_{y_2}$ path visits both $S_i$ and $S_j$ is similar.

\textbf{Case 4. $y_1,y_2 \ge \ell+1$ and $x_1 \le \ell <x_2$.} Symmetric to Case 3.

\textbf{Case 5. $x_1 \le \ell < y_1$ and $x_2 \le \ell < y_2$.} In this case, both paths visit either segment $S_{\ell}$ or both $S_i,S_j$. If they visit the same segments, then by induction hypothesis they intersect once. Assume now that they visit different segments and without loss of generality that $x_1 \le x_2$. We distinguish the following sub-cases.

\textbf{Case 5.1. $y_1 \le y_2$, the $t_{x_1}$-$t_{y_1}$ path visits $S_{\ell}$, and the $t_{x_2}$-$t_{y_2}$ path visits $S_i,S_j$}. Similar to the previous discussion, both paths are the concatenation of two subpaths. By the induction hypothesis, the first subpaths intersect at most once. It remains to show that the second subpaths, one from $t_{j+1}$ to $t_{y_2}$ and the other from $t_{\ell}$ to $t_{y_1}$, are disjoint.

Assume for contradiction that they intersect on segment $S_z$.
Assume first that $S_z$ lies between $t_{j+1},t_{y_2}$ (see \Cref{fig: case5.1_1}). Since the $t_{\ell}$-$t_{y_1}$ path visits $S_z$, $(t_{\ell},t_{y_1})$ does not repel $(t_z,t_{z+1})$. Also, since the $t_{x_1}$-$t_{y_1}$ path visits $S_{\ell}$, $(t_{x_1},t_{y_1})$ does not repel $(t_{\ell},t_{\ell+1})$. From \Cref{obs:sub}, $(t_{x_1},t_z)$ does not repel $(t_{\ell},t_{\ell+1})$. Then from \Cref{obs:add}, $(t_{x_2},t_{y_2})$ does not repel $(t_{\ell},t_{\ell+1})$. But this means that our algorithm should have routed the $t_{x_2}$-$t_{y_2}$ path through $S_{\ell}$ rather than $S_i$ and $S_j$, a contradiction.
Assume next that $S_z$ lies between $t_{\ell+1}$ and $t_{y_1}$ (see \Cref{fig: case5.1_2}). Since $(t_{x_1},t_{y_1})$ does not repel $(t_{\ell},t_{\ell+1})$, by \Cref{obs:add}, $(t_{x_2},t_{z+1})$ does not repel $(t_{\ell},t_{\ell+1})$. On the other hand, by \Cref{clm:path-repel}, $(t_{x_2},t_{y_2})$ does not repel $(t_z,t_{z+1})$. Thus, from \Cref{obs:trans}, $(t_{x_2},t_{y_2})$ does not repel $(t_{\ell},t_{\ell+1})$ which again means that we should have routed the $t_{x_2}$-$t_{y_2}$ path through $S_{\ell}$ rather than $S_i$ and $S_j$, a contradiction.
Finally, we assume that $S_z$ lies between $t_{y_1}$ and $t_{y_2}$, from the algorithm, in this case the two subpaths may not intersect on $S_z$.

\begin{figure}[h]
\centering
\subfigure[When $S_z$ lies between $t_{j+1}$ and $t_{y_2}$.]
{\scalebox{0.12}{\includegraphics{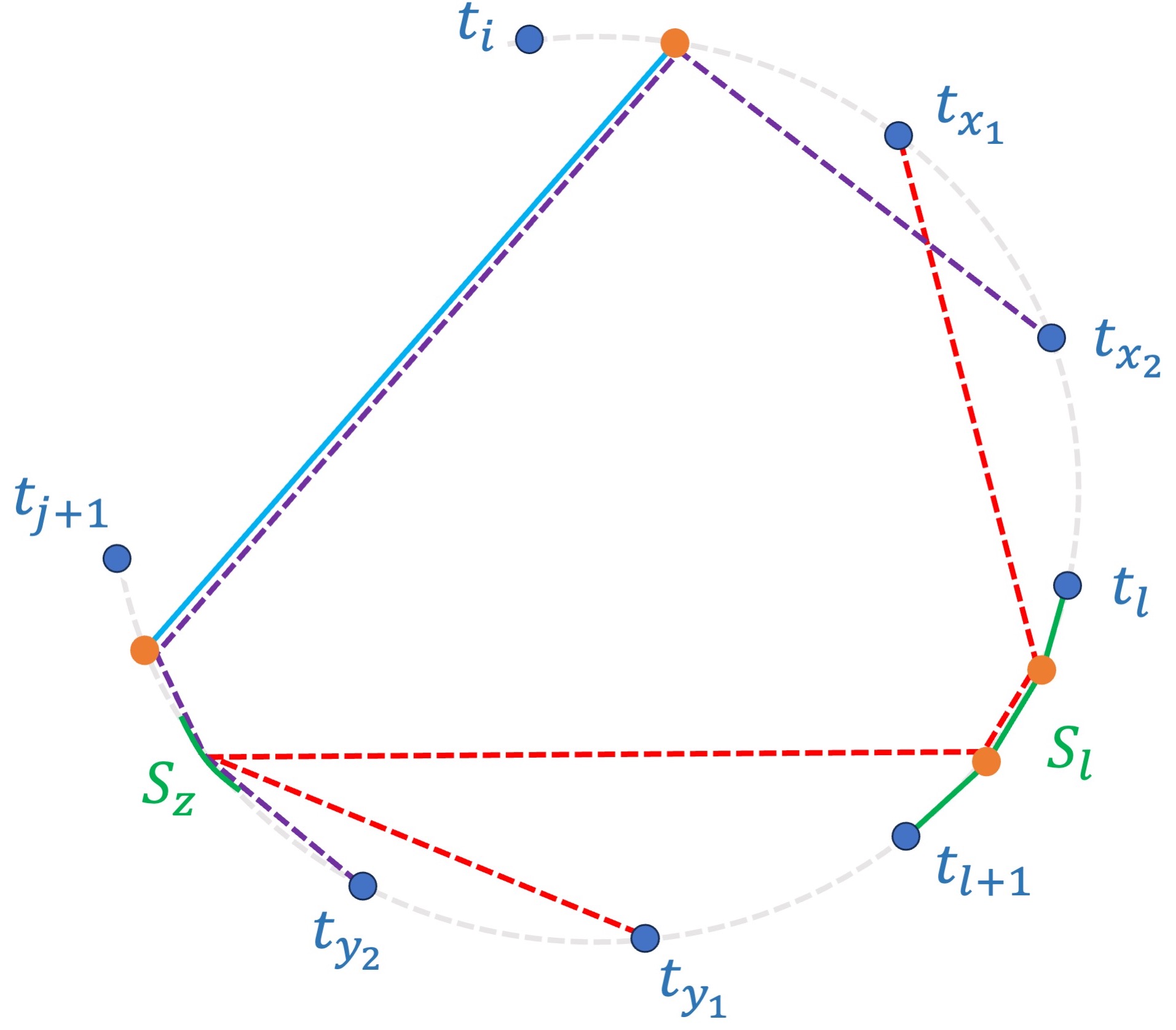}}\label{fig: case5.1_1}}
\hspace{0.5cm}
\subfigure[When $S_z$ lies between $t_{\ell+1}$ and $t_{y_1}$.]
{\scalebox{0.12}{\includegraphics{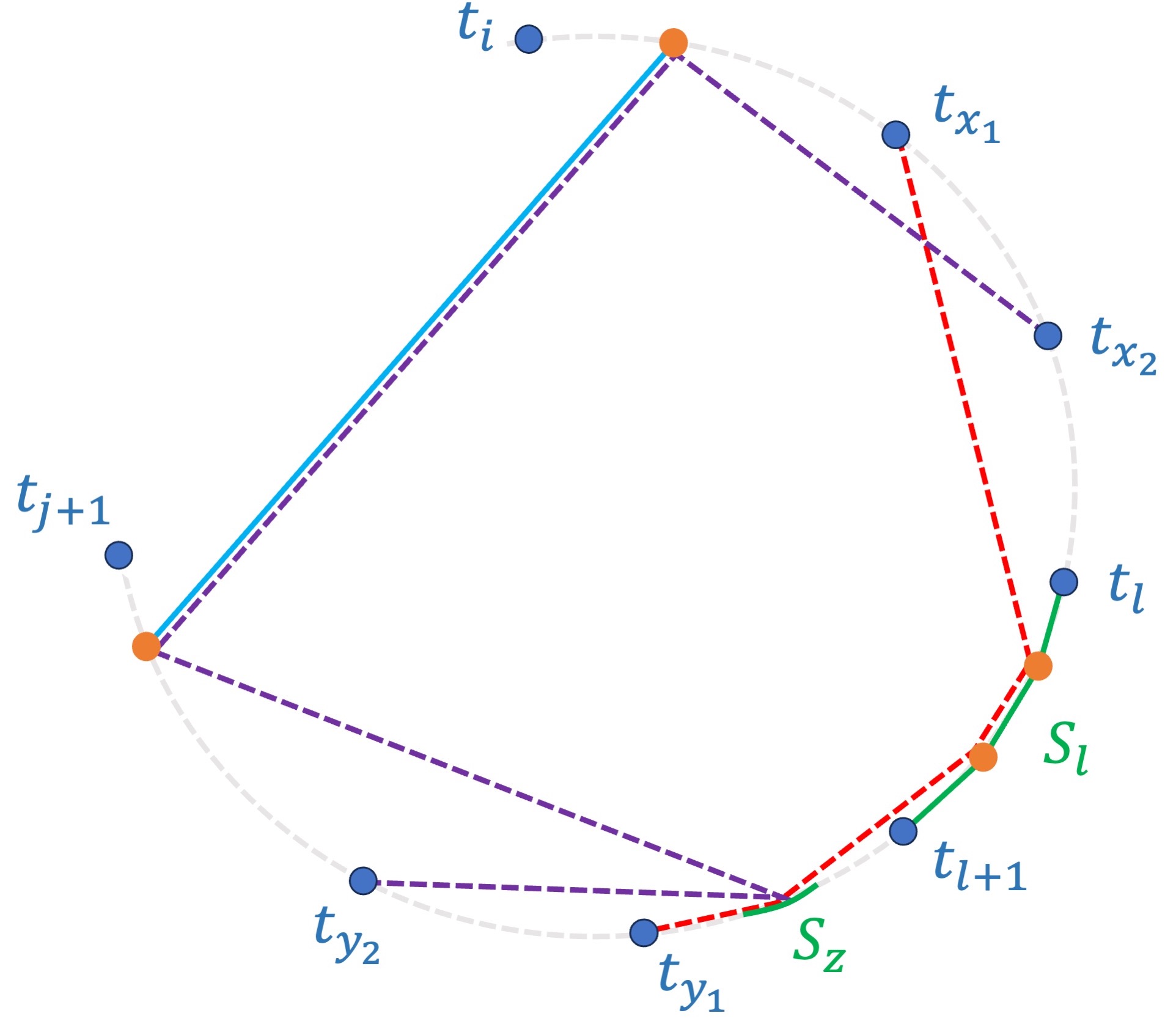}}\label{fig: case5.1_2}}
\caption{An illustration of Case 5.1.}\label{fig: case 5}
\end{figure}

\textbf{Case 5.2. $y_1 \le y_2$, the $t_{x_1}$-$t_{y_1}$ path visits $S_i,S_j$, and the $t_{x_2}$-$t_{y_2}$ path visits $S_{\ell}$}.
Same as 5.1.

\textbf{Case 5.3. $y_1 \ge y_2$.} 
In this case, the only possibility is that the $t_{x_1}$-$t_{y_1}$ path visits $S_i$ and $S_j$, while the $t_{x_2}$-$t_{y_2}$ path visits $S_{\ell}$, since otherwise we can get a contradiction by \Cref{obs:add}. 
As both paths are the concatenation of two subpaths, by induction hypothesis, the two first subpaths may intersect once, and the two second subpaths may intersect once. We show that the intersections cannot both happen. 
Assume for contradiction that the $t_{x_1}$-$t_i$ subpath and the $t_{x_2}$-$t_{\ell+1}$ subpath intersect at a segment $S_z$, and the $t_{j+1}$-$t_{y_2}$ subpath and the $t_{\ell}$-$t_{y_1}$ subpath intersect at a segment $S_w$.

\textbf{Case 5.3.1. $S_z$ lies between $t_{x_2}$ and $t_{\ell}$, and $S_w$ lies between $t_{\ell+1}$ and $t_{y_2}$.} In this case, $(t_{x_2},t_{y_2})$ does not repel $(t_{\ell},t_{\ell+1})$. By \Cref{obs:add}, $(t_z,t_{w+1})$ does not repel $(t_{\ell},t_{\ell+1})$. Since $(t_{x_1},t_{y_1})$ does not repel $(t_z,t_{z+1})$, \Cref{obs:add} implies that $(t_{x_1},t_{w+1})$ does not repel $(t_z,t_{z+1})$. Then by \Cref{obs:trans}, $(t_{x_1},t_{w+1})$ does not repel $(t_{\ell},t_{\ell+1})$. Since $(t_{x_1},t_{y_1})$ does not repel $(t_w,t_{w+1})$, \Cref{obs:trans} also implies that $(t_{x_1},t_{y_1})$ does not repel $(t_{\ell},t_{\ell+1})$. From our algorithm, we should have routed the $t_{x_1}$-$t_{y_1}$ path via $S_{\ell}$, a contradiction.

\textbf{Case 5.3.2. $S_z$ lies between $t_{x_2}$ and $t_{\ell}$, and $S_w$ lies between $t_{j+1}$ and $t_{y_1}$.} In this case, $(t_{x_2},t_{y_2})$ does not repel $(t_{\ell},t_{\ell+1})$, and $(t_{\ell},t_{y_2})$ does not repel $(t_w,t_{w+1})$. From \Cref{obs:sub}, $(t_{x_2},t_w)$ does not repel $(t_{\ell},t_{\ell+1})$, then by \Cref{obs:add}, $(t_z,t_w)$ does not repel $(t_{\ell},t_{\ell+1})$. On the other hand, since $(t_{x_1},t_i)$ does not repel $(t_z,t_{z+1})$, by \Cref{obs:add}, $(t_{x_1},t_w)$ does not repel $(t_z,t_{z+1})$, then by \Cref{obs:trans}, $(t_{x_1},t_w)$ does not repel $(t_{\ell},t_{\ell+1})$, and finally by \Cref{obs:add}, $(t_{x_1},t_{y_1})$ does not repel $(t_{\ell},t_{\ell+1})$. But this means that our algorithm should have routed the $t_{x_1}$-$t_{y_1}$ path via $S_{\ell}$, a contradiction.

\textbf{Case 5.3.3. $S_z$ lies between $t_{i+1}$ and $t_{x_1}$.} We can further consider two subcases: the case where $S_w$ lies between $t_{\ell+1}$ and $t_{y_2}$ and the case where $S_w$ lies between $t_{j+1}$ and $t_{y_1}$. A contradiction can be derived in the same way as Case 5.3.1 and Case 5.3.2. 
\end{proof}

\begin{claim} \label{clm:repel}
The constructed shortest path structure satisfies all the repelling paths conditions: if pairs $(t_1,t_2)$ and $(t'_1,t'_2)$ repel each other, then the $t_1$-$t_2$ shortest path and the $t'_1$-$t'_2$ shortest path are disjoint.
\end{claim}

\begin{proof}
If the two pairs of terminals cross each other (they appear on the boundary in order $t_1,t'_1,t_2,t'_2$), then from the four-point condition of Okamura-Seymour instances, the two pairs $(t_1,t_2),(t'_1,t'_2)$ cannot repel each other.
Therefore, we only need to consider the case where they appear on the boundary in order $t_1,t_2,t'_2,t'_1$.
Consider the $t_1$-$t_2$ path and the $t'_1$-$t'_2$ path given in $\Pi[1,k]$.
If they intersect only on the $t_1$-to-$t'_1$ boundary segment and the $t_2$-to-$t'_2$ boundary segment, then from the construction of the final graph and our rules of picking the edges, their final shortest paths will be disjoint. If not, we may assume without loss of generality that both paths visit the same segment $S_i$ on the $t_1$-to-$t_2$ boundary segment. From \Cref{clm:path-repel}, $(t'_1,t'_2)$ does not repel $(t_i,t_{i+1})$, and then by \Cref{obs:add}, $(t'_1,t'_2)$ does not repel $(t_1,t_2)$, a contradiction.
\end{proof}

\paragraph{Running time.} There are a total of $O(k^2)$ dynamic programming entries $\Pi[i,j]$ to compute. For each entry, we enumerate all $\ell$ within interval $[i,j]$, and for every pair $a,b$, we check whether the pair $(t_a,t_b)$ repels $(t_{\ell},t_{\ell+1})$ or $(t_{i},t_{i+1})$ and $(t_{j},t_{j+1})$, which takes $O(k^3)$ time. So this simple analysis shows that the running time of the dynamic programming algorithm is $O(k^5)$. We remark that the tasks of repelling pairs testing can be better organized globally so that in total they only take time $O(k^4)$.

If the algorithm returns a graph and a shortest path structure that satisfies all repelling paths conditions, we will solve the LP in \Cref{sec: repelling paths} for finding proper edge weights, whose existence is guaranteed by Theorem~\ref{lem: 2-path condition}. The LP has $O(k)$ variables (as the constructed outerplanar graph has $O(k)$ edges) and $2^{\Omega(k)}$ constraints, but can be solved by the Ellipsoid algorithm with a separation oracle that computes all-pairs shortest paths (taking $O(k^2)$ time for an outerplanar graph). Therefore, the running time of solving the LP is $O(k^2)\cdot O(k)\cdot O(k^2)=O(k^5)$. Together with the discussion at the end of \Cref{sec: ordering}, the running time of our algorithm is $O(k^5)$.

\section{Minimum Realization of OS Instances: Proof of Theorem 4}\label{sec:minimum-proof}

\subsection{High level overview}

We now prove Theorem~\ref{thm: minimum main}. Throughout this overview, we assume that Theorem~\ref{lem: 2-path condition} has been proved. Thus, for a candidate OS graph, the numerical problem of assigning edge lengths is settled once we can construct a $D$-good shortest path structure. The proof has two main ingredients: a cut characterization for the existence of such a shortest path structure, and a medial-graph construction that satisfies all of the cut requirements with the fewest crossings.

\paragraph{Step 1. Characterize good shortest path structures by chain capacities.}
Consider a boundary cut $x$--$y$ and a repelling set $M$ of terminal pairs crossing it. In a $D$-good shortest path structure, the paths for the pairs in $M$ are mutually vertex-disjoint. Every one of these paths must meet every chain joining the two sides of the cut, and distinct paths must use distinct vertices of the chain. Therefore, a necessary condition is
\[
|M|\le |A|
\]
for every repelling set $M$ and every chain $A$ crossed by all pairs in $M$.

Section~\ref{sec: repelling path theorem} proves that this necessary condition is also sufficient. The proof recursively constructs a path $\Pi[i,j]$ for every cyclic interval of terminals. If $(i,j)$ repels $(i+1,j-1)$, the path $\Pi[i,j]$ is routed immediately outside $\Pi[i+1,j-1]$; otherwise it is formed by splicing $\Pi[i,j-1]$ and $\Pi[i+1,j]$ at an intersection. Every vertex of every constructed path is accompanied by a tight certificate: a chain ending at that vertex and a repelling set of the same cardinality crossing the chain. If the clockwise and counterclockwise constructions for some terminal pair overlap in the wrong way, two such certificates combine to produce a chain crossed by more repelling pairs than its length. This contradicts the assumed capacity inequalities. The remaining paths form a $D$-good shortest path structure, and hence Theorem~\ref{lem: 2-path condition} assigns realizing edge lengths.

\paragraph{Step 2. Translate the capacities to a medial template.}
For each boundary cut $x$--$y$, let
\[
a_{x,y}=\max\{|M|:M\text{ is a repelling set crossing the $x$--$y$ cut}\}.
\]
The chain characterization requires every $x$--$y$ chain to contain at least $a_{x,y}$ primal vertices. A regular neighborhood of a chain of length $r$ is crossed by $2r$, $2r-1$, or $2r-2$ medial strands, according as neither, one, or both boundary endpoints are terminals. Accordingly, define
\[
b_{x,y}=2a_{x,y}-\mathbf 1[x\in T]-\mathbf 1[y\in T].
\]
The medial template of every OS realization of $D$ must have at least $b_{x,y}$ chords crossing the $x$--$y$ cut. We show that these lower bounds are simultaneously tight. The values $b_{x,y}$ satisfy a discrete Monge inequality on consecutive boundary reference points. Taking circular second differences recovers the indicator that one medial boundary point is matched to another. For every boundary point, exactly one second difference is equal to one and all the others are zero. Hence the cut values determine a unique perfect matching of the medial boundary endpoints, denoted by $\Phi(D)$.

\paragraph{Step 3. Prove minimality and recover the primal realizations.}
Any feasible medial template dominates $\Phi(D)$ on every circular cut. The chord-uncrossing result proved in Section~\ref{sec: repelling paths} shows that strict dominance can be reduced by uncrossing crossing chord pairs. Every nontrivial uncrossing decreases the number of crossing pairs, so $\Phi(D)$ is the unique feasible template of minimum crossing number. Since each medial crossing corresponds to one primal edge, this crossing number is the minimum number of edges in an OS realization.

Finally, take any arrangement of the chords of $\Phi(D)$ and form its primal graph. Its medial cut counts meet all the chain-capacity lower bounds, so Section~\ref{sec: repelling path theorem} constructs a $D$-good shortest path structure in the primal graph. Theorem~\ref{lem: 2-path condition} then computes nonnegative edge lengths realizing $D$. Conversely, the medial graph of every minimum realization is an arrangement of $\Phi(D)$. This gives the compact representation of all minimum graph structures asserted in Theorem~\ref{thm: minimum main}.

\subsection{Existence of Good Shortest Path Structures}
\label{sec: repelling path theorem}

The main result in this section is the following theorem.

\begin{theorem}
\label{thm:repelling-paths}
Given a metric $D$ on $T$, an OS instance $(G,T)$ admits a $D$-good shortest path structure iff for each repelling set $\mset$ and each chain $A$, the number of pairs in $\mset$ that cross $A$ is at most $|A|$. Moreover, such a $D$-good shortest path structure, if it exists, can be computed efficiently.
\end{theorem}

For the ``only if'' direction: Consider a metric $D$ on $T$ and an OS instance $(G,T)$ that admits a $D$-good shortest path structure. Consider any repelling set $\mset$ and any chain $A$. If the number of pairs in $\mset$ that cross $A$ is greater than $|A|$, then (i) from the definition of a $D$-good shortest path structure, the shortest paths connecting these pairs are mutually vertex-disjoint; and (ii) since these pairs cross $A$, each of their shortest paths contains a distinct vertex from $A$, a contradiction.

The remainder of this section is dedicated to the proof of the ``if'' direction. Assume that we are given a metric $D$ on $T$, an OS instance $(G,T)$, such that for each repelling set $\mset$ and each chain $A$, the number of pairs in $\mset$ that cross $A$ is at most $|A|$. We will construct a $D$-good shortest path structure.

For simplicity, we denote the terminals by $1,\ldots,n$, that appear on the disc boundary clockwise in this order. For a pair $i,j$ of terminals, we denote by $\partial[i,j]$ the boundary segment from $i$ clockwise to $j$ (so $i,j$ partition the boundary into two segments: $\partial[i,j]$ and $\partial[j,i]$). 

We now construct, inductively, for each pair $i,j$ of terminals, a path $\Pi[i,j]$ connecting $i$ to $j$ as follows. The base case is when $j=i+1$, and we simply let $\Pi[i,j]$ be the path characterizing the boundary of the graph between $i,j$, so $\Pi[i,j]$ and $\partial[i,j]$ surround a region with no vertices or edges in its interior. Consider now a general pair $i,j$ with $j\ge i+2$. We distinguish between the following two cases.

\textbf{Case 1: the pairs $(i,j)$ and $(i+1,j-1)$ repel each other.}
In this case, we let $\Pi[i,j]$ be the path ``right next to'' the path $\Pi[i+1,j-1]$. Specifically, consider all regions that (i) lie outside the area enclosed by $\Pi[i+1,j-1]$ and $\partial[i+1,j-1]$; and (ii) are incident with a vertex of $\Pi[i+1,j-1]$. 
The union of all such regions forms a strip between boundary segments $\partial[i,i+1]$ and $\partial[j-1,j]$, with one side being $\Pi[i+1,j-1]$ and the other side being a walk starting at $i$ and ending at $j$. 
We simply remove all the cycles in this walk and let $\Pi[i,j]$ be the resulting simple path from $i$ to $j$. See \Cref{fig: path} for an illustration.

\textbf{Case 2: the pairs $(i,j)$ and $(i+1,j-1)$ do not repel each other.} We have already constructed paths $\Pi[i,j-1]$ and $\Pi[i+1,j]$ inductively. Since terminals $i,i+1,j-1,j$ appear on the boundary in this order, paths $\Pi[i,j-1]$ and $\Pi[i+1,j]$ must intersect, say at a vertex $p$. We simply let $\Pi[i,j]$ be the concatenation of the subpath of $\Pi[i,j-1]$ between $i$ and $p$ and the subpath of $\Pi[i+1,j]$ between $p$ and $j$. 

From the description of the algorithm, it is easy to verify that in both cases, we can compute the path $\Pi[i,j]$ efficiently. Since the number of terminal pairs $(i,j)$ is $O(k^2)$, our algorithm for constructing all paths $\set{\Pi[i,j]}_{i,j}$ is efficient.

\begin{figure}[h]
\centering
\subfigure
{\scalebox{0.11}{\includegraphics{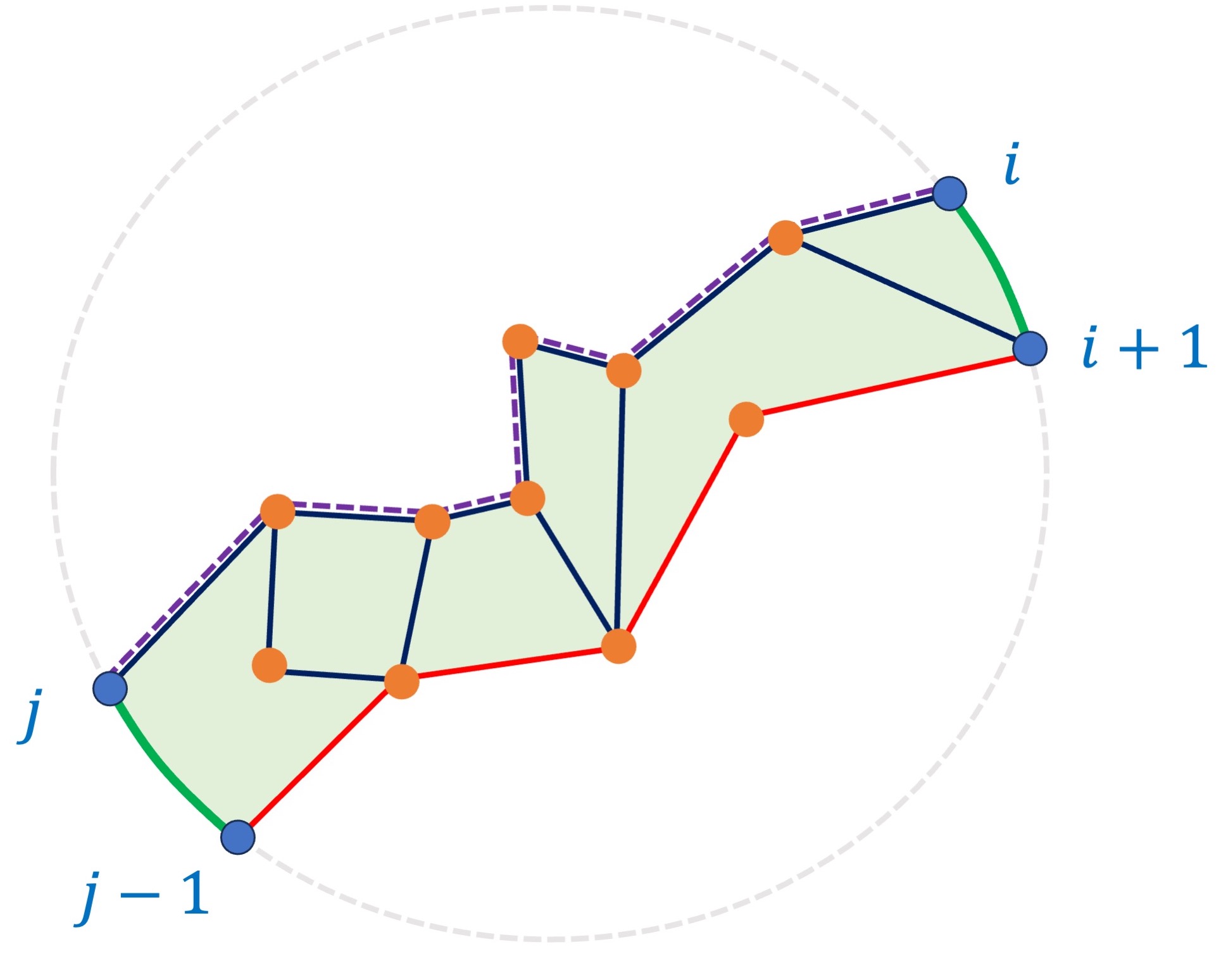}}}
\hspace{1.0cm}
\subfigure
{\scalebox{0.11}{\includegraphics{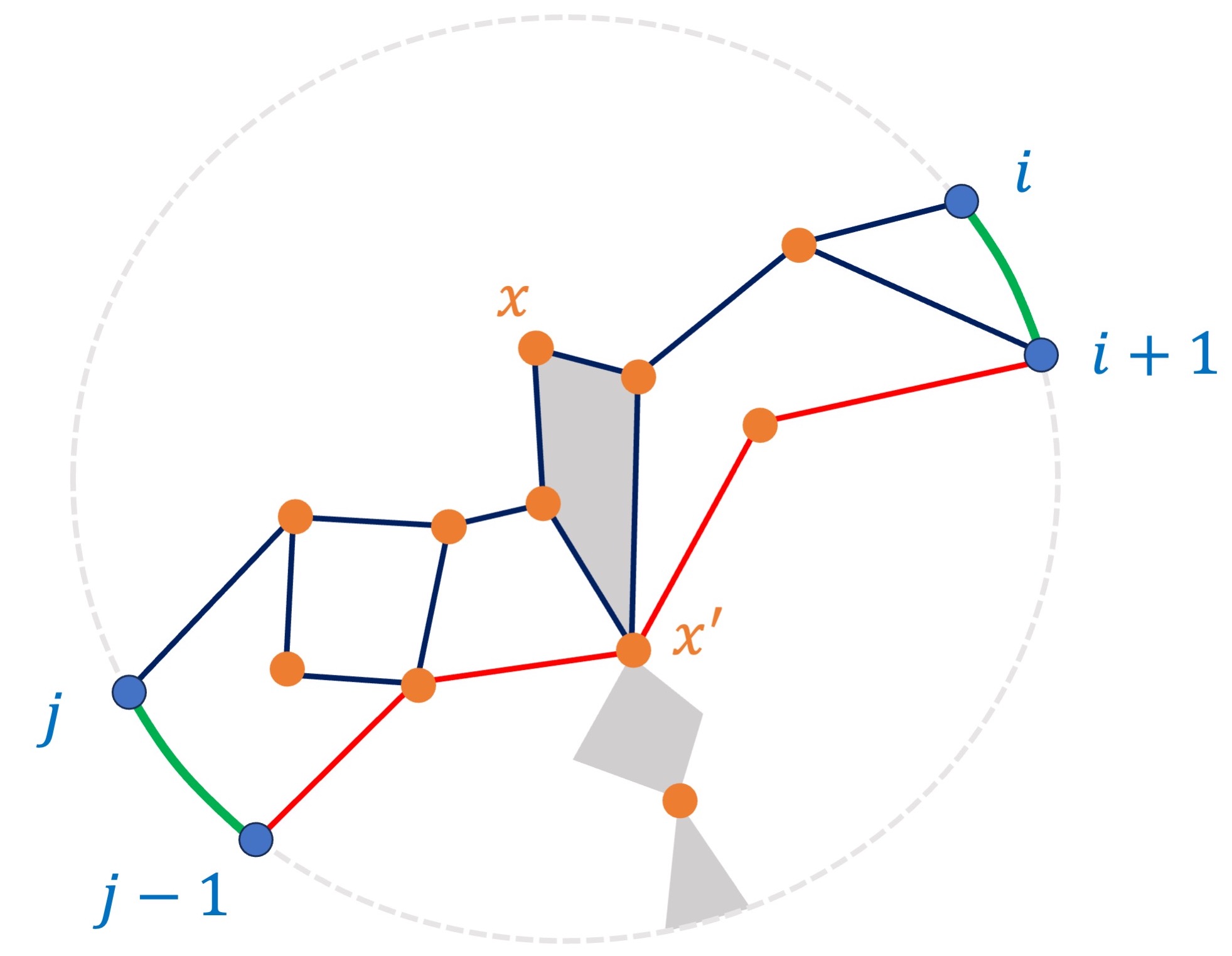}}}
\caption{Path $\Pi[i,j]$ in Case 1 (purple dashed line on the left): the other side of the the strip (light green area) supported by path $\Pi[i+1,j+1]$ (red), and the chain from vertex $x$ built inductively (gray sequence on the right).}\label{fig: path}
\end{figure}

Denote by $G[i,j]$ the subgraph of $G$ inside the area enclosed by $\partial[i,j]$ and $\Pi[i,j]$ (boundary included), so $G[i,j]$ is also an OS instance with terminals $i,i+1,\ldots,j-1,j$ appearing on the boundary in this order.
We prove the following crucial claim.

\begin{claim}
\label{clm: tight cut}
For each pair $i,j$ of terminals and each vertex $x$ on path $\Pi[i,j]$, there is (i) a repelling set $\mset$ of pairs among terminals in $\set{i,\ldots,j}$; and (ii) a chain $A$ from $x$ to $\partial[i,j]$ in $G[i,j]$, such that all pairs in $\mset$ cross $A$, and $|\mset|=|A|$.
\end{claim}
\begin{proof}
We prove by induction. The base case is when $j=i+1$, and the claim trivially holds since for each vertex $x$ on $\Pi[i,j]$, the (unique) region incident to the boundary segment $\partial[i,j]$ together with $x$ gives a chain, and the singleton set $\mset=\set{(i,j)}$ satisfies the conditions in the claim with this chain.
Consider now a general pair $i,j$ with $j\ge i+2$. We distinguish between the following two cases.

Case 1: the pairs $(i,j)$ and $(i+1,j-1)$ repel each other. 
Consider a vertex $x$ on path $\Pi[i,j]$. From our construction of $\Pi[i,j]$, there must be a vertex $x'$ on $\Pi[i+1,j-1]$ that is co-facial with $x$ (say that they are both incident with face $R$).
By induction hypothesis, there is a chain $A'$ from $x'$ to $\partial[i+1,j-1]$ and a set $\mset'$ of pairs among terminals in $\set{i+1,\ldots,j-1}$, such $\mset'$ cross $A'$ and $|\mset'|=|A'|$. Since $(i,j)$ repels $(i+1,j-1)$, $(i,j)$ also repels every pair in $\mset'$. We simply let $\mset=\mset'\cup \set{(i,j)}$, and augment $A'$ with region $R$ and vertex $x$ at its end (increasing its length by $1$). It is easy to verify that the new $A$ and $\mset$ satisfy the conditions. See \Cref{fig: path} for an illustration.

Case 2: the pairs $(i,j)$ and $(i+1,j-1)$ do not repel each other. From our construction, every vertex $x$ in $\Pi[i,j]$ belongs to either $\Pi[i+1,j]$ or $\Pi[i,j-1]$. Assume without loss of generality that $x\in \Pi[i+1,j]$. By induction hypothesis, there is a chain $A$ from $x$ to $\partial[i+1,j]$ and a set $\mset$ of pairs among terminals in $\set{i+1,\ldots,j}$, such that $\mset$ cross $A$ and $|\mset|=|A|$. Clearly, $\mset$ and $A$ also satisfy the conditions for the terminal pair $i,j$.
\end{proof}

Consider now a pair $i,j$. We have constructed two paths connecting them: $\Pi[i,j]$ and $\Pi[j,i]$. We say that the pair $i,j$ is \emph{bad}, iff the area enclosed by $\Pi[i,j]$ and $\partial[i,j]$ and the area enclosed by $\Pi[j,i]$ and $\partial[j,i]$ intersect in their interiors.

\begin{claim}
If there is a bad pair, then there is a repelling set $\mset$ and a chain $A$ such that all pairs in $\mset$ cross $A$ and $|\mset|>|A|$.
\end{claim}
\begin{proof}
Let $i,j$ be a bad pair, so $\Pi[i,j]\ne \Pi[j,i]$ (since otherwise the $\Pi[i,j]$-$\partial[i,j]$ area and the $\Pi[j,i]$-$\partial[j,i]$ area intersect only at the boundary).
Consider now a vertex $x$ in $\Pi[i,j]$ that lies in the interior of the $\Pi[j,i]$-$\partial[j,i]$ area. By \Cref{clm: tight cut}, there is a repelling set $\mset$ of pairs of terminals in $\set{i,i+1,\ldots,j}$ and a chain $A$ from $x$ to $\partial[i,j]$, such that $\mset$ cross $A$ and $|\mset|=|A|$. Since the $\Pi[i,j]$-$\partial[i,j]$ area and the $\Pi[j,i]$-$\partial[j,i]$ area intersect in their interiors, chain $A$ must contain internally a vertex of path $\Pi[j,i]$, which we denote by $x'$.
By \Cref{clm: tight cut}, there is a repelling set $\mset'$ of pairs of terminals $\set{j,j+1,\ldots,i}$ and a chain $A'$ from $x'$ to $\partial[j,i]$, such that $\mset'$ cross $A'$ and $|\mset'|=|A'|$. See \Cref{fig: violating_cut} for an illustration.

Consider now a chain $A^*$ defined as the concatenation of chain $A'$ and the subchain of $A$ from $x'$ to $\partial[i,j]$ (glued at $x'$). Since $x\ne x'$, $|A^*|\le |A|+|A'|-2$. On the other hand, consider a set $\mset^*$ defined as follows: let $(\bar i,\bar j)$ be the pair in $\mset$ that is closest to $(i,j)$ (since all pairs in $\mset$ cross $A$, such a pair is well-defined), and we let $\mset^*=\big(\mset\setminus \set{(\bar i, \bar j)}\big)\cup \mset'$. It is easy to verify that the pairs in $\mset^*$ repel each other, and they all cross $A^*$. Moreover, 
\[
|\mset^*|= |\mset|+|\mset'|-1= |A|+|A'|-1 > |A|+|A'|-2 = |A^*|. 
\]
This completes the proof of the claim.
\end{proof}

\begin{figure}[h]
\centering
\includegraphics[width=0.45\linewidth]{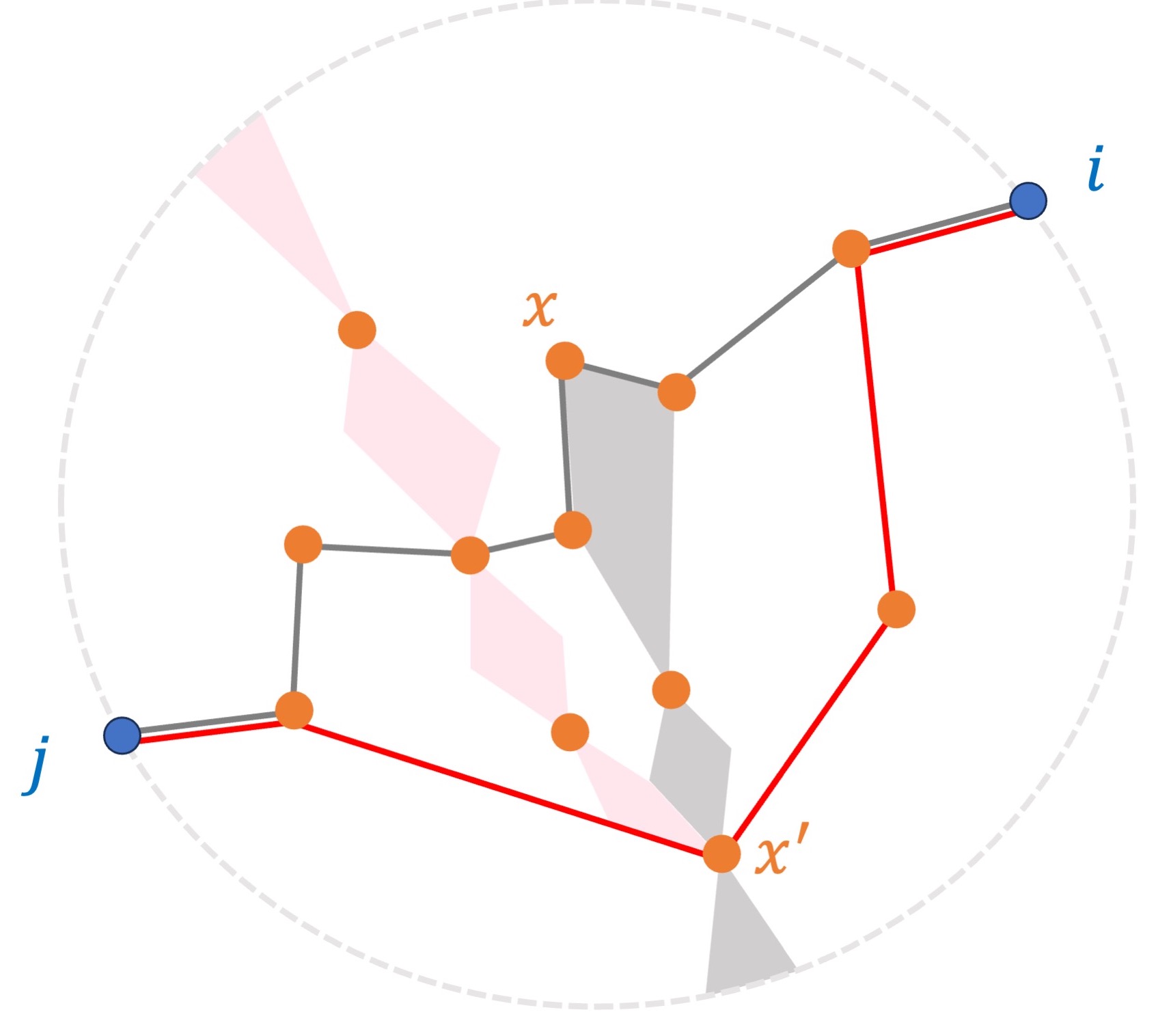}
\caption{Path $\Pi[i,j]$ (black), chain $A$ (gray), path $\Pi[j,i]$ (red) and chain $A'$ (pink).\label{fig: violating_cut}}
\end{figure}

In the case where there is no bad pair, we construct a good shortest path structure as follows. If $n$ is odd, then for every pair $i,j$ of terminals, one of the two segments $\partial[i,j],\partial[j,i]$ contains strictly fewer terminals than the other, say $\partial[i,j]$, and we let $P_{i,j}=\Pi[i,j]$.
If $n$ is even, then we arbitrarily pick a non-terminal point $x$ lying on the boundary (so $x\in X\setminus T$). So for each pair $i,j$, exactly one of the two segments $\partial[i,j],\partial[j,i]$ contains $x$, say $\partial[i,j]$, and we let $P_{i,j}=\Pi[i,j]$.
It remains to show that the constructed shortest path structure is good. 

On the one hand, we show that for all repelling pairs $(i,j), (i',j')$, the paths $P_{i,j}$ and $P_{i',j'}$ are vertex-disjoint. 
We only need to consider the case where $P_{i,j}=\Pi[i,j]$, $P_{i',j'}=\Pi[i',j']$, and terminals  $i',j'$ both lie on $\partial[i,j]$. 
We prove by induction on $|i-i'|+|j-j'|$. The base case is when $|i-i'|=|j-j'|=1$, that is, $i'=i+1$ and $j'=j-1$. Since $(i,j)$ repels $(i',j')$, from our construction, $\Pi[i,j]$ is vertex-disjoint from $\Pi[i+1,j-1]$.
Consider now a general pair $i',j'$.
Assume first that $(i,j)$ repels $(i+1,j-1)$. By our construction, $\Pi[i,j]$ is vertex-disjoint from $\Pi[i+1,j-1]$, and since $\Pi[i',j']$ lies in the area enclosed by $\Pi[i+1,j-1]$ and $\partial[i+1,j-1]$, $\Pi[i,j]$ is vertex-disjoint from $\Pi[i',j']$. Assume next that $(i,j)$ does not repel $(i+1,j-1)$. From \Cref{lem: W-repel}, since $(i,j)$ repels $(i',j')$, either $(i+1,j)$ repels $(i',j')$, in which case we can derive that $\Pi[i,j],\Pi[i',j']$ are vertex-disjoint from the induction hypothesis that $\Pi[i+1,j],\Pi[i',j']$ are disjoint; 
or $(i,j)$ repels $(i+1,j')$, in which case we can derive that 
$\Pi[i,j],\Pi[i',j']$ are vertex-disjoint from the induction hypothesis that $\Pi[i,j],\Pi[i+1,j']$ are disjoint and the fact that $\Pi[i',j']$ lies in the area enclosed by $\Pi[i+1,j']$ and $\partial[i+1,j']$.

On the other hand, we show that the intersection between every pair $P_{i,j}$ and $P_{i',j'}$ of paths is either empty or a subpath of both. We first assume that $P_{i,j}=\Pi[i,j]$, $P_{i',j'}=\Pi[i',j']$, and $i',j'$ both lie on $\partial[i,j]$. Under this assumption, it suffices to consider the case where $i=i'$ and $j'=j-1$ and show that the intersection between $\Pi[i,j]$ and $\Pi[i,j-1]$ is a path.

First we prove the following observation.
\begin{observation}
\label{obs: subpath}
For every $i<j$,
$\Pi[i,j]\cap \Pi[i,j-1]$
is a starting subpath of both $\Pi[i,j]$ and $\Pi[i,j-1]$, and symmetrically
$\Pi[i,j]\cap \Pi[i+1,j]$
is an ending subpath of both.    
\end{observation}
\begin{proof}
We only prove the first half, and the second half can be proved symmetrically.

If the path 
$\Pi[i,j]$ was constructed in Case~2, from the algorithm,
$\Pi[i,j]=\Pi[i,j-1][i,p]\cup \Pi[i+1,j][p,j]$
where $\Pi[i+1,j][p,j]$ is internally disjoint from $\Pi[i,j-1]$.
Hence
$\Pi[i,j]\cap \Pi[i,j-1]=\Pi[i,j-1][i,p]$,
which is a starting subpath.

It remains to consider Case~1.  Then $\Pi[i,j]$ is the outer side of
the strip supported by $\Pi[i+1,j-1]$.  The path
$\Pi[i,j-1]$ starts at the endpoint $i$ of this outer side.  Suppose
that, after leaving $\Pi[i,j]$, the path $\Pi[i,j-1]$ meets
$\Pi[i,j]$ again.  Let $u,v$ be two consecutive common vertices with
this property, ordered along $\Pi[i,j-1]$.  Then the two internally
disjoint subpaths
$\Pi[i,j-1][u,v]$ and $\Pi[i,j][u,v]$
bound a disc whose interior contains no boundary point of $X$.  This
disc lies on the outer side of the strip.  Replacing the segment
$\Pi[i,j][u,v]$ of the outer boundary of the strip by
$\Pi[i,j-1][u,v]$ would therefore give a boundary point of the same
union of regions lying strictly outside the chosen outer side, a
contradiction to the definition of $\Pi[i,j]$ as the outer side of the
full strip.  Thus $\Pi[i,j-1]$ cannot leave and later re-enter
$\Pi[i,j]$, and so $\Pi[i,j]\cap \Pi[i,j-1]$ is an initial subpath.
\end{proof}

We now prove the stated intersection property for nested boundary
intervals.  Suppose that
$i\le i'\le j'\le j$
in the clockwise order on $\partial[i,j]$.  We prove by induction on
$(i'-i)+(j-j')$
that $\Pi[i,j]\cap \Pi[i',j']$ is empty or a subpath of both.  The base
case $(i',j')=(i,j)$ is trivial.  If $j'<j$, then by construction
$\Pi[i',j']$ is contained in the closed region bounded by
$\Pi[i,j-1]\cup \partial[i,j-1]$. By \Cref{obs: subpath}, the part of
$\Pi[i,j]$ outside this region is disjoint from $\Pi[i,j-1]$, and
hence is disjoint from $\Pi[i',j']$.  Therefore
\[
    \Pi[i,j]\cap \Pi[i',j']
    =
    \bigl(\Pi[i,j]\cap \Pi[i,j-1]\bigr)\cap \Pi[i',j'] .
\]
The first term on the RHS is a subpath of $\Pi[i,j-1]$,
and by the induction hypothesis
$\Pi[i,j-1]\cap \Pi[i',j']$ is a subpath of $\Pi[i,j-1]$.  The
intersection of two subpaths of a path is again empty or a subpath.
Hence $\Pi[i,j]\cap \Pi[i',j']$ is empty or a subpath.  The remaining
case $j'=j$ and $i'>i$ is identical using \Cref{obs: subpath}.

Consider now the remaining case where the terminals $i,i',j,j'$ appear on the boundary in this order.
By planarity, the two paths must meet.  Let $u$ and $v$ be the first
and last common vertices encountered along $\Pi[i,j]$.  If the common
vertices did not form the subpath $\Pi[i,j][u,v]$, then there exists a pair of
consecutive common vertices that, together with the corresponding
subpaths of $\Pi[i,j]$ and $\Pi[i',j']$ between them, surrounds an area whose interior contains
no boundary point of $X$.  From \Cref{obs: subpath}, this contradicts the outermost choice in
the construction of the relevant $\Pi$-path.  Therefore
$\Pi[i,j]\cap \Pi[i',j']$
is exactly the common subpath from $u$ to $v$.

Thus the intersection between any two paths in the constructed family is
either empty or a subpath of both paths.  Together with the vertex-disjointness
for repelling pairs proved above, this shows that the constructed family
$\mathcal P$ is a $D$-good shortest path structure.  This completes
the proof of \Cref{thm:repelling-paths}.
\subsection{Computing the Minimum Realization}
\label{sec: min_realization}

In this section, we complete the proof of Theorem~\ref{thm: minimum main}.

\subsubsection{Medial graphs and Y-Delta transformations}

Let $G$ be a plane graph. We call $G$ the \emph{primal graph}, and its \emph{medial graph} $M(G)$ is obtained by placing a vertex on each edge of $G$ and joining two such vertices whenever the corresponding primal edges occur consecutively around a face. See \Cref{fig: medial} for an example. In a disc-embedded OS instance, the medial graph can equivalently be viewed as a collection of chords in the disc. Intersections of chords correspond exactly to edges of $G$; hence
\[
   |E(G)|=\text{number of chords intersections in }M(G).
\]

The medial graph is determined by two kinds of data.  The first is the pairing of its boundary endpoints; this is the \emph{template}.  The second is the intersection pattern of curves with that boundary pairing.  We formalize the template as follows.

\begin{figure}[h]
\centering
\subfigure
{\scalebox{0.4}{\includegraphics{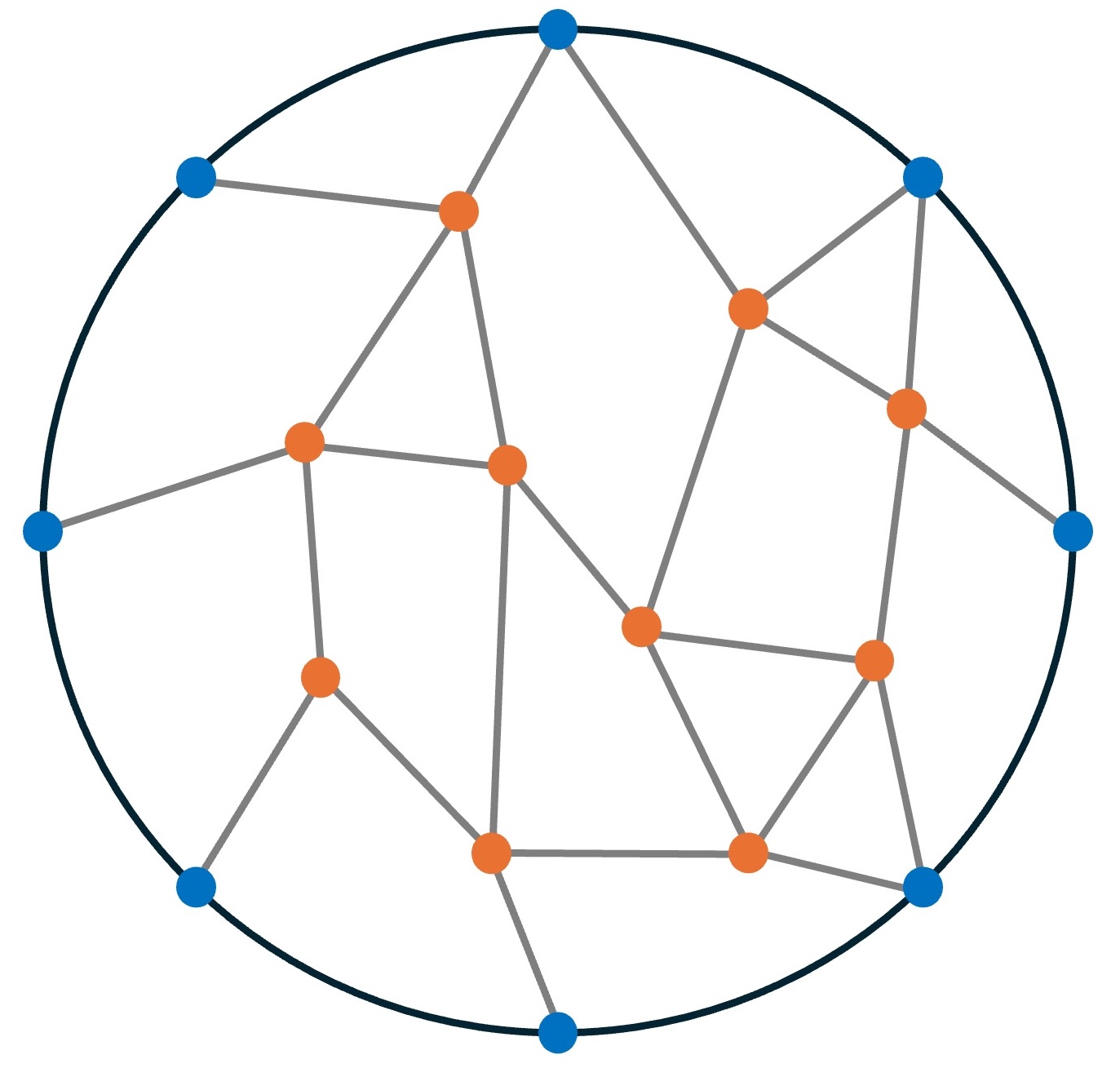}}}
\hspace{2.0cm}
\subfigure
{\scalebox{0.4}{\includegraphics{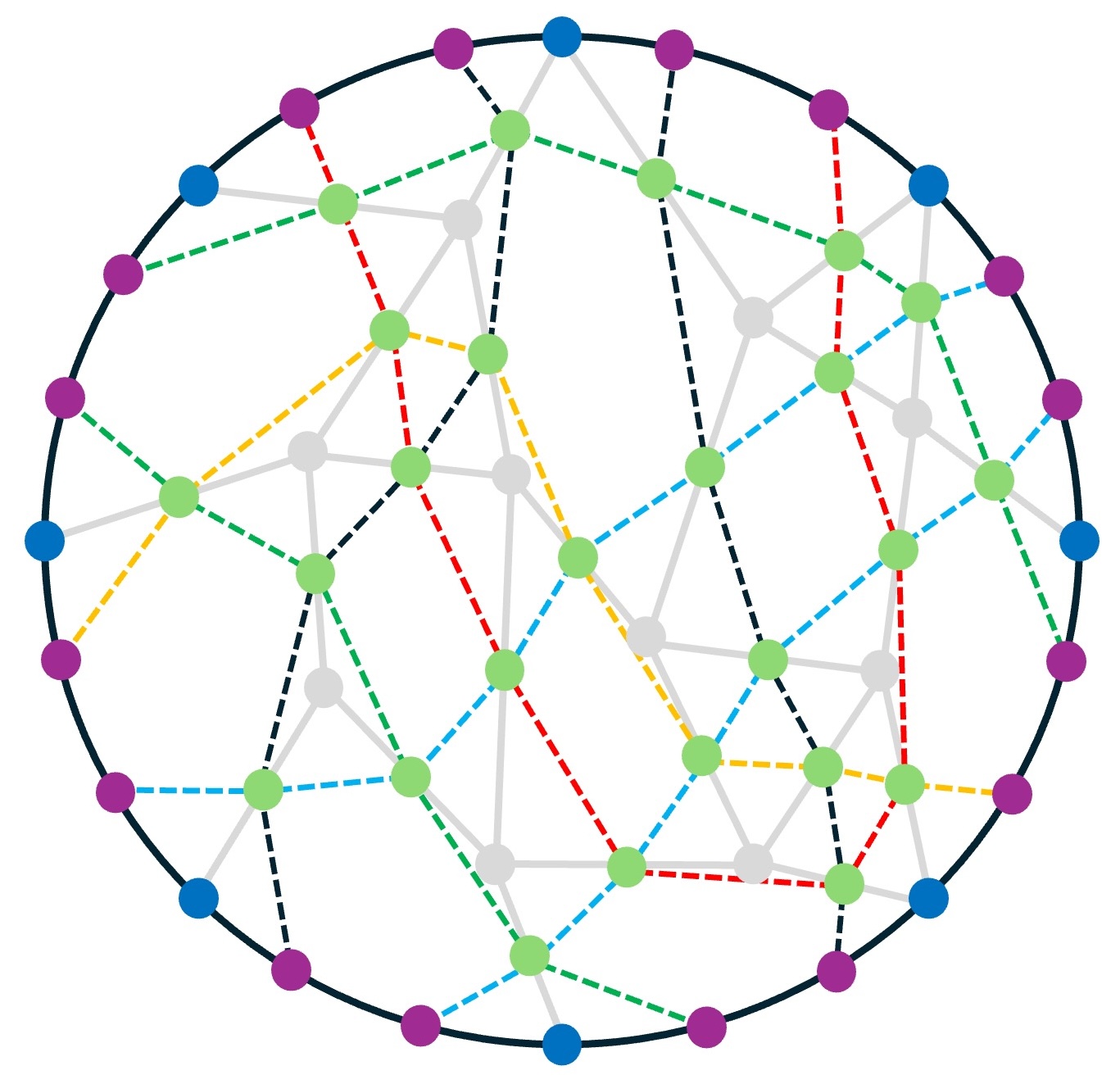}}}
\caption{Left: the primal graph. Terminals are shown in blue (the other reference points in $X$ are not shown). Right: the medial graph. Vertices of $Z$ are shown in purple, and medial chords are shown in dashed lines.}\label{fig: medial}
\end{figure}

Let $X=\{x_1,\ldots,x_m\}$, $m=2|T|$, be the boundary reference points from Section~\ref{sec:prelim}, in clockwise order.  Let
$Z=\{z_1,\ldots,z_m\}$
be another set of boundary points so that for each $i$, $z_i$ lies between $x_{i-1}$ and $x_i$.  Thus points of $X$ and $Z$ alternate around the boundary.  A \emph{template} is a perfect matching $\Phi$ on $Z$, drawn as chords in the disc.  For $x,y\in X$, let $\dist_\Phi(x,y)$ be the number of chords of $\Phi$ crossing the $x$-$y$ boundary cut.  Let $\cro(\Phi)$ be the number of crossing pairs of chords.

An \emph{arrangement} of $\Phi$ is a set of simple curves realizing the matched
pairs, with no self-intersections, no triple crossings, and with any two curves crossing at most once.
In an arrangement, two chords cross if and only if the corresponding chords cross, so every
arrangement of $\Phi$ has exactly $\cro(\Phi)$ intersections.

\paragraph{The Y-Delta transformation.}
The \(Y\)--\(\Delta\) transformation is the local replacement of a
degree-three non-terminal \(v\), adjacent to three vertices
\(a,b,c\), by a triangle on \(a,b,c\), or conversely (which we may call the Delta-Y transformation).  If the three
edges of the \(Y\) have lengths \(\alpha,\beta,\gamma\), then the
distance-preserving triangle has side lengths
\[
    \ell_{ab}=\alpha+\beta,\qquad
    \ell_{bc}=\beta+\gamma,\qquad
    \ell_{ca}=\gamma+\alpha .
\]
Conversely, a triangle with side lengths
\(\ell_{ab},\ell_{bc},\ell_{ca}\) can be replaced by a \(Y\) whenever
the values
\[
    \alpha=\frac{\ell_{ab}+\ell_{ca}-\ell_{bc}}{2},\qquad
    \beta =\frac{\ell_{ab}+\ell_{bc}-\ell_{ca}}{2},\qquad
    \gamma=\frac{\ell_{bc}+\ell_{ca}-\ell_{ab}}{2}
\]
are nonnegative.  In the medial picture, this move has a particularly
concrete interpretation: the three primal edges incident to the center
of the \(Y\), and the three primal edges of the corresponding triangle,
give the same three local medial crossings arranged around a small
triangular cell.  The move changes only whether that local cell is read
as a primal degree-three vertex or as a primal triangular face; the
boundary pairing of the medial chords is unchanged.


\subsubsection{Repelling numbers}

For $x,y\in X$, define
\[
   a_{x,y}=\max\{|\mset|:\mset\text{ is a repelling set and every pair in }\mset\text{ crosses the }x\text{-}y\text{ cut}\}.
\]
By \Cref{thm:repelling-paths}, the minimum $x$-$y$ chain has length at least $a_{x,y}$. The corresponding number of medial chords crossing has a small endpoint correction.  Define
\[
 b_{x,y}= 
 \begin{cases}
  2a_{x,y}-2, & x,y\in T,\\
  2a_{x,y}-1, & |\{x,y\}\cap T|=1,\\
  2a_{x,y},   & x,y\notin T.
 \end{cases}
\]
Indeed, if a chain $A$ has boundary endpoints $x,y$, then a regular neighborhood of $A$ is crossed by $2|A|$, $2|A|-1$, or $2|A|-2$ medial chords according as neither, one, or both endpoints are terminals.  Thus, the template $\Phi$ of any OS instance that realizes $D$ must satisfy that
\begin{equation}\label{eq:template-feasible}
   \dist_\Phi(x,y)\ge b_{x,y}\qquad \text{for all }x,y\in X.
\end{equation}

\begin{observation}\label{obs:neighbor}
Let $x,x'$ be consecutive points of $X$ with $x\in T$, and let $y\in X$.  Then
\[
   a_{x',y}\le a_{x,y}\le a_{x',y}+1.
\]
\end{observation}

\begin{proof}
Every pair crossing the $x'$-$y$ cut also crosses the $x$-$y$ cut.  Conversely, if a pair crosses the $x$-$y$ cut but not the $x'$-$y$ cut, then it must contain the terminal $x$.  In a repelling set, at most one pair can contain $x$: two such pairs fail to repel because their endpoint blocks cannot be separated by a cyclic interval split.  Hence the maximum can increase by at most one.
\end{proof}

\begin{lemma}\label{lem:a-monge}
For consecutive pairs $x,x'$ and $y,y'$ in $X$,
\begin{equation}\label{eq:a-monge}
   a_{x,y}+a_{x',y'}\ge a_{x,y'}+a_{x',y}.
\end{equation}
\end{lemma}

\begin{proof}
We distinguish between the following two cases up to symmetry.

First assume that $x$ and $y'$ are terminals and $x'$ and $y$ are non-terminals. By Observation~\ref{obs:neighbor}, all four quantities in \eqref{eq:a-monge} differ from each other by at most $1$ in each adjacent step. The inequality is immediate unless $a_{x,y'}=a_{x',y}+1$.
Denote $k=a_{x',y}$. Choose a maximum repelling set $\mset_1$ of size $k+1$ crossing the $x$-$y'$ cut and a maximum repelling set $\mset_2$ of size $k$ crossing the $x'$-$y$ cut.  Subject to maximum cardinality, choose them lexicographically closest to the boundary points $x$ and $x'$.

If $\mset_1$ contains no pair using $x$, or if the other endpoint of the pair in $\mset_1$ using $x$ lies on the boundary segment from $x'$ to $y$, then all pairs of $\mset_1$ also cross the $x'$-$y'$ cut, giving $a_{x',y'}\ge k+1$ and hence \eqref{eq:a-monge}.  The same argument applies to $y'$.  Therefore, we may assume that $\mset_1$ has a pair using $x$ and a pair using $y'$ with the other endpoints lying on the segment from $x$ to $y'$.
Order the pairs of $\mset_1$ and $\mset_2$ from left to right, with respect to the boundary ordering cut between $x$ and $x'$.  Write
\[
   \mset_1=\{(u_0,v_0),\ldots,(u_k,v_k)\},\qquad
   \mset_2=\{(s_1,t_1),\ldots,(s_k,t_k)\},
\]
where $v_0=x$, $v_k=y'$, the $u_i$ lie on the segment from $x$ to $y'$, and the $v_i$ lie on the opposite segment from $x'$ to $y$.  The lexicographic choice gives $s_1=x$.

We say that a pair $(u_i,v_i)$ is \emph{irregular} if $v_i$ is farther from $x$ than $t_i$, and call a pair $(s_i,t_i)$ irregular if $s_i$ is farther from $x$ than $u_{i-1}$.  There is a first irregular pair in the alternating list
\[
   (s_1,t_1),(u_1,v_1),(s_2,t_2),\ldots,(u_k,v_k),
\]
for the last pair is irregular because $v_k=y'$.  If the first irregular pair is $(s_i,t_i)$, then the pair $(u_{i-1},t_i)$ cannot repel $(s_{i-1},t_{i-1})$; otherwise we could replace $(s_i,t_i)$ by a closer pair, contradicting the choice of $\mset_2$. \Cref{lem: W-repel} implies that $(s_i,t_i)$ repels $(u_{i-1},t_{i-1})$.  Since no earlier pair is irregular, $t_{i-1}$ is farther from $x$ than $v_{i-1}$, and \Cref{lem: W-repel} implies that $(s_i,t_i)$ repels $(u_{i-1},v_{i-1})$.  Hence
\[
   (u_0,v_0),\ldots,(u_{i-1},v_{i-1}),(s_i,t_i),\ldots,(s_k,t_k)
\]
forms a repelling set of size $k+1$ crossing the $x$-$y$ cut.  Thus $a_{x,y}\ge k+1$, proving \eqref{eq:a-monge}.

If the first irregular pair is $(u_i,v_i)$, the symmetric exchange shows that
\[
   (s_1,t_1),\ldots,(s_i,t_i),(u_i,v_i),\ldots,(u_k,v_k)
\]
is a repelling set of size $k+1$ crossing the $x'$-$y'$ cut.  Thus $a_{x',y'}\ge k+1$, again proving \eqref{eq:a-monge}.

It remains to consider the second parity case, where $x$ and $y$ are terminals and $x'$ and $y'$ are nonterminal points.  If $a_{x,y'}\ne a_{x',y}$, then Observation~\ref{obs:neighbor} immediately implies \eqref{eq:a-monge}.  The only critical case is $a_{x,y'}=a_{x',y}=k$.  The same exchange argument as above applies, with $\mset_1$ a size-$k$ repelling set crossing $x$-$y'$ and $\mset_2$ a size-$k$ repelling set crossing $x'$-$y$.  If the first irregular pair comes from $\mset_2$, the exchange produces a size-$(k+1)$ repelling set crossing $x$-$y$; if it comes from $\mset_1$, it produces a size-$k$ repelling set crossing $x'$-$y'$.  In both subcases \eqref{eq:a-monge} follows, using Observation~\ref{obs:neighbor} for the remaining adjacent term.
\end{proof}

\begin{corollary}\label{cor:b-monge}
For consecutive pairs $x,x'$ and $y,y'$ in $X$,
$b_{x,y}+b_{x',y'}\ge b_{x,y'}+b_{x',y}$.
\end{corollary}

\begin{proof}
The endpoint correction in the definition of $b_{x,y}$ terms is additive in the two endpoints. Therefore, the correction cancels from on both sides, and the corollary follows from \Cref{lem:a-monge}.
\end{proof}

\begin{lemma}\label{lem:unit-diff}
Fix a pair $x,x'$ of consecutive points in $X$.  Among all consecutive pairs $y,y'\in X$ with $x,x',y,y'$ appearing on the boundary in this order,
\begin{itemize}
    \item there is a unique pair for which
$b_{x,y}+b_{x',y'}=b_{x,y'}+b_{x',y}+2$ holds; and
\item for every other such pair, the equality $b_{x,y}+b_{x',y'}=b_{x,y'}+b_{x',y}$ holds.
\end{itemize}
\end{lemma}

\begin{proof}
Since $x,x'$ are consecutive, $a_{x,x'}=1$ and $b_{x,x'}=1$.  For each consecutive pair $y,y'$, define
\[
   \Delta(y)=b_{x,y}+b_{x',y'}-b_{x,y'}-b_{x',y}.
\]
By Corollary~\ref{cor:b-monge}, $\Delta(y)\ge 0$.  Also $\Delta(y)$ is even since the endpoint correction cancels and the remaining expression is twice an integer.  Summing around the boundary telescopes:
\[
\sum_y \Delta(y)
 =\sum_y \bigl((b_{x,y}-b_{x',y})-(b_{x,y'}-b_{x',y'})\bigr)
 =b_{x,x'}+b_{x',x}-b_{x,x}-b_{x',x'}=2.
\]
Thus exactly one summand is $2$ and all others are $0$.
\end{proof}

\subsubsection{Medial template}
\label{sec: medial template}

In this subsection we prove the following lemma.

\begin{lemma}\label{lem:exact-template}
There is a unique template $\Phi$ such that
$\dist_{\Phi}(x,y)=b_{x,y}$ for all $x,y\in X$.
\end{lemma}

\begin{proof}
Let $m=|X|=|Z|$ and use cyclic indices modulo $m$.  
For each pair $i,j$, let $c_{i,j}$ be the indicator that the chord $(z_i,z_j)$ is present.  The vector $c$ is determined by its cut counts. Indeed, if $B_{p,q}=b_{x_p,x_q}$, then the standard inverse for the circular cut-incidence matrix gives
\begin{equation}\label{eq:finite-diff}
   c_{i,j}=\frac12\bigl(B_{i,j}+B_{i-1,j-1}-B_{i-1,j}-B_{i,j-1}\bigr).
\end{equation}

By Lemma~\ref{lem:unit-diff}, for each $i$ there is a unique $j$ for which the right-hand side of \eqref{eq:finite-diff} is $1$, and it is $0$ for every other $j$.  Symmetry gives $c_{i,j}=c_{j,i}$.  Therefore the nonzero entries of $c$ form a perfect matching on $Z$.  The construction is forced by the cut counts, so the matching is unique.
\end{proof}

We denote by $\Phi(D)$ the template given by \Cref{lem:exact-template}.

If two chords $(z_1,z_3)$ and $(z_2,z_4)$ cross, with $z_1,z_2,z_3,z_4$ appearing in this order, then replacing them by $(z_1,z_2)$ and $(z_3,z_4)$, or by $(z_1,z_4)$ and $(z_2,z_3)$, is called an \emph{uncrossing}.  For templates $\Phi,\Phi'$, write $\Phi\to \Phi'$ if $\Phi'$ can be obtained from $\Phi$ by a sequence of uncrossings.  Write $\Phi\succeq \Phi'$ if $\dist_\Phi(x,y)\ge \dist_{\Phi'}(x,y)$ for all $x,y\in X$.
Recall that \Cref{thm:main-uncross} has shown that
If $\Phi\succeq \Phi'$, then $\Phi\to \Phi'$.

\begin{theorem}\label{thm:unique-template}
For every OS metric $D$, $\Phi(D)$ is the unique template with minimum crossing number among all templates satisfying $\dist_\Phi(x,y)\ge b_{x,y}$ for all $x,y\in X$.
That is, every other feasible template $\Phi\ne \Phi(D)$ satisfies that $\cro(\Phi)>\cro(\Phi(D))$.
\end{theorem}

\begin{proof}
If $\Phi$ is feasible, then $\dist_\Phi(x,y)\ge b_{x,y}=\dist_{\Phi(D)}(x,y)$ for all $x,y$, so $\Phi\succeq \Phi(D)$.  By \Cref{thm:main-uncross}, $\Phi\to \Phi(D)$. Since each nontrivial uncrossing strictly decreases the crossing number, $\cro(\Phi)\ge \cro(\Phi(D))$, with equality only if no uncrossing is needed, i.e., only if $\Phi=\Phi(D)$.
\end{proof}

\subsubsection{Minimum realizations}

We are now ready to prove the following theorem on minimum realizations of OS metrics.

\begin{lemma}\label{thm:minimum-realization}
For any OS metric $D$, the graph structures of all minimum realizations of $D$ is exactly the set of primal graphs of all arrangements of template $\Phi(D)$.
The number of edges in these minimum realizations is $\cro(\Phi(D))$.
\end{lemma}

\begin{proof}

Let $(G,T)$ be a minimum OS realization of $D$. Let $\Gamma$ be its
medial graph with medial template $\Phi$.  We first show a basic property of $\Gamma$.
\begin{claim}
If two medial chords
$\alpha$ and $\beta$ in $\Phi$ have non-alternating endpoints on the boundary, then they do not cross each other.    
\end{claim}
\begin{proof}
Suppose towards contradiction that they meet. Since their endpoints are non-alternating, they
meet an even positive number of times.  Hence two consecutive
intersection points of $\alpha$ and $\beta$ bound a lens.  Choose an
innermost such lens, its interior contains no medial vertex. After moving the crossing strands out of its interior (this operation does not change the number of edges in the primal graph) and obtaining an empty medial lens, in the
primal graph, such an empty medial lens can be simplified: either a degree-$2$ nonterminal, or a two-sided
face bounded by two parallel edges.  In the first case we replace the
two incident edges by one whose length is their sum; in the second
case we delete the longer of the two parallel edges.  Both operations
preserve all terminal distances and strictly decrease the number of
edges, contradicting the minimality of $G$.  Thus two medial chords
with non-alternating endpoints do not meet. 
\end{proof}  

The same lens argument
shows that two chords with alternating endpoints meet exactly once.
Consequently, $\Gamma$ is an arrangement of its template $\Phi$.
Now fix a boundary cut $x$-$y$, and let $L$ and $R$ be the two boundary
arcs determined by $x,y$.  A medial chord is called crossing if
it has one endpoint on $L$ and the other endpoint on $R$.  We prove the following claim. 
\begin{claim}
There is a simple arc $\gamma$ from $x$ to $y$ which meets precisely the
crossing chords, and meets each of them exactly once.    
\end{claim}
\begin{proof}
Indeed, if no such arc existed, then the chords whose two endpoints
lie on the same side of the $x$-$y$ cut would contain a topological
barrier between $x$ and $y$.  Equivalently, some connected component of
their union would meet both $L$ and $R$.  But every chord in this union
has both endpoints on $L$ or both endpoints on $R$.  Therefore any
component meeting both sides must contain an intersection between an
$L$-$L$ chord and an $R$-$R$ chord.  These two chords have
non-alternating endpoints, contradicting the reducedness proved above.
Hence such an arc $\gamma$ exists.  Since $\Gamma$ is an arrangement,
each crossing chord meets $\gamma$ exactly once, and no same-side
chord meets $\gamma$.
\end{proof}

Reading the primal and dual cells crossed by $\gamma$ gives an
$x$-$y$ chain $A$.  The primal cells in this sequence are precisely the
vertices of the chain.  Since $\gamma$ meets exactly the medial chords
crossing the $x$-$y$ cut, we obtain
$\operatorname{dist}_{\Phi}(x,y)=
    2|A|-\mathbf 1[x\in T]-\mathbf 1[y\in T]$.
The endpoint correction comes from the fact that a nonterminal boundary
point lies in a peripheral region, whereas a terminal endpoint lies at a
primal vertex.

Since $(G,T)$ realizes $D$, Theorem~\ref{lem: 2-path condition} implies that $G$ admits a
$D$-good shortest path structure.  Therefore Theorem~\ref{thm:repelling-paths} applies.  In
particular, every $x$-$y$ chain has length at least $a_{x,y}$, so the
chain $A$ above satisfies $|A|\ge a_{x,y}$.  Hence
$\operatorname{dist}_{\Phi}(x,y)
    \ge
    2a_{x,y}-\mathbf 1[x\in T]-\mathbf 1[y\in T]
    =
    b_{x,y}$.
Thus the medial template of every minimum realization is feasible.
By Theorem~\ref{thm:unique-template},
$\operatorname{cr}(\Phi)\ge \operatorname{cr}(\Phi(D)).
$
Thus every minimum realization has at least
$\operatorname{cr}(\Phi(D))$ edges.

On the other hand, take any arrangement of template $\Phi(D)$ and form the primal graph by the standard medial-to-primal construction. 
Since $\dist_{\Phi(D)}(x,y)=b_{x,y}$, the minimum chain lengths in this primal graph are exactly the lower bounds $a_{x,y}$, with the endpoint correction.  Hence all chain inequalities of Theorem~\ref{thm:repelling-paths} hold, so the resulting graph admits a $D$-good shortest-path structure, and by Theorem~\ref{lem: 2-path condition} it can be assigned nonnegative edge lengths realizing $D$. The number of edges is $\cro(\Phi(D))$.

The uniqueness of the template follows from Theorem~\ref{thm:unique-template}.  The only freedom in the graph structure is the intersection pattern of chords with this fixed boundary pairing.
\end{proof}

\begin{theorem}
Let $D$ be an OS metric and let $G$ be the primal graph of any arrangement
of template $\Phi(D)$. Then one can efficiently compute
nonnegative edge lengths on $G$ so that the resulting weighted graph
realizes $D$. Consequently, such graphs are the graph structures for minimum realizations of
$D$. The edge lengths need not be unique.
\end{theorem}

\begin{proof}
Let $G$ be the primal graph of an arrangement of $\Phi(D)$. We argue
that $G$ satisfies the hypotheses of Theorem \ref{thm:repelling-paths}. Indeed, fix a boundary cut $x$-$y$, let $A$ be any chain in $G$ with boundary
cut $x$-$y$, and let $M$ be any repelling set whose pairs cross $A$. 
Since every pair in $M$ crosses the boundary cut $x$-$y$, the definition of
$a_{x,y}$ gives $|M|\leq a_{x,y}$. On the other hand, by Lemma \ref{lem:exact-template},
$\operatorname{dist}_{\Phi(D)}(x,y)=b_{x,y}$. Since every medial chord
crossing the $x$-$y$ cut crosses a regular neighborhood of the chain $A$,
the endpoint correction in the definition of $b_{x,y}$ gives
$a_{x,y}\leq |A|$. Thus
$
|M|\leq a_{x,y}\leq |A|,
$
as desired.

Then the proof of Theorem \ref{thm:repelling-paths} constructs a $D$-good shortest path structure in $G$
efficiently. Applying the constructive part of Theorem \ref{lem: 2-path condition}, we obtain
efficiently a nonnegative edge-length function on $G$ whose shortest-path
metric on the terminals is $D$.

The graph $G$ has $|E(G)|=\operatorname{cr}(\Phi(D))$ edges. On the
other hand, every OS realization of $D$ has a feasible medial template,
and by Theorem~\ref{thm:unique-template} no feasible medial template for
$D$ has fewer crossings than $\Phi(D)$. Hence no OS realization of $D$
has fewer edges than $G$, and therefore the realization constructed above
is a minimum realization.

The example in Appendix \ref{apd: example} shows that the edge lengths need not be unique.
\end{proof}


\appendix
\section{Refuting $O(1)$-Point Conditions: Proof of \Cref{thm: no O(1) point}}
\label{apd: no O(1)-point}

In this section, we show that outerplanar metrics do not admit an ``$O(1)$-point condition'', by providing the proof of \Cref{thm: no O(1) point}. Specifically, we will show that, for every positive integer $k\ge 20$, there exists a metric $D$ on $k$ terminals $\set{1,\ldots,k}$, such that
\begin{itemize}
    \item $D$ is not outerplanar (that is, $D$ is not realizable by any outerplanar graph); and
    \item the restriction of $D$ onto any proper subset $I\subsetneq \set{1,\ldots,k}$ is outerplanar.
\end{itemize} 

Consider now $k$ as any fixed integer. Metric $D$ is defined as follows.
\begin{itemize}
    \item For each $1\le i\le k$, $D(i,i+1)=1$ (convention: $k+1=1$).
    \item For each pair $i,j$ with $|i-j|\ne 1$, $D(i,j)=2$.
\end{itemize} 
It is easy to verify that $D$ satisfies all triangle inequalities, and is therefore a well-defined metric.

On the one hand, we show that the restriction of $D$ onto any proper subset of terminals is outerplanar.
Consider the set $I=\set{1,\ldots,i-1,i+1,\ldots, k}$. Denote the restricted metric of $D$ by $D_{I}$. It is easy to verify that the following outerplanar graph realizes $D_{I}$.

\begin{figure}[h]
\centering
\includegraphics[width=0.5\linewidth]{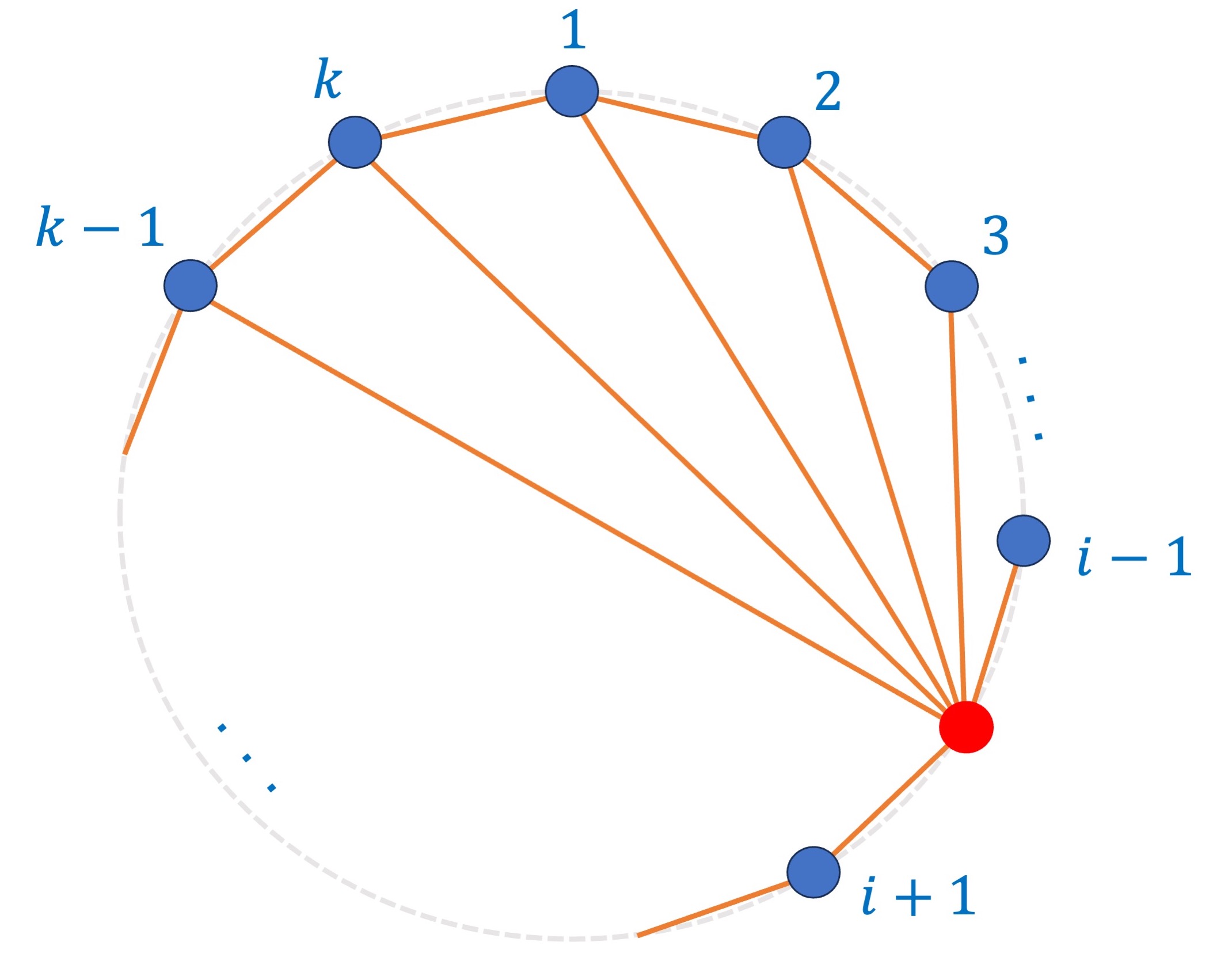}
\caption{The outerplanar graph realizing the metric $D$ restricted onto  $I=\set{1,\ldots,i-1,i+1,\ldots, k}$: the only Steiner node is in red, and every orange edge has length $1$. Every pair $i,i+1$ is connected by a direct orange edge, and every other pair $i,j$ has no direct edge but a length-$2$ path via the red node.}
\end{figure} 

On the other hand, we show that any graph realizing $D$ must contain $K_4$ as a minor. Since outerplanar graphs are $K_4$-free, $D$ cannot be realized by any outerplanar graph.

Let $G$ be a graph with $T=\set{1,\ldots, k}\subseteq V(G)$ realizing metric $D$. For each pair $i,j\in T$, denote by $P_{i,j}$ the shortest path in $G$ connecting $i,j$. We prove the following observation.

\begin{observation}
\label{obs: disjoint}
For distinct indices $i,j$, paths $P_{i,i+1}$ and $P_{j,j+1}$ are internally disjoint.
\end{observation}
\begin{proof}
Since $i,j$ are distinct, by definition of $D$, either $D(i,j+1)=2$ or $D(i+1,j)=2$ holds. Assume without loss of generality that $D(i,j+1)=2$. Assume for contradiction that paths $P_{i,i+1}, P_{j,j+1}$ are not internally disjoint, then we can use a proper subset of their edges (like the proof of \Cref{lem: intersecting shortest paths}) to construct an $i$-to-$(j+1)$ path, whose total length is strictly smaller than $2$, a contradiction.
\end{proof}

By similar arguments, we can prove the following observation.

\begin{observation}
\label{obs: disjoint2}
For $i,j,j'$ where $|i-j|,|i-j'|\ge 2$, paths $P_{i,i+1}$ and $P_{j,j'}$ are internally disjoint.
\end{observation}

As a corollary of \Cref{obs: disjoint}, in graph $G$, paths $P_{1,2},P_{2,3},\ldots,P_{k,1}$ are internally disjoint.
Therefore, by sequentially concatenating them, we obtain a simple cycle $C$.

Consider $4$ terminals $x,y,z,w$ that appear on $C$ in this order, such that every pair of them differs by at least $5$ (this is possible since $k\ge 20$). Consider the $x$-$z$ shortest path $P_{x,z}$ and the $y$-$w$ shortest path $P_{y,w}$. We distinguish between the following cases.

\paragraph{Case 1. $P_{x,z}$ and $P_{y,w}$ are vertex-disjoint.}
In this case, we construct a $K_4$-minor in $G$ as follows. 
\begin{itemize}
    \item The four components are $X=\set{x-1,x,x+1}$, $Y=\set{y-1,y,y+1}$, $Z=\set{z-1,z,z+1}$, and $W=\set{w-1,w,w+1}$; clearly they are disjoint and each induces a connected subgraph of $G$ disjoint from others (e.g., vertices of $X$ are connected by $P_{x-1,x}$ and $P_{x,x+1}$). 
    \item The six connections are:
    \begin{itemize}
        \item the $X$-$Y$ connection is the segment of the cycle $C$ between $x+1$ and $y-1$;
        \item the $Y$-$Z$ connection is the segment of the cycle $C$ between $y+1$ and $z-1$;
        \item the $Z$-$W$ connection is the segment of the cycle $C$ between $z+1$ and $w-1$;
        \item the $W$-$X$ connection is the segment of the cycle $C$ between $w+1$ and $x-1$;
        \item the $X$-$Z$ connection is the path $P_{x,z}$; and
        \item the $Y$-$W$ connection is the path $P_{y,w}$.
    \end{itemize}
    From \Cref{obs: disjoint2} and the assumption that $P_{x,z}$ and $P_{y,w}$ are vertex-disjoint, all six connections are vertex-disjoint, and they are disjoint from all component subgraphs.
\end{itemize}

\paragraph{Case 2. $P_{x,z}$ and $P_{y,w}$ intersect.}

Let $a$ be the middle-point of path $P_{x,z}$ (that is, the unique point on path $P_{x,z}$ that is at distance $1$ from both $x$ and $z$), and let $b$ be the middle-point of path $P_{y,w}$. The only possibility for paths $P_{x,z}$ and $P_{y,w}$ to intersect is when $a$ and $b$ coincide, as otherwise it will lead to either $\dist_G(x,y)<2$ or $\dist_G(z,y)<2$ or $\dist_G(x,w)<2$ or $\dist_G(z,w)<2$, a contradiction. Therefore, we can construct a $K_4$-minor in $G$ as follows. 
\begin{itemize}
    \item The four components are $X=\set{x-1,x,x+1}$, $Y=\set{y-1,y,y+1}$, $W=\set{w-1,w,w+1}$, and $A=\set{a}$. Each of them occupies a connected territory of $G$ disjoint from others.
    \item The six connections are:
    \begin{itemize}
        \item the $X$-$Y$ connection is the segment of the cycle $C$ between $x+1$ and $y-1$;
        \item the $Y$-$W$ connection is the segment of the cycle $C$ between $y+1$ and $z-1$;
        \item the $X$-$W$ connection is the segment of the cycle $C$ between $z+1$ and $w-1$;
        \item the $X$-$A$ connection is the subpath of $P_{x,z}$ between $x$ and $a$;
        \item the $Y$-$A$ connection is the subpath of $P_{y,w}$ between $y$ and $a$; and
        \item the $W$-$A$ connection is the subpath of $P_{y,w}$ between $w$ and $a$.
    \end{itemize}
    From \Cref{obs: disjoint2} and our discussion above, all six connections are vertex-disjoint, and they are disjoint from all component subgraphs.
\end{itemize}

\section{Missing Proofs in \Cref{sec: ordering}}
\label{sec: ordering proofs}
\subsection*{Proof of \Cref{obs:add}}

Since the pairs $(t_1,t_6)$ and $(t_3,t_4)$ do not repel each other, it follows that
\[D(t_1,t_4)+D(t_3,t_6)-D(t_1,t_6)-D(t_3,t_4)=0.\]
On the other hand,
\begin{align*}
    D(t_1,t_4)+D(t_3,t_6)-D(t_1,t_6)-D(t_3,t_4)
     & = \bigl(D(t_1,t_5)+D(t_2,t_6)-D(t_1,t_6)-D(t_2,t_5)\bigr) \\
    & + \bigl(D(t_2,t_5)+D(t_3,t_6)-D(t_3,t_5)-D(t_2,t_6)\bigr) \\
    & + \bigl(D(t_1,t_4)+D(t_2,t_5)-D(t_1,t_5)-D(t_2,t_4)\bigr) \\
    & + \bigl(D(t_2,t_6)+D(t_3,t_5)-D(t_2,t_5)-D(t_3,t_4)\bigr).
\end{align*}
Each of these bracketed terms is non-negative, and their sum is zero; hence every bracketed term must itself be zero. This implies, in particular, that $(t_2,t_5)$ does not repel $(t_3,t_4)$.

\subsection*{Proof of \Cref{obs:trans}}

Since $(t_1,t_6)$ does not repel $(t_2,t_3)$, \Cref{obs:add} implies that $(t_1,t_6)$ does not repel $(t_2,t_5)$, and hence
    \[
    D(t_1,t_5)+D(t_2,t_6)-D(t_1,t_6)-D(t_2,t_5)=0.
    \]
    Similarly, because the pair $(t_2,t_6)$ does not repel $(t_4,t_5)$, we obtain
    \[
    D(t_2,t_5)+D(t_4,t_6)-D(t_2,t_6)-D(t_4,t_5)=0.
    \]
    Adding these two equations yields
    \[
    D(t_1,t_5)+D(t_4,t_6)-D(t_1,t_6)-D(t_4,t_5)=0,
    \]
    which shows that the pair $(t_1,t_6)$ does not repel $(t_4,t_5)$.

\subsection*{Proof of \Cref{obs:sub}}

Since $(t_2,t_4)$ does not repel $(t_6,t_5)$, \Cref{obs:add} implies that $(t_2,t_4)$ does not repel $(t_1,t_5)$, and hence
    \[
    D(t_2,t_5)+D(t_1,t_4)-D(t_2,t_4)-D(t_1,t_5)=0.
    \]
    Similarly, because the pair $(t_1,t_4)$ does not repel $(t_2,t_3)$, we obtain
    \[
    D(t_1,t_3,)+D(t_2,t_4)-D(t_1,t_4)-D(t_2,t_3)=0.
    \]
    Adding these two equations yields
    \[
    D(t_1,t_3)+D(t_2,t_5)-D(t_1,t_5)-D(t_2,t_3)=0
    \]
    which shows that the pair $(t_1,t_5)$ does not repel $(t_2,t_3)$.
\section{An example showing non-uniqueness of edge weights}
\label{apd: example}

In this section, we provide an example OS metric together with its minimum realization, demonstrating the fact that while all graph structures are determined by the medial template computed in \Cref{sec: medial template}, the edge weights are not necessarily unique.

\begin{figure}[h]
\centering
\includegraphics[width=0.4\linewidth]{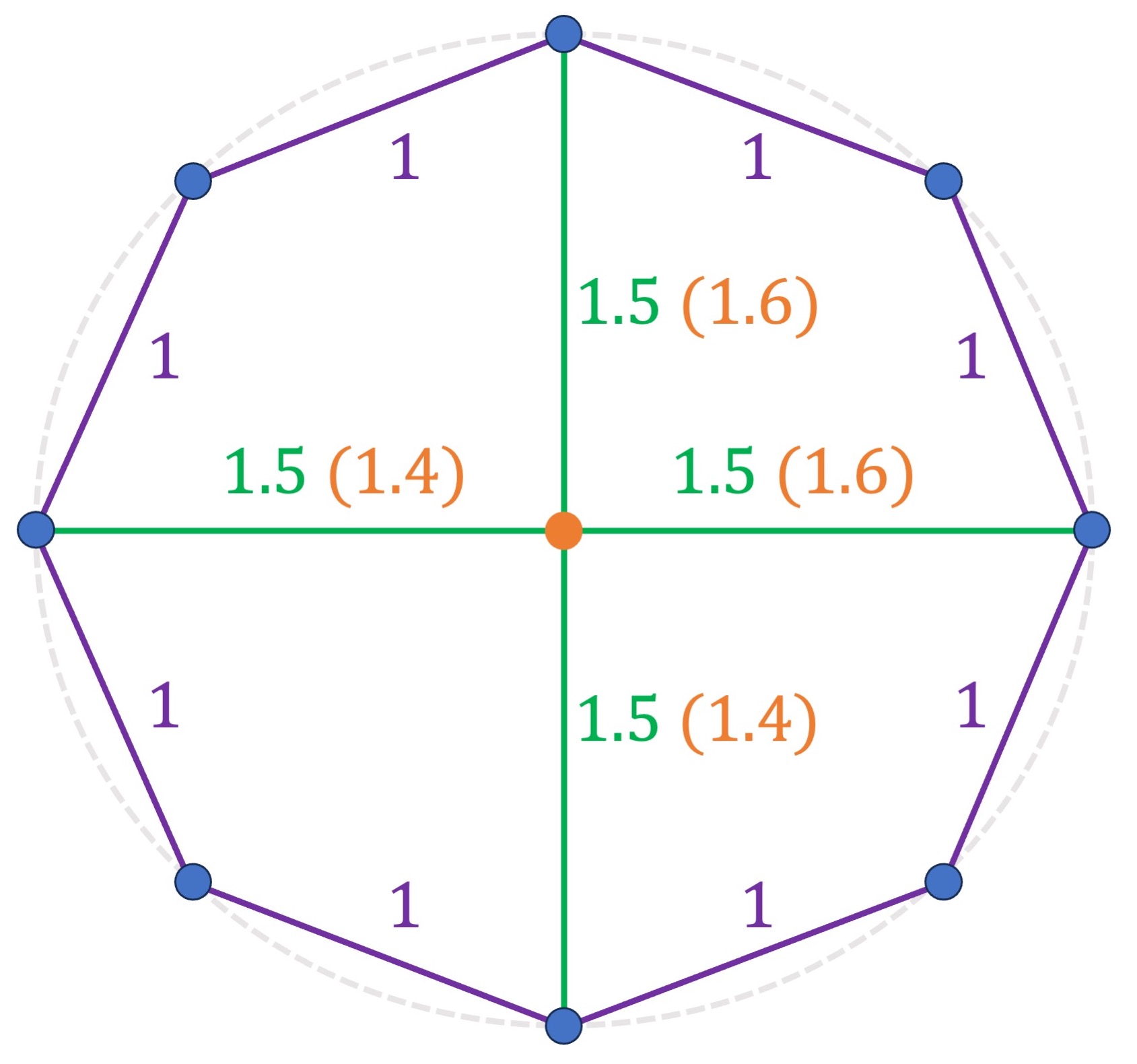}
\caption{Minimum realizations of metric $D$: green edge weights are not unique.\label{fig: example}}
\end{figure} 

Consider a metric $D$ on $8$ points $\set{t_1,\ldots,t_8}$ defined as follows (indices modulo $8$).
\begin{itemize}
\item for each $i$, $D(t_i,t_{i+1})=1$, $D(t_i,t_{i+2})=2$, and $D(t_i,t_{i+3})=3$; and
\item $D(t_1,t_{5})=D(t_3,t_7)=3$, and $D(t_2,t_{6})=D(t_4,t_8)=4$.
\end{itemize}

It is easy to verify that $D$ is an OS metric in that all Kalmanson conditions are satisfied with respect to the natural circular ordering $t_1,\ldots,t_8$ of terminals.
By our approach, the minimum realization of this has graph structure shown in \Cref{fig: example}, while the edge weights are not unique.

\printbibliography

\end{document}